%% file: cqnt-tqe.tex
\documentclass[10pt,journal,twoside]{IEEEtran}

\input{_default_packages.tex}

\usetikzlibrary{shadows, backgrounds, patterns, decorations.pathreplacing}
\usepackage{tikz-cd}

\usepackage{wrapfig}

\usepackage[absolute]{textpos}
\usepackage{xcolor}
\newcommand{\newtext}[1]{{#1}}

\usepackage{tabularx}

\title{Quantum-Classical Coexistence\\ Network Tomography}
\author{\IEEEauthorblockN{
    Xuchuang Wang\IEEEauthorrefmark{1},
    Joseph C. Chapman\IEEEauthorrefmark{2}, 
    Aneesh Ramaswamy\IEEEauthorrefmark{2}, 
     Matheus Guedes de Andrade\IEEEauthorrefmark{1},
     \\
 Yu-Zhen Janice Chen\IEEEauthorrefmark{1}, 
 Joseph M. Lukens\IEEEauthorrefmark{2}\IEEEauthorrefmark{3}, 
 Gayane Vardoyan\IEEEauthorrefmark{1}, 
 Don Towsley\IEEEauthorrefmark{1}}
 \\
    \IEEEauthorblockA{\IEEEauthorrefmark{1}College of Information and Computer Sciences, University of Massachusetts Amherst, Amherst, MA, USA
    \\
        \IEEEauthorrefmark{2}Quantum Information Science Section,
Oak Ridge National Laboratory, Oak Ridge, TN, USA
        \\
        \IEEEauthorrefmark{3}Elmore Family School of Electrical and Computer Engineering and
        \\
        Purdue Quantum Science and Engineering Institute, Purdue University, West Lafayette, IN, USA
        \\
        Emails: xuchuangw@gmail.com, \{mguedesdeand, yuzhenchen, gvardoyan, towsley\}@cs.umass.edu,
        \\
        \{chapmanjc, ramaswamya\}@ornl.gov, jlukens@purdue.edu}
}

\newlength{\doenoticeheight}
\AddToHook{shipout/after}{%
  \ifnum\value{page}=1\relax
    \global\advance\textheight by \doenoticeheight
  \fi}

\begin{document}

\maketitle

\begin{abstract}
    Quantum-classical coexistence networks (QCNs) share optical fiber between quantum and classical signals via wavelength-division multiplexing, offering a practical path to deploying quantum communication over existing telecom infrastructure.
    However, co- and counter-propagating classical traffic introduce distinct depolarization noise on the quantum channel, complicating channel characterization.
    We develop a tomography framework that infers per-link channel parameters of a QCN from end-to-end measurements alone.
    We first model each coexisting fiber by decomposing the quantum signal evolution into photon loss, successful transmission, and three direction-dependent depolarization components.
    Building on this model, we derive closed-form link-level estimators for all channel parameters, and then extend the approach to star-topology networks through a system of multiplicative equations across end-node pairs\newtext{, along} with a simple classical-signal-direction-switching protocol that resolves the remaining unknowns.
    Using these methods, we analyze single-link experimental testbed data and confirm accurate recovery of per-link depolarization probabilities, with estimated process fidelities closely tracking the Bayesian-process-tomography baseline across multiple fiber lengths and wavelengths; the small residual gaps reflect the depolarization-only model approximation. As no multi-link coexistence testbed yet exists, we validate the star-network estimators on \newtext{emulated multi-link paths composed of copies of measured single-link channels}.
    Furthermore, we extend our framework in two directions: (i) we propose a quantum-channel model that \newtext{factorizes} the coexisting fiber into a depolarizing-with-loss signal channel and a Raman-\newtext{noise-}injection channel \newtext{acting on separate optical modes---a completely-positive, trace-preserving \emph{tensor product}---}and show that the link observables reduce exactly to our basic model; and (ii) we generalize the star-network approach to arbitrary topologies via a peeling algorithm (trees) and a least-squares estimator (meshes), which Monte-Carlo simulations confirm recover per-link parameters on tree and cyclic-mesh networks under a correctly-specified model.
\end{abstract}


\begin{textblock}{13.3}(1.4,15)\noindent\fontsize{7}{7}\selectfont\textcolor{black!30}{This manuscript has been co-authored by UT-Battelle, LLC, under contract DE-AC05-00OR22725 with the US Department of Energy (DOE). The US government retains and the publisher, by accepting the article for publication, acknowledges that the US government retains a nonexclusive, paid-up, irrevocable, worldwide license to publish or reproduce the published form of this manuscript, or allow others to do so, for US government purposes. DOE will provide public access to these results of federally sponsored research in accordance with the DOE Public Access Plan (http://energy.gov/downloads/doe-public-access-plan).}\end{textblock}

\input{sections/intro.tex}

\input{sections/model.tex}
\input{sections/cptp-model.tex}

\input{sections/link-estimator.tex}

\input{sections/network-estimator.tex}

\input{sections/general-topology.tex}

\input{sections/experiment.tex}

\input{sections/conclusion.tex}

\section*{Data and Code Availability}
The experimental dataset analyzed in this work was collected by Chapman et al.~\cite{chapman2023coexistent}; access is available from the authors of that work subject to the applicable institutional release procedures. The closed-form link and star-network estimators, the general-topology peeling and least-squares solvers, and the multi-link Monte-Carlo sampler used for the star-network emulation and general-topology simulations will be made available in a public repository (with an archival DOI) upon publication, subject to Oak Ridge National Laboratory and U.S. Department of Energy software-release approval; until then, the code is available from the corresponding author upon reasonable request.

\section*{Acknowledgments}
\textbf{Generative AI disclosure:} During the preparation of this work the author(s) used Gemini Pro to assist with language editing and paragraph rephrasing. After using this tool, the author(s) reviewed and edited the content as needed and take full responsibility for the content of the publication.

This work was performed in part at Oak Ridge National Laboratory, operated by UT-Battelle for the U.S. Department of Energy under contract no.\ DE-AC05-00OR22725. Funding was provided by the U.S. Department of Energy, Office of Science, Office of Advanced Scientific Computing Research, through the Performance Integrated Quantum Scalable Internet program (ERKJ432).

\bibliographystyle{IEEEtran}
\bibliography{bibliography}





\newpage

\end{document}

%% file: _default_packages.tex
\usepackage{cite}

\usepackage{graphicx}
\usepackage{diagbox}
\usepackage{multirow}
\usepackage{booktabs}
\usepackage{subcaption}
\usepackage{threeparttable, array, float}
\usepackage{colortbl}
\usepackage{tkz-graph}
\usepackage{tikz}
\usepackage{pgfplots}

\usepackage{textcomp}
\usepackage{enumitem}
\setlist[itemize]{leftmargin=0pt, itemindent=1em, topsep=2pt, itemsep=1pt, parsep=0pt, partopsep=0pt}
\setlist[enumerate]{leftmargin=0pt, itemindent=1.5em, topsep=2pt, itemsep=1pt, parsep=0pt, partopsep=0pt}

\usepackage{xcolor}

\usepackage{hyperref}
\definecolor{mydarkblue}{rgb}{0,0.08,0.45}
\hypersetup{ %
    pdftitle={},
    pdfsubject={},
    pdfkeywords={},
    pdfborder=0 0 0,
    pdfpagemode=UseNone,
    colorlinks=true,
    linkcolor=mydarkblue,
    citecolor=mydarkblue,
    filecolor=mydarkblue,
    urlcolor=mydarkblue,
}

\usepackage{amsmath,amsfonts,amssymb}
\usepackage{bm, bbm}
\usepackage{mathtools}
\usepackage{nicefrac}
\usepackage{physics}
\usepackage{empheq}

\makeatletter
\newcommand\addstarred[1]{%
    \expandafter\let\csname\string#1@nostar\endcsname#1%
    \edef#1{\noexpand\@ifstar\expandafter\noexpand\csname\string#1@star\endcsname\expandafter\noexpand\csname\string#1@nostar\endcsname}%
    \expandafter\newcommand\csname\string#1@star\endcsname%
}
\makeatother
\newcommand{\1}[1]{\mathbbm{1}{\{#1\}}}
\addstarred\1[1]{\mathbbm{1}{\left\{#1\right\}}}

\renewcommand{\ge}{\geqslant}
\renewcommand{\le}{\leqslant}
\renewcommand{\geq}{\geqslant}

\newcommand{\type}[1]{Type-\uppercase\expandafter{\romannumeral#1}}

\usepackage{amsthm}
\newtheorem{theorem}{Theorem}
\newtheorem{lemma}[theorem]{Lemma}
\newtheorem{proposition}[theorem]{Proposition}

\newtheorem{remark}[theorem]{Remark}

\newtheorem{assumption}[theorem]{Assumption}

\usepackage{algorithm}
\usepackage[noend]{algpseudocode}

\let\oldComment=\Comment
\renewcommand{\Comment}[1]{\oldComment{\texttt{#1}}}
\algnewcommand{\LeftComment}[1]{\Statex $\triangleright$ \texttt{#1}}
\algnewcommand{\RightComment}[1]{\Statex \leavevmode\hfill$\triangleright$ \texttt{#1}}

\algnewcommand\algorithmicinput{\textbf{Input:}}
\algnewcommand\Input{\item[\algorithmicinput]}%

\algnewcommand\algorithmicoutput{\textbf{Output:}}
\algnewcommand\Output{\item[\algorithmicoutput]}%

\algnewcommand\algorithmicinitial{\textbf{initialize:}}
\algnewcommand\Initial{\item[\algorithmicinitial]}%

\usepackage{xspace}
\usepackage{comment}

\usepackage{fontawesome5}

\usepackage[colorinlistoftodos,prependcaption,textsize=tiny,textwidth=1cm,color=green]{todonotes}


%% file: sections/intro.tex
\section{Introduction}

Quantum networks~\cite{azuma2023quantum} are essential infrastructure for emerging quantum technologies, including quantum key distribution (QKD)~\cite{bennett2014quantum}, distributed quantum computing~\cite{jiang2007distributed,cacciapuoti2019quantum}, and quantum sensing~\cite{guo2020distributed}.
However, deploying dedicated optical fibers exclusively for quantum communication is prohibitively expensive and impractical at scale.
A promising alternative is the \emph{quantum-classical coexistence network} (QCN), in which quantum and classical signals share the same optical fiber (called a \emph{coexisting fiber})---typically separated via wavelength-division multiplexing (WDM)~\cite{townsend1997simultaneous}.
By leveraging existing fiber deployments, QCNs offer a cost-effective path toward large-scale quantum networking.
At the same time, classical signals generate noise---primarily through spontaneous Raman scattering---that degrades the fidelity of the co-existing quantum states, making the characterization of these impairments a prerequisite for the reliable design and operation of practical QCNs.

Network tomography~\cite{he2021network} provides a principled approach to this characterization problem: it infers internal link-level parameters from end-to-end measurements alone, without requiring direct access to every individual link.
In classical networks, tomography techniques have been widely studied for estimating link delays, loss rates, and available bandwidth from path-level observations.
For QCNs, network tomography is particularly compelling because, as the network scales, independently characterizing each fiber link becomes infeasible, yet accurate knowledge of per-link parameters---such as photon loss and depolarization induced by classical traffic---is critical for tasks such as entanglement routing~\cite{pant2019routing,wang2025learn}, error budgeting~\cite{azuma2025networking}, and resource allocation~\cite{vardoyan2023quantum}.
Developing tomography methods tailored to the physics of quantum-classical coexistence channels can therefore enable efficient, scalable characterization of large QCNs using only measurements at network edge nodes.

Achieving this goal, however, poses several key challenges.
First, no established mathematical model exists for the \emph{coexisting fiber} channel that captures the quantum signal degradation caused by classical traffic; such a model is a necessary foundation before any tomography inference can be formulated.
Second, a tractable network-level model must also be developed to relate end-to-end observations to internal link parameters across multiple nodes and links.
Third, the coexistence setting introduces a richer parameter space than dedicated quantum or classical networks alone: each link must be characterized not only by its intrinsic transmission loss but also by \emph{direction-dependent} depolarization parameters arising from co- and counter-propagating classical signals, substantially increasing the number of unknowns and making the tomographic inverse problem more challenging.

\vspace{-5pt}
\subsection{Contributions}

The main contributions of this work are as follows.
\begin{itemize}
    \item \textbf{Coexisting-fiber channel model.} We formulate the first channel model for a coexisting fiber, decomposing the quantum signal evolution into photon loss, successful transmission, and \emph{direction-dependent} depolarization with three parameters for the no-classical-signal, co-propagation, and counter-propagation scenarios---capturing the noise structure unique to QCNs (\S\ref{sec:model}).
    \item \textbf{Closed-form link tomography.} We derive a quantum link tomography (QLT) method that estimates all parameters of a single coexisting fiber---the success ratio, the photon-loss ratio, and the three depolarization ratios---in closed form from \newtext{end-to-end} measurements, and we characterize its bias and variance (\S\ref{sec:link-tomography}).
    \item \textbf{Quantum-channel (CPTP) formulation.} \newtext{We recast the coexisting fiber as a CPTP channel on two optical modes: the \emph{tensor product} of a depolarizing-with-loss \emph{signal} channel on the signal mode and a \emph{Raman-noise-injection} channel on the broadband mode into which the classical signal scatters. The excess photons that coexistence delivers are then carried by a two-mode photon-number observable rather than by the trace of the state, and the link observables reduce \emph{exactly} to the basic model (\S\ref{sec:cptp-model}).}
    \item \textbf{Star-network tomography with provable identifiability.} We extend QLT to star-topology QCNs, recovering per-link success ratios from a system of multiplicative equations across end-node pairs and resolving the remaining per-link depolarization parameters via a classical-signal-direction-switching protocol. We prove that a single such flip removes the residual gauge degeneracy (\S\ref{sec:network-tomography}).
    \item \textbf{General topologies.} We prove a path-observable composition rule (Proposition~\ref{prop:path-observables}) showing that the coexistence observables compose multiplicatively along a path \newtext{even though Raman injection makes the received ratio a per-photon \emph{yield} rather than a survival probability}. \newtext{Building on it, we give a tree-peeling estimator that adapts the progressive-etching peeling of~\cite{wang2025quantum}, developed there for bit-/phase-flip channels with SPAM errors. For general graphs, we give a log-linear least-squares solver together with an incidence-rank identifiability condition (\S\ref{sec:general-topology}).}
    \item \textbf{Validation by experiment, emulation, and simulation.} We validate QLT on coexisting-fiber experimental testbed data, confirming model correctness and estimator accuracy against measured depolarization probabilities and process fidelities (matching the Bayesian-process-tomography baseline) across multiple fiber lengths and wavelengths (\S\ref{sec:experiment}). As no multi-link coexistence testbed yet exists, we emulate star networks by recombining the measured link data into multi-link paths (\S\ref{subsec:star-experiment}), and validate the general-topology estimators with fully synthetic Monte-Carlo simulations, together verifying estimator correctness, identifiability, and sample-complexity scaling (\S\ref{subsec:general-topology-experiment}).
\end{itemize}

\subsection{Related Work}

\textbf{Quantum-classical coexistence networks.}
Townsend~\cite{townsend1997simultaneous} first demonstrated the feasibility of transmitting quantum information---specifically for QKD---over a classical fiber network by applying WDM to separate quantum and classical channels.
Subsequent work has focused on improving the capacity and reach of QCNs; an early systematic study of this regime is Patel et al.~\cite{patel2012coexistence}, who quantified the direction-dependent Raman noise that limits high-bit-rate QKD-plus-data coexistence.
For example, Wang et al.~\cite{wang2017long} demonstrated long-distance co-propagation of QKD alongside terabit-rate classical data streams over distances up to \(80\)~km, and Chapman et al.~\cite{chapman2023two} distributed two-mode squeezed light via a coexistence network, a key resource for distributed continuous-variable quantum systems.
Beyond WDM, alternative separation techniques have been proposed, including the electro-optic serrodyne approach of R\"{u}beling et al.~\cite{rubeling2024quantum} and the time-interleaved co-propagation method of Wang et al.~\cite{wang2024time}; reconfigurable deployments have also been demonstrated, e.g.\ the dynamically switched coexistence network field trial of Wang et al.~\cite{wang2023field}.
More recently, Thomas et al.~\cite{thomas2024quantum} demonstrated quantum state teleportation in a coexistence fiber network carrying real internet traffic, and Antesberger et al.~\cite{antesberger2024distribution} proposed hollow-core fibers (HCFs) to further suppress classical-channel noise by confining most \newtext{of the} light in air rather than glass, thereby significantly reducing Raman scattering.

\textbf{Network tomography.}
Network tomography has a long history in classical networks.
Coates et al.~\cite{coates2002internet} \newtext{provided} a comprehensive survey of internet tomography, covering both delay and loss inference from end-to-end measurements.
For loss tomography, C\'{a}ceres et al.~\cite{caceres1999multicast} \newtext{developed} maximum likelihood estimators from multicast probes, while for delay tomography, Lo~Presti et al.~\cite{lopresti2002multicast} and Tsang et al.~\cite{tsang2003network} \newtext{proposed} inference methods for internal delay distributions.
A systematic treatment of identifiability, measurement design, and network state inference \newtext{was given} by He et al.~\cite{he2021network}.
More recently, network tomography has been extended to quantum networks:
De~Andrade et al.~\cite{deandrade2022quantum} \newtext{introduced} quantum network tomography (QNT) using multi-party state distribution, De~Andrade et al.~\cite{deandrade2023characterization} \newtext{further studied} the characterization of quantum flip-star networks, and De~Andrade et al.~\cite{deandrade2024quantum} \newtext{provided} a broader overview of QNT as a principled framework for inferring quantum channel parameters from end-to-end measurements.
Building on these foundations, Wang et al.~\cite{wang2025optimal} \newtext{studied} optimal online probe allocation for both classical and quantum network tomography, formulating the task as an online experimental design problem.
However, these QNT works \newtext{modeled} each link by a single scalar noise parameter under a fixed, static probing configuration---for instance the quantum flip-star of De~Andrade et al.~\cite{deandrade2023characterization}, where every link \newtext{was} an independent bit-flip channel with one flip probability. The coexistence setting we study is structurally different in three ways. First, each coexisting fiber carries three coupled unknowns---the success ratio \(s\) and the co- and counter-propagation depolarization ratios \(d^{(1)}, d^{(2)}\)---that compose \emph{multiplicatively} across a path rather than as a single additive flip probability. Second, this richer per-link structure leaves a residual gauge degeneracy that no static end-to-end probing can resolve \newtext{(Lemma~\ref{lem:star-identifiability} in \S\ref{sec:network-tomography})}. Third, we resolve it with an \emph{active} reconfiguration---reversing the classical-signal direction on a single link, which has no analogue in the passive flip-star---rather than by adding more passive probes. To the best of our knowledge, network tomography for quantum-classical coexistence networks, with this direction-dependent depolarization structure, has not been studied.

%% file: sections/model.tex
\section{Coexisting Fiber Model}
\label{sec:model}

This section formulates the coexisting fiber channel model illustrated in Figure~\ref{fig:qcnt-model}.
We consider a single optical fiber link shared by a quantum channel and a classical channel via WDM.
The quantum signal occupies one wavelength while the classical signal occupies another; at the receiver, a WDM demultiplexer and narrowband filter separate the two.
Despite this spectral separation, noise from the classical channel---primarily spontaneous Raman scattering---leaks into the quantum channel and degrades the transmitted quantum state.
Below, we describe each component of the model and then define the observable quantities and parameters to be estimated.

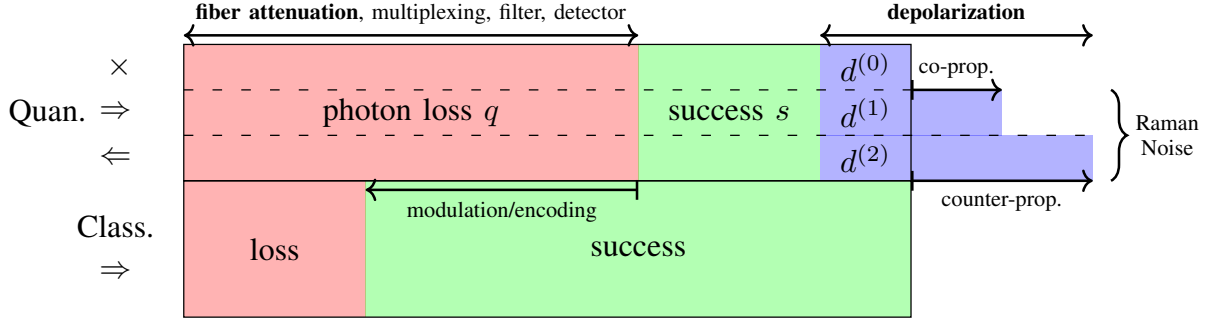
\begin{figure*}[tb]
    \centering
    \resizebox{0.9\textwidth}{!}{
        \begin{tikzpicture}
            \fill[red, fill opacity=0.3] (0, -1.5) rectangle (2, 0); 
            \fill[green, fill opacity=0.3] (2, -1.5) rectangle (8, 0); 

            \fill[red, fill opacity=0.3] (0, 0) rectangle (5, 1.5);   
            \fill[green, fill opacity=0.3] (5, 0) rectangle (7, 1.5);   

            \fill[blue, fill opacity=0.3] (7, 1) rectangle (8, 1.5);   
            \fill[blue, fill opacity=0.3] (7, 0.5) rectangle (9, 1);   
            \fill[blue, fill opacity=0.3] (7, 0) rectangle (10, 0.5);   

            \draw[draw=black] (0, 0) rectangle (8, 1.5);
            \draw[draw=black] (0, -1.5) rectangle (8, 0);

            \draw[draw=black, loosely dashed] (0, 0.5) -- (10, 0.5);
            \draw[draw=black, loosely dashed] (0, 1) -- (9, 1);

            \node at (2.5, 0.75) {{photon} loss \(q\)};
            \node at (7.5, 0.25) {\(d^{(2)}\)};
            \node at (7.5, 0.75) {\(d^{(1)}\)};
            \node at (7.5, 1.25) {\(d^{(0)}\)};
            \node at (6, 0.75) {success \(s\)};

            \node at (1, -0.75) {loss};
            \node at (5, -0.75) {success};

            \node[left=0.2cm] at (-0.75, 0.75) {Quan.};
            \node[left=0.2cm] at (-0.25, 1.25) {\(\times\)};
            \node[left=0.2cm] at (-0.25, 0.75) {\(\Rightarrow\)};
            \node[left=0.2cm] at (-0.25, 0.25) {\(\Leftarrow\)};
            \node[left=0.2cm] at (0, -0.5) {Class.};
            \node[left=0.2cm] at (-0.25, -1) {\(\Rightarrow\)};

            \draw[<->, thick]   (0,1.6)  -- (5,1.6)  node[midway, above]{\scriptsize \textbf{fiber attenuation}, multiplexing, filter, detector};
            \draw[<->, thick]   (7,1.6)  -- (10,1.6)  node[midway, above]{\scriptsize \textbf{depolarization}};
            \draw[<-|, thick]   (2,-0.1)  -- (5,-0.1)  node[midway, below]{\scriptsize modulation/encoding};

            \draw[|->, thick]   (8,1)  -- (9,1)  node[midway, above]{\scriptsize co-prop.};
            \draw[|->, thick]   (8,0)  -- (10,0)  node[midway, below]{\scriptsize counter-prop.};
            \draw[decorate, decoration={brace, amplitude=5pt, mirror}, thick] (10.2, 0) -- (10.2, 1) node[midway, right=4pt]{\scriptsize \shortstack{Raman\\Noise}};

        \end{tikzpicture}
    }
    \caption{Model of quantum-classical coexistence fiber. The upper block represents the quantum channel with three possible outcomes: photon loss (ratio \(q\)), successful transmission (ratio \(s\)), and depolarization (ratio \(d^{(x)}\)). The three rows (separated with dashed lines) in the upper block correspond to no classical signal (\(\times\), depolarization \(d^{(0)}\)), co-propagation (\(\Rightarrow\), depolarization \(d^{(1)}\)), and counter-propagation (\(\Leftarrow\), depolarization \(d^{(2)}\)), respectively. The lower block represents the classical channel.}
    \label{fig:qcnt-model}
\end{figure*}

\subsection{Physical Background}

When a high-power classical signal propagates through an optical fiber, it generates broadband spontaneous Raman scattering (SpRS) noise that extends over a wide spectral range, overlapping with the quantum channel~\cite{1253508,Peters_2009}.
The SpRS noise photons are indistinguishable from the quantum signal photons \newtext{at} the detector, effectively injecting randomly polarized photons into the quantum channel and \newtext{decohering} the quantum state.
\newtext{Most of the scattered light nevertheless occupies broadband optical modes other than the signal mode, whose ladder operators commute with those of the signal mode to good approximation; \S\ref{sec:cptp-model} uses this to model the Raman photons as an ancilla system, with their indistinguishability from signal photons carried by the detector's measurement operator rather than by the state space~\cite{drummond2001quantum}.}
Crucially, the level of Raman noise depends on the \emph{relative propagation direction} of the classical signal with respect to the quantum signal:
\begin{itemize}
    \item \textbf{Co-propagation} (\(\Rightarrow\)): The classical signal travels in the same direction as the quantum signal. The Raman noise accumulates along the fiber and is partially attenuated before reaching the receiver, resulting in a moderate noise-photon injection rate.
    \item \textbf{Counter-propagation} (\(\Leftarrow\)): The classical signal travels in the opposite direction. The Raman noise generated near the (quantum) receiver end experiences less attenuation, typically producing a higher noise-photon injection rate at longer fiber lengths than co-propagation~\cite{patel2012coexistence}.
    \item \textbf{No classical signal} (\(\times\)): In the absence of classical traffic, only intrinsic fiber and system imperfections convert some photons to depolarized ones, yielding the lowest noise baseline.
\end{itemize}

In addition to noise-photon injection, the fiber also attenuates the quantum signal through standard mechanisms---fiber absorption, scattering, and imperfect coupling at the WDM multiplexer, filter, and detector---which collectively \newtext{represent} \emph{photon loss}.

\subsection{Channel Model}

Based on the physical effects described above, we model each coexisting fiber as a channel with three outcomes for each input photon, as shown in Figure~\ref{fig:qcnt-model}:
\begin{enumerate}
    \item \textbf{Photon loss} (ratio \(q\)): The photon is absorbed or scattered and never reaches the receiver. This accounts for fiber attenuation, coupling losses, and imperfect filtering.
    \item \textbf{Successful transmission} (ratio \(s\)): The photon arrives at the receiver with its polarization state intact, i.e., the quantum information is preserved.
    \item \textbf{Depolarization} (ratio \(d^{(x)}\)):
          Among the input photons, a fraction \(d^{(0)}\) is depolarized by mechanisms that act on the transmitted photon itself and survive the deterministic drift correction of \S\ref{subsec:dataset}---residual fast polarization-mode dispersion, imperfect state preparation and measurement (basis misalignment and finite extinction), and the residual frame error left after drift correction; input-independent detector dark and accidental counts instead act like a small, classical-signal-independent injection and are negligible relative to the signal here.
          When a classical signal coexists, Raman scattering \emph{injects} additional depolarized photons, contributing an extra fraction \(d^{(1)} - d^{(0)}\) under co-propagation and \(d^{(2)} - d^{(0)}\) under counter-propagation, both of which are indistinguishable from quantum-signal photons at the receiver.
          The total depolarization ratio \(d^{(x)}\) therefore depends on the classical-signal configuration \(x\), with \(d^{(0)}\) for no classical signal, \(d^{(1)}\) for co-propagation, and \(d^{(2)}\) for counter-propagation.
\end{enumerate}

The success ratio \(s\) and photon loss ratio \(q\) are intrinsic to the fiber and do not depend on the classical traffic, whereas the depolarization ratio \(d^{(x)}\) varies across the three scenarios, with \(d^{(0)} \le d^{(1)}\) and \(d^{(0)} \le d^{(2)}\) in general.
Since the three outcomes are mutually exclusive and exhaustive, in the no-classical-signal scenario (\(x=0\)) the channel parameters satisfy the normalization condition
\begin{align}\label{eq:channel-normalization}
    q + s + d^{(0)} = 1.
\end{align}
Under co- or counter-propagation, the additional Raman-induced photons inflate the received-photon count beyond the input population, so this normalization no longer holds; instead\newtext{,}
\begin{align}\label{eq:channel-normalization-coex}
    \newtext{q + s + d^{(x)} \geq 1, \quad x = 1, 2.}
\end{align}
\newtext{This population accounting is given a quantum-channel form in \S\ref{sec:cptp-model}, where the coexisting fiber is written as a \emph{product}---specifically a tensor product---of two CPTP maps acting on separate optical modes, rather than as a sum of maps on a single mode. There the surplus of~\eqref{eq:channel-normalization-coex} is carried by a two-mode photon-number observable, so that the state itself remains normalized.}

\subsection{Observable Quantities}

In practice, we do not directly observe the channel parameters \((q, s, d^{(x)})\); instead we measure aggregate photon counts over many transmissions. Let \(T \in \mathbb{N}^+\) denote the total number of input photons (each prepared in a fixed state, e.g., \(\ket{0}\)), \(R \in \mathbb{N}^+\) the number of received photons, and \(M \in \mathbb{N}^+\) the number of received photons whose measurement outcome matches the input state (e.g., outcome \(\ket{0}\) under the basis \(\{\ket{0}, \ket{1}\}\)). Further details on the experimental setup appear in \S\ref{subsec:dataset}.

\begin{figure}[t]
    \centering
    \resizebox{0.5\textwidth}{!}{
        \begin{tikzcd}[row sep=large, column sep=large, every cell/.append style={draw, rectangle, inner sep=5pt}, ampersand replacement=\&]
            \text{ \(q+s+d^{(0)}\)} \arrow[r, "\text{Receive}"] \& \text{\(s+d^{(x)}\)} \arrow[r, "\text{Measure}"] \& \text{\(s + \frac{d^{(x)}}{2}\)}
            \\[-2em]
            \text{Input \(T\)} \arrow[r, "\text{Receive}"] \& \text{Received \(R\)} \arrow[r, "\text{Measure}"] \& \text{Outcome \(M\)}
        \end{tikzcd}
    }
    \caption{Channel parameters (top row) and photon numbers (bottom row) in the coexisting fiber model. Here \(x=0,1,2\) for no classical signal, co-propagation, and counter-propagation depolarization, respectively.
    }
    \label{fig:channel-params}
\end{figure}

As illustrated in Figure~\ref{fig:channel-params}, the measurement process proceeds in two stages:
\begin{itemize}
    \item \textbf{Reception.}
          Of the \(T\) input photons, \((s + d^{(0)})T\) reach the receiver via the channel itself; under co- or counter-propagation, an additional \((d^{(x)} - d^{(0)})T\) Raman-injected depolarized photons are also detected, so the total number of received photons is \(R = (s + d^{(x)})T\).
          Among these \(R\) received photons, a fraction \(s/(s+d^{(x)})\) are signal photons with their polarization intact, while the remaining fraction \(d^{(x)}/(s+d^{(x)})\) are depolarized photons with uniformly random polarization.
    \item \textbf{Measurement.} Each received photon is measured in the predetermined polarization basis. A successfully transmitted photon always yields the matched outcome (i.e., the same polarization as sent), while a depolarized photon yields the matched outcome with probability \(1/2\) (since its polarization is uniformly random). We denote by \(M\) the number of photons that yield the matched outcome.
\end{itemize}
From these two stages, the expected values of the observables relate to the channel parameters via
\begin{align}
    \mathbb{E}\left[\frac{R}{T}\right]  = s + d^{(x)} \eqqcolon r^{(x)}, \quad
    \mathbb{E}\left[\frac{M}{T}\right]  = s + \frac{d^{(x)}}{2}. \label{eq:measure-ratio}
\end{align}
We call \(r^{(x)} \coloneqq s + d^{(x)} = \mathbb{E}[R/T]\) the \emph{received ratio}: the expected number of received photons per input photon. \newtext{It is therefore a per-photon \emph{yield} rather than a probability, is consistent with \(q+s+d^{(x)}\ge 1\) in~\eqref{eq:channel-normalization-coex}, and can exceed the loss-limited baseline \(1-q\) once Raman photons are injected (\(x=1,2\)); for \(x=0\) it equals \(1-q\).} (The lowercase \(r\) echoes the received count \(R\).)
    As we show later, received ratios compose \emph{multiplicatively} across concatenated links (Section~\ref{sec:network-tomography} and its general-topology extension in \S\ref{sec:general-topology}), which is the basis of the network estimators.
These equations form the basis for the tomography estimators developed in \S\ref{sec:link-tomography}.

\subsection{Estimation Objective}

\textbf{Data.} The experiments of Chapman et al.~\cite{chapman2023coexistent} collect \((T, R, M)\) triples for three coexistence scenarios: no classical signal (\(T^{(0)}, R^{(0)}, M^{(0)}\)), co-propagation (\(T^{(1)}, R^{(1)}, M^{(1)}\)), and counter-propagation (\(T^{(2)}, R^{(2)}, M^{(2)}\)).

\textbf{Depolarization channel and probability.}
If we focus on the quantum state of the received photons (post-selection), each coexisting fiber can be modeled as one of three \emph{depolarizing channels}, indexed by \(x\).
It is often convenient to express the depolarization in terms of the \emph{depolarization probability} \(p^{(x)} \in [0,1]\), defined as the conditional probability that a received photon is a depolarized noise photon rather than a successfully transmitted signal photon:
\begin{align}
    p^{(x)} \coloneqq \frac{d^{(x)}}{s + d^{(x)}} = \frac{d^{(x)}}{r^{(x)}}, \quad x = 0,1,2.
\end{align}
A larger \(p^{(x)}\) indicates a noisier channel with greater quantum-state degradation; equivalently, a fraction \(s/r^{(x)} = 1 - p^{(x)}\) of the received photons carry intact polarization.

\textbf{Objective.} Our goal is to estimate the channel parameters \(s, q, d^{(0)}, d^{(1)}, d^{(2)}\) (with \(q = 1 - s - d^{(0)}\), and consequently the depolarization probabilities \(p^{(0)}, p^{(1)}, p^{(2)}\)) of a coexisting fiber from the observed photon counts across the three scenarios.

%% file: sections/cptp-model.tex
\section{Quantum-Channel Model for the Coexisting Fiber}\label{sec:cptp-model}

In this section we recast the coexisting-fiber model of Section~\ref{sec:model} as a quantum channel. \newtext{The fiber acts on a \emph{two-mode} space: the signal mode, together with the broadband mode into which the classical signal scatters Raman photons. On that space we model it as the \emph{tensor product} $\mathcal{E}=\mathcal{E}_1\otimes\mathcal{E}_2$ of a depolarizing-with-loss \emph{signal} channel and a \emph{Raman-noise-injection} channel, and we show that the link observables $R/T$ and $M/T$ reduce \emph{exactly} to the basic-model parameters of \S\ref{sec:model}. The two-mode split is itself an approximation, made precise below. It does not discard the indistinguishability of the Raman and signal photons, which the detector's measurement operator retains.}

\subsection{Channel Model Formulation}

\newtext{\paragraph{Mode structure} Spontaneous Raman scattering does not act exclusively on the signal mode. The classical pump scatters photons into a broad continuum of optical modes, and the receiver's filter admits those that overlap the signal band~\cite{drummond2001quantum,1253508}. We collect the admitted photons into a single ancilla mode $S_2$, held distinct from the signal mode $S_1$. Treating $S_2$ as a system separate from $S_1$ requires the two modes to be independent degrees of freedom: formally, that the creation and annihilation (ladder) operators of one commute with those of the other, which holds exactly when the two mode profiles are orthogonal~\cite{drummond2001quantum}. Here this is an \emph{approximation}, since the scattering continuum is not exactly orthogonal to the signal mode. Their residual overlap $\eta$ sets its accuracy: the ladder operators commute up to $\eta$, and equivalently one may take the modes orthogonal by construction and neglect the $O(\eta)$ fraction of Raman photons scattered into the signal mode itself. The receiver's mode selectivity keeps this small, since the signal occupies one narrow spatiotemporal mode inside a much broader filter passband, so $\eta\ll1$. We accordingly work on the two-mode space $S_1\otimes S_2$, where $S_1$ carries the transmitted signal and the ancilla $S_2$ carries the Raman photons admitted by the filter.}

\newtext{\paragraph{Where indistinguishability enters} Separating the modes is not a claim that the two photon species can be told apart. Raman photons inside the filter passband remain operationally indistinguishable from signal photons: nothing in the experiment resolves which of $S_1$, $S_2$ a detected photon came from. We therefore carry that indistinguishability in the \emph{measurement operator} rather than in the state space. We model the detector by the mode-symmetric two-mode photon-number observable $\hat{N}$ of~\eqref{eq:two-mode-number} below. It carries no mode label and so cannot distinguish the two contributions; $S_1$ and $S_2$ are internal bookkeeping that the measurement erases. Placing indistinguishability at the detector is what makes the tensor-product form~\eqref{eq:cptp-map} available: because the detector sums over the two modes, the channel may factorize across modes it never has to identify.}

\newtext{Each mode is truncated to at most one photon: we write $\mathcal{H}=\operatorname{span}\{\ket{\Omega},\ket{0},\ket{1}\}$, where the vacuum $\ket{\Omega}$ records the absence of a photon and $\{\ket{0},\ket{1}\}$ spans the polarization qubit. Writing $I=\ket{0}\bra{0}+\ket{1}\bra{1}$ for the qubit identity, $I/2$ is the maximally mixed (fully depolarized) qubit. We adopt the standard weak-injection assumption and drop higher-order photon-number states: spontaneous Raman scattering into the filtered mode overlapping the signal is weak enough that multi-photon terms contribute negligibly~\cite{patel2012coexistence,Peters_2009}. At the photon-count level at which $R$ and $M$ are defined in \S\ref{sec:model}, the two link observables below then depend on $\mathcal{E}_2$ only through the probability that it injects a photon into $S_2$.}

\newtext{\paragraph{Signal channel} The channel $\mathcal{E}_1$ acts on the input photon $\rho_{\text{in}}$ in $S_1$: it transmits the photon intact with probability $s$, depolarizes it with the intrinsic probability $d^{(0)}$, and loses it with probability $1-s-d^{(0)}$,}
\begin{equation}
    \mathcal{E}_1(\rho) = s\, \rho + d^{(0)} \frac{I}{2} + (1-s-d^{(0)}) \ket{\Omega}\bra{\Omega}. \label{eq:cptp-signal}
\end{equation}
With $s, d^{(0)} \ge 0$ and the normalization $q + s + d^{(0)} = 1$ of \eqref{eq:channel-normalization}, the weights sum to one, so $\mathcal{E}_1$ is trace-preserving and hence a genuine CPTP channel \newtext{on $S_1$}.

\newtext{\paragraph{Raman-noise-injection channel} The channel $\mathcal{E}_2$ acts on the ancilla $S_2$ and captures the crosstalk of the coexisting classical signal. It is the \emph{depolarizing channel} of strength $d^{(x)}-d^{(0)}$,}
\begin{equation}
    \mathcal{E}_2(\rho) = (d^{(x)}-d^{(0)}) \frac{I}{2} + \bigl(1-(d^{(x)}-d^{(0)})\bigr) \newtext{\rho}. \label{eq:cptp-raman}
\end{equation}
\newtext{Being a convex combination of the completely depolarizing map $\rho\mapsto\operatorname{Tr}[\rho]I/2$ and the identity, $\mathcal{E}_2$ is CPTP whenever $0\le d^{(x)}-d^{(0)}\le1$. Unlike a replacement map, it acts on its argument, so it remains well defined on an arbitrary ancilla state. Here the ancilla is prepared in vacuum, $\rho_{\text{vac}}=\ket{\Omega}\bra{\Omega}$, and}
\begin{equation}
    \mathcal{E}_2(\rho_{\text{vac}}) = (d^{(x)}-d^{(0)}) \frac{I}{2} + \bigl(1-(d^{(x)}-d^{(0)})\bigr) \ket{\Omega}\bra{\Omega}: \label{eq:cptp-raman-vac}
\end{equation}
\newtext{a randomly polarized photon is injected with probability $d^{(x)}-d^{(0)}$, and otherwise the mode stays empty. The map takes vacuum to a one-photon state because the classical pump supplying the scattered photon is treated as an external classical drive rather than as part of the quantum system; photon number therefore need not be conserved on $S_1\otimes S_2$ alone.}

\newtext{\paragraph{The coexisting fiber} Both processes occur on the same fiber and are registered by the same detector. Because that detector, rather than the state space, carries their indistinguishability, the two may be assigned to separate factors, and the link channel is their \emph{tensor product} rather than a sum,}
\begin{equation}
    \newtext{\mathcal{E} \;=\; \mathcal{E}_1\otimes\mathcal{E}_2, \qquad \mathcal{E}(\rho_{\text{in}}\otimes\rho_{\text{vac}}) = \mathcal{E}_1(\rho_{\text{in}})\otimes\mathcal{E}_2(\rho_{\text{vac}}).} \label{eq:cptp-map}
\end{equation}
\newtext{A tensor product of CPTP maps is CPTP, so $\mathcal{E}$ is a quantum channel and $\operatorname{Tr}[\mathcal{E}(\rho_{\text{in}}\otimes\rho_{\text{vac}})]=1$. The state therefore stays normalized, as the measurement postulate requires, and the excess photons that coexistence delivers are carried instead by the photon-number observable, as we show next.}

\subsection{Reduction to the Basic Model}

The basic model of \S\ref{sec:model} estimates $s$ and $d^{(x)}$ from the observables $R/T$ and $M/T$ in~\eqref{eq:measure-ratio}. We now recover those observables directly from the channel $\mathcal{E}$. \newtext{As set out above, the detector cannot tell which mode a photon arrived in, so it counts the \emph{two-mode photon-number observable}}
\begin{equation}
    \newtext{\hat{N} \;=\; \Pi\otimes I_{\mathcal{H}} \;+\; I_{\mathcal{H}}\otimes\Pi, \qquad \Pi = \ket{0}\bra{0}+\ket{1}\bra{1} = I,} \label{eq:two-mode-number}
\end{equation}
\newtext{where $\Pi$ projects onto the one-photon subspace of a mode and is orthogonal to the vacuum, and $I_{\mathcal{H}}$ is the identity on the full single-mode space $\mathcal{H}$. The two differ on the vacuum: $I_{\mathcal{H}}$ retains it, whereas $\Pi=I$ annihilates it, which is what keeps an undetected photon from being counted in~\eqref{eq:cptp-R}. Because $\mathcal{E}$ factorizes, so does the expectation, and the expected number of received photons per input photon is}
\begin{align}
    \mathbb{E}\left[\frac{R}{T}\right] & = \operatorname{Tr}_{12}\bigl[\hat{N}\, \mathcal{E}(\rho_{\text{in}}\otimes\rho_{\text{vac}})\bigr] \nonumber                                                 \\
                                       & \newtext{= \operatorname{Tr}_1\bigl[\Pi\mathcal{E}_1(\rho_{\text{in}})\bigr]\operatorname{Tr}_2\bigl[\mathcal{E}_2(\rho_{\text{vac}})\bigr]} \nonumber        \\
                                       & \newtext{\qquad + \operatorname{Tr}_1\bigl[\mathcal{E}_1(\rho_{\text{in}})\bigr]\operatorname{Tr}_2\bigl[\Pi\mathcal{E}_2(\rho_{\text{vac}})\bigr]} \nonumber \\
                                       & \newtext{= \bigl(s+d^{(0)}\bigr)\cdot 1 \;+\; 1\cdot\bigl(d^{(x)}-d^{(0)}\bigr) \;=\; s + d^{(x)},} \label{eq:cptp-R}
\end{align}
\newtext{where the two middle factors are unity because $\mathcal{E}_1$ and $\mathcal{E}_2$ are trace-preserving. Replacing $\Pi$ by the matched-outcome projector $\ket{\psi_{\text{in}}}\bra{\psi_{\text{in}}}$ in both slots (retaining, in either mode, only the photons whose measured polarization matches the prepared state) gives}
\begin{align}
    \mathbb{E}\left[\frac{M}{T}\right] & \newtext{= \bra{\psi_{\text{in}}}\mathcal{E}_1(\rho_{\text{in}})\ket{\psi_{\text{in}}} + \bra{\psi_{\text{in}}}\mathcal{E}_2(\rho_{\text{vac}})\ket{\psi_{\text{in}}}} \nonumber \\
                                       & \newtext{= \Bigl(s+\frac{d^{(0)}}{2}\Bigr) + \frac{d^{(x)}-d^{(0)}}{2} \;=\; s + \frac{d^{(x)}}{2},} \label{eq:cptp-M}
\end{align}
since a depolarized photon ($I/2$) matches with probability $1/2$ and the vacuum never clicks. These are exactly the link observables of~\eqref{eq:measure-ratio}, with the received ratio $r^{(x)} = s + d^{(x)} = \mathbb{E}[R/T]$. \newtext{The received ratio is thus the expectation of $\hat{N}$, whose spectrum is $\{0,1,2\}$, and not a probability. It may therefore exceed the loss-limited baseline $1-q$, consistent with $q+s+d^{(x)}\ge1$ in~\eqref{eq:channel-normalization-coex}, while the state itself stays normalized, $\operatorname{Tr}[\mathcal{E}(\rho_{\text{in}}\otimes\rho_{\text{vac}})]=1$.} The link estimators of \S\ref{sec:link-tomography} therefore estimate the physical parameters of this channel directly.

Beyond reproducing the basic-model observables, the channel form makes the per-link figure of merit explicit. \newtext{Since $\hat{N}$ has spectrum $\{0,1,2\}$, a reception event need not be a single photon. The signal mode delivers one with probability $1-q=s+d^{(0)}$ and, independently, $\mathcal{E}_2$ injects one with probability $d^{(x)}-d^{(0)}$, so both fire together with probability $(1-q)(d^{(x)}-d^{(0)})$.
    The polarization state of the surviving photon is then}
\begin{equation}
    \newtext{\frac{s\,\rho_{\text{in}} + d^{(x)}\,I/2 \;+\; \overbrace{O\bigl((1-q)(d^{(x)}-d^{(0)})\bigr)}^{\text{coincidences }(N\ge 2)}}{s+d^{(x)}},} \label{eq:cptp-postselect}
\end{equation}
\newtext{where the braced remainder is of the order of that discarded coincidence weight (i.e., receive two or more photons), left over because the three single-click weights are used as if they exhausted the outcome space.
    Setting the higher order term aside, $\mathcal{E}$ conditioned on a received photon acts as} a single-qubit depolarizing channel with depolarization probability $p^{(x)} = d^{(x)}/(s+d^{(x)})$ of \S\ref{sec:model}, whose process (entanglement) fidelity with the identity is $F^{(x)} = 1 - \tfrac34 p^{(x)}$. This is the quantity we compare against Bayesian process tomography in \S\ref{sec:experiment} (Fig.~\ref{fig:link-tomography-process-fidelity-comparison}). The depolarizing fixed point $I/2$ also makes the received ratios $r^{(x)}$ compose multiplicatively across concatenated links in \S\ref{sec:network-tomography}--\S\ref{sec:general-topology}.

\subsection{Physical Interpretation of $d^{(x)}$}\label{subsec:physical-interpretation}

The Raman excess $d^{(x)}-d^{(0)}$ that weights the injection channel $\mathcal{E}_2$ in~\eqref{eq:cptp-raman} grows with the classical-signal power $P$ and the fiber length $L$, and depends on the relative propagation direction $x$. Writing $d^{(x)} = d^{(0)} + d_{\text{Raman}}^{(x)}(P, L)$, a first-order accumulation model gives, up to a direction-dependent prefactor,
\begin{align}
    d_{\text{Raman}}^{(1)} & \propto P \int_0^L e^{-\gamma_Q(L-z)} e^{-\gamma_C z}\, dz, \label{eq:raman-co}   \\
    d_{\text{Raman}}^{(2)} & \propto P \int_0^L e^{-\gamma_Q z} e^{-\gamma_C z}\, dz, \label{eq:raman-counter}
\end{align}
for co-propagation ($x=1$) and counter-propagation ($x=2$), with $\gamma_Q,\gamma_C$ the quantum- and classical-band attenuation coefficients~\cite{patel2012coexistence,1253508,Peters_2009}. \newtext{These integrals reproduce the qualitative power- and length-scaling of the injected noise. Attenuation alone, however, does not account for the magnitude of the co/counter asymmetry. That asymmetry is instead dominated by the direction-dependent spontaneous-Raman scattering and capture coefficients. We therefore treat $d^{(1)}$ and $d^{(2)}$ as empirical per-link, per-direction parameters, estimated directly by the tomography of \S\ref{sec:link-tomography} rather than predicted from first principles.}

%% file: sections/link-estimator.tex
\begin{figure}[t]
    \centering
    \begin{subfigure}{0.45\textwidth}
        \centering
        \resizebox{\textwidth}{!}{
            \begin{tikzpicture}[>=stealth]
                \fill[red] (0, -0.15) rectangle (0.3, 0.15);
                \fill[red] (6.7, -0.15) rectangle (7, 0.15);

                \draw[->, blue, line width=3pt] (0.3, 0) -- (6.7, 0);
                \node[blue, font=\scriptsize, above] at (3.5, 0.05) {\textbf{Classical}};

                \draw[->, green!60!black, line width=1.5pt] (6.7, 0.55) -- (0.3, 0.55);
                \node[green!60!black, font=\scriptsize, above] at (3.5, 0.55) {Quantum ($\Leftarrow$)};

                \draw[->, purple!80!black, line width=1.5pt] (0.3, -0.55) -- (6.7, -0.55);
                \node[purple!80!black, font=\scriptsize, below] at (3.5, -0.55) {Quantum ($\Rightarrow$)};
            \end{tikzpicture}
        }
        \caption{Single link tomography}
        \label{subfig:single-link-tomography}
    \end{subfigure}
    \\
    \begin{subfigure}{0.4\textwidth}
        \centering
        \resizebox{\textwidth}{!}{
            \begin{tikzpicture}[>=stealth]
                \coordinate (C) at (0, 0);        
                \coordinate (N1) at (-3, -2.5);   
                \coordinate (N2) at (0, 3);       
                \coordinate (N3) at (3, -2.5);    

                \fill[black] (C) circle (0.18);

                \fill[red, rotate around={39.8:(-2.625,-2.195)}] (-2.80, -2.37) rectangle (-2.45, -2.02);
                \fill[red] (-0.175, 2.43) rectangle (0.175, 2.78);
                \fill[red, rotate around={-39.8:(2.625,-2.195)}] (2.45, -2.37) rectangle (2.80, -2.02);

                \node[below left, font=\footnotesize\bfseries] at (-2.85, -2.50) {1};
                \node[above, font=\footnotesize\bfseries] at (0, 2.80) {2};
                \node[below right, font=\footnotesize\bfseries] at (2.85, -2.50) {3};

                \draw[->, blue, line width=3pt] (-2.45, -2.04) -- (-0.14, -0.12);
                \draw[->, blue, line width=3pt] (0, 0.18) -- (0, 2.43);
                \draw[->, blue, line width=3pt] (0.14, -0.12) -- (2.45, -2.04);


                \draw[->, green!60!black, line width=1.2pt] (-2.58, -1.89) -- (-0.43, -0.10);  
                \draw[->, green!60!black, line width=1.2pt] (-0.20, 0.30) -- (-0.20, 2.43);    
                \draw[->, cyan!60!black, line width=1.2pt] (-0.38, 2.43) -- (-0.38, 0.30);     
                \draw[->, cyan!60!black, line width=1.2pt] (-0.54, 0.04) -- (-2.69, -1.75);    

                \draw[->, orange, line width=1.2pt] (0.20, 2.43) -- (0.20, 0.30);              
                \draw[->, orange, line width=1.2pt] (0.43, -0.10) -- (2.58, -1.89);            
                \draw[->, red!70!black, line width=1.2pt] (2.69, -1.75) -- (0.54, 0.04);       
                \draw[->, red!70!black, line width=1.2pt] (0.38, 0.30) -- (0.38, 2.43);        

                \draw[->, purple!80!black, line width=1.2pt] (2.21, -2.33) -- (0.06, -0.54);   
                \draw[->, purple!80!black, line width=1.2pt] (-0.06, -0.54) -- (-2.21, -2.33); 
                \draw[->, magenta!70!black, line width=1.2pt] (-2.32, -2.19) -- (-0.17, -0.40);
                \draw[->, magenta!70!black, line width=1.2pt] (0.17, -0.40) -- (2.32, -2.19);  

            \end{tikzpicture}
        }
        \caption{Star network tomography}
        \label{subfig:star-network-tomography}
    \end{subfigure}
    \caption{Link and network tomography illustrations.
        \textcolor{blue}{Blue} thick arrows represent classical signals, and other thinner arrows represent quantum signals.
        In Figure~\ref{subfig:single-link-tomography}, the quantum and classical signals share a single fiber link via WDM.
        In Figure~\ref{subfig:star-network-tomography}, a star-topology network carries two or more classical signals on different wavelengths (from the bottom left node) into the center node, where wavelength switching routes each classical signal to a separate output fiber.
    }
\end{figure}
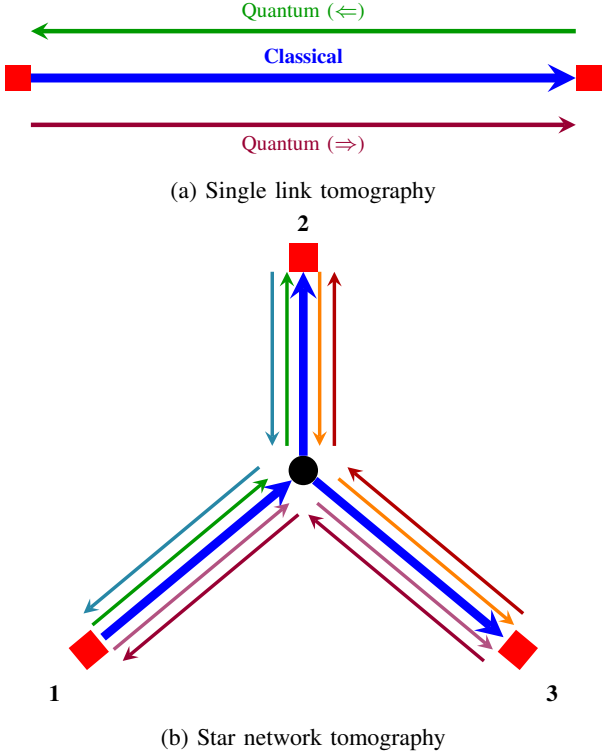

\section{Coexisting Quantum Link Tomography}\label{sec:link-tomography}

In this section, we present the quantum link tomography (QLT) method for estimating the parameters of a single coexisting fiber channel (Figure~\ref{subfig:single-link-tomography}).
Specifically, we estimate the success ratio \(s\) and the depolarization ratios \(d^{(0)}, d^{(1)}, d^{(2)}\) corresponding to the three coexistence scenarios (no classical signal, co-propagation, and counter-propagation).

From the channel model in~\eqref{eq:measure-ratio}, the observable ratios satisfy the system
\begin{equation}
    \label{eq:link-tomography-equations}
    \begin{dcases}
        s + d^{(x)}            = \frac{R^{(x)}}{T^{(x)}}, \quad x = 0,1,2, \\
        s + \frac{d^{(x)}}{2}  = \frac{M^{(x)}}{T^{(x)}}, \quad x = 0,1,2.
    \end{dcases}
\end{equation}
Subtracting the two equations within each scenario and rearranging yields closed-form estimators for all parameters:
\begin{align}
    \label{eq:link-tomography-estimator}
    \begin{cases}
        \hat{d}^{(x)}
         & = \dfrac{2(R^{(x)}-M^{(x)})}{T^{(x)}}, \quad x = 0,1,2,
        \\[0.4em]
        \hat{s}
         & =
        \dfrac{1}{3}\sum_{x=0}^2
        \dfrac{2M^{(x)}-R^{(x)}}{T^{(x)}},
    \end{cases}
\end{align}
where each \(\hat{d}^{(x)}\) follows directly from the difference between the received count \(R^{(x)}\) and the matched-outcome count \(M^{(x)}\), and \(\hat{s}\) combines the three per-scenario success estimates to reduce measurement noise.
\newtext{Explicitly, each scenario supplies its own unbiased estimate \(\hat{s}^{(x)}\coloneqq(2M^{(x)}-R^{(x)})/T^{(x)}\) of the same \(s\), and \(\hat{s}=\sum_{x=0}^{2}w_x\hat{s}^{(x)}\) with \(\sum_x w_x=1\); \eqref{eq:link-tomography-estimator} uses the equal weights \(w_x=1/3\). The three scenarios are independent runs, so by~\eqref{eq:link-variances} below the minimum-variance choice would be \(w_x\propto T^{(x)}/[s(1-s)+d^{(x)}]\), which reduces to \(1/3\) exactly when that ratio is constant in \(x\); we keep equal weights because the optimal ones depend on the parameters being estimated.}
The photon-loss ratio then follows from the normalization in~\eqref{eq:channel-normalization}, \(\hat{q} = 1 - \hat{s} - \hat{d}^{(0)}\), completing the closed-form estimate of all link parameters \((\hat{s},\hat{q},\hat{d}^{(0)},\hat{d}^{(1)},\hat{d}^{(2)})\).
From these estimates, the depolarization probabilities are \(\hat{p}^{(x)} = \hat{d}^{(x)} / (\hat{s} + \hat{d}^{(x)})\) for \(x = 0, 1, 2\).

\begin{remark}
    The estimators in~\eqref{eq:link-tomography-estimator} assume that data from all three scenarios---including the no-classical-signal baseline \((T^{(0)}, R^{(0)}, M^{(0)})\)---are available.
    If only the two coexistence scenarios are measured, the same form applies with
    \(\hat{d}^{(x)} = 2(R^{(x)} - M^{(x)})/T^{(x)}\) for \(x = 1, 2\) and
    \(\hat{s} = \tfrac{1}{2}\sum_{x=1}^2 (2M^{(x)} - R^{(x)})/T^{(x)}\).
    In this two-scenario case the loss ratio \(\hat{q}\) is not identifiable, as it requires the \(x=0\) baseline \(\hat{d}^{(0)}\).
\end{remark}

\subsection{Statistical Properties}\label{subsec:link-statistics}
The estimators in~\eqref{eq:link-tomography-estimator} are not merely algebraic inversions: they are unbiased method-of-moments estimators with closed-form variances. Under the binomial statistics implied by \S\ref{sec:model}---for each scenario \(x\), the intact (signal) photons and the depolarized photons (intrinsic plus Raman-injected) form two independent binomial populations over the \(T\) inputs, \(R_s\sim\mathrm{Binom}(T,s)\) and \(R_d\sim\mathrm{Binom}(T,d^{(x)})\), with \(R=R_s+R_d\) and \(M=R_s+\mathrm{Binom}(R_d,\tfrac12)\) (so \(R\) may exceed \(T\) under Raman injection)---moment matching gives \(\mathbb{E}[\hat{d}^{(x)}]=d^{(x)}\) and \(\mathbb{E}[\hat{s}]=s\) for any fixed \(T\), with
\begin{equation}\label{eq:link-variances}
    \operatorname{Var}(\hat{d}^{(x)}) = \frac{d^{(x)}(2-d^{(x)})}{T^{(x)}},\qquad
    \operatorname{Var}(\hat{s}^{(x)}) = \frac{s(1-s)+d^{(x)}}{T^{(x)}},
\end{equation}
where \(\hat{s}^{(x)}\) is the single-scenario success estimate defined above; averaging the three independent equal-\(T\) scenarios yields \(\operatorname{Var}(\hat{s})=[3s(1-s)+\sum_{x}d^{(x)}]/(9T)\). All estimators are therefore \(\sqrt{T}\)-consistent, and the depolarization probability \(\hat{p}^{(x)}=\hat{d}^{(x)}/(\hat{s}+\hat{d}^{(x)})\) inherits an \(O(1/\sqrt{T})\) standard error by the delta method. These closed forms match Monte-Carlo simulation and provide the \(\pm\)standard-error bands against which the QLT-versus-BPT agreement of \S\ref{sec:experiment} should be read.



%% file: sections/network-estimator.tex
\section{Coexisting Quantum Network Tomography}\label{sec:network-tomography}

We now turn to the network tomography setting~\cite{he2021network}, in which a quantum network has multiple nodes and links and we aim to estimate the parameters of each link from \emph{end-to-end} probes.
In this section, we focus on the network tomography method for a star network; an extension to general multi-node, multi-link topologies is presented in \S\ref{sec:general-topology} below and validated by simulation in \S\ref{subsec:general-topology-experiment}.


We take the star network in Figure~\ref{subfig:star-network-tomography} as a running example.
The network has three links connecting a central hub node to three end nodes; each link \(k\) has its own parameters \(s_k\) and \(d_k^{(x)}\) for \(k \in \{1,2,3\}\) and \(x \in \{0,1,2\}\).
Because every network probe is sent while the classical traffic is present, each hop is always under coexistence, so only \(x\in\{1,2\}\) (co- or counter-propagation) appears in the network estimators below; the no-classical-signal case \(x=0\) enters only in single-link tomography (\S\ref{sec:link-tomography}).
We assume a steady classical signal flow (blue thick arrows) from the bottom-left end node to the central node, where it bifurcates toward the other two end nodes.
To estimate per-link parameters, we perform end-to-end \emph{quantum} measurements between each ordered pair of end nodes (illustrated by the surrounding thinner arrows in Figure~\ref{subfig:star-network-tomography}) and then solve a system of equations relating these measurements to the link parameters.
We describe the method formally below.

Consider a probe sent from end node \(k\), through a coexisting fiber, to the central hub node, which then forwards the received quantum signal through another coexisting fiber to a different end node \(\ell\) (with \(k, \ell \in \{1,2,3\}\) and \(k \neq \ell\)).
Let \(T_{k,\ell}\) denote the total number of input photons, \(R_{k,\ell}\) the number received at node \(\ell\), and \(M_{k,\ell}\) the number that yield matched measurement outcomes.
Since the quantum signal traverses two hops---first link \(k\) (from node \(k\) to the hub) and then link \(\ell\) (from the hub to node \(\ell\))---the coexistence scenario (co- or counter-propagation) on each link depends on the direction of its classical signal relative to the quantum signal.
To make this dependence explicit, we define an indicator \(x(e,i) \in \{1,2\}\), where \(e\) is the link index and \(i\) is the quantum-signal direction: \(i=1\) for end-to-hub and \(i=2\) for hub-to-end; \(x(e,i)=1\) denotes co-propagation and \(x(e,i)=2\) denotes counter-propagation.
For example, \(x(k,1)=1\) means that on link \(k\) the end-to-hub quantum signal is co-propagating with the classical signal, while \(x(k,2)=2\) means that on the same link the hub-to-end quantum signal is counter-propagating.
With this notation, the first hop uses \(x(k,1)\) and the second hop uses \(x(\ell,2)\).

Following the procedure in Figure~\ref{fig:channel-params}, and writing \(r_k^{(x)} \coloneqq s_k + d_k^{(x)}\) for the received ratio of link~\(k\) under scenario~\(x\) (the per-link instance of~\eqref{eq:measure-ratio}), the parameters of the two links involved (those incident to nodes \(k\) and \(\ell\)) satisfy, for \(k, \ell \in \{1,2,3\}\) and \(k \neq \ell\),
\begin{empheq}[left=\empheqlbrace]{align}
    r_k^{(x(k,1))}\, r_\ell^{(x(\ell,2))}
    & \!=\! \frac{R_{k,\ell}}{T_{k,\ell}}
    \label{eq:group2}
    \\
    s_k\!\! \left(\!\! s_\ell {+} \frac{d_\ell^{(x(\ell,2))}}{2} \!\!\right)\! +\! d_k^{(x(k,1))}\frac{\! s_\ell {+} d_\ell^{(x(\ell,2))}}{2}\!
    & \!=\! \frac{M_{k,\ell}}{T_{k,\ell}}
    \label{eq:network-tomography-measurement}
\end{empheq}
The first equation, \eqref{eq:group2}, accounts for all photons reaching node \(\ell\) (those successfully transmitted plus those injected by Raman noise on either link): it is the product of the two per-link received ratios, the two-link instance of a multiplicative law that we generalize to arbitrary paths in Proposition~\ref{prop:path-observables}.
The second equation accounts for the photons that yield matched measurement outcomes:
the term \(s_k(s_\ell + d_\ell^{(x(\ell,2))}/2)\) corresponds to photons successfully transmitted on link \(k\) and then either successfully transmitted on link \(\ell\) (\(s_k s_\ell\), all matched) or depolarized on link \(\ell\) (\(s_k d_\ell^{(x(\ell,2))}/2\), only half matched);
the term \(d_k^{(x(k,1))}(s_\ell + d_\ell^{(x(\ell,2))})/2\) corresponds to depolarized photons emerging from link \(k\) that then traverse link \(\ell\)---all of which are depolarized, so only half yield matched outcomes.
Rearranging, we obtain, for \(k, \ell \in \{1,2,3\}\) and \(k \neq \ell\),
\begin{equation}
    s_k s_\ell  =  \frac{2M_{k,\ell} - R_{k,\ell}}{T_{k,\ell}}. \label{eq:group1}
\end{equation}
The six instances of~\eqref{eq:group1} (one per ordered pair \(k\neq \ell\)) involve only the success ratios \(\{s_k\}\); we can therefore solve for them independently as
\begin{align}
    \label{eq:parameter-s-estimation}
    \hat s_k  \coloneqq \sqrt{\frac{ Q_{k,\ell}  Q_{k,m}}{ Q_{\ell,m}}},
\end{align}
for distinct \(k, \ell, m \in \{1,2,3\}\),
where \(Q_{k,\ell} \coloneqq \frac{(2  M_{k,\ell} -  R_{k,\ell}) + (2 M_{\ell,k} -  R_{\ell,k})}{2  T_{k,\ell}}\) symmetrizes over the two probe directions between nodes \(k\) and \(\ell\) (assuming \(T_{k,\ell} = T_{\ell, k}\) for simplicity).

\newtext{The quantity \(2M_{k,\ell}-R_{k,\ell}\) isolates the photons that survive \emph{polarization-intact} through both hops. By~\eqref{eq:group1}, the intact population, \(s_k s_\ell\), composes multiplicatively---exactly as per-link pass-probabilities do in classical multicast loss tomography~\cite{caceres1999multicast}.} The success-ratio estimator~\eqref{eq:parameter-s-estimation} is therefore the quantum analogue of the classical tree loss estimator, with \(s_e\) playing the role of a link pass-probability, which lets us inherit its identifiability and sample-complexity results. What is \emph{new} to the coexistence setting is the direction-dependent received ratio \(r_e^{(x)}=s_e+d_e^{(x)}\) and the resulting depolarization-gauge degeneracy in~\eqref{eq:group2}, which has no classical loss-tomography counterpart and motivates the single-link classical-flip strategy below.

\newtext{Physically, that degeneracy is the following. Each link contributes two unknowns to~\eqref{eq:group2}---the ratio a probe sees travelling end-to-hub, \(r_e^{(x(e,1))}\), and the one it sees travelling hub-to-end, \(r_e^{(x(e,2))}\)---and every probe traverses exactly one link in each direction. Scaling all three end-to-hub ratios by \(\lambda>0\) and all three hub-to-end ratios by \(1/\lambda\) therefore leaves all six products unchanged, so static probing determines the products \(r_e^{(1)}r_e^{(2)}\) but not the ratios \(r_e^{(1)}/r_e^{(2)}\) that separate co- from counter-propagation. Reversing the classical signal on one link is the \(\mathbb{Z}_2\) operation on \(x(e,i)\) that breaks this symmetry; Lemma~\ref{lem:star-identifiability} makes it precise.}


Although~\eqref{eq:group2} also produces six equations (one per ordered pair), only five are independent, so this group alone does not determine all six unknowns \(\{d_k^{(x)} : k\in\{1,2,3\},\, x\in\{1,2\}\}\). We make this precise, together with its resolution, after stating the assumption that the flip relies on.

\begin{assumption}[Flip invariance]\label{ass:flip-invariance}
Reversing the classical-signal direction on a link changes which propagation mode (co- vs.\ counter-) a given quantum probe encounters, but leaves the two per-link depolarization ratios \(d_e^{(1)},d_e^{(2)}\) unchanged in value. Equivalently, the flip interchanges \(x(e,1)\leftrightarrow x(e,2)\) without introducing new unknowns, so the reversed-direction batch reuses the same six unknowns \(\{d_k^{(x)}\}\).
\end{assumption}

\begin{lemma}[\newtext{Gauge degeneracy and gauge fixing}]\label{lem:star-identifiability}
\newtext{Taking logarithms of the six instances of~\eqref{eq:group2} gives a linear system \(A\mathbf{u}=\mathbf{b}\) in the six log-received-ratios \(u_{e,x}=\log r_e^{(x)}\), with \(A\) the \(6\times6\) probe-by-unknown incidence matrix and \(b_{k,\ell}=\log(R_{k,\ell}/T_{k,\ell})\). Then \(\operatorname{rank}A=5\) and \(\ker A=\operatorname{span}\{v\}\), where \(v\) assigns \(+1\) to each link's end-to-hub received ratio and \(-1\) to its hub-to-end one. Consequently \(\mathbf{u}\)---and hence \(d_e^{(x)}=r_e^{(x)}-s_e\), with \(\{s_e\}\) already fixed by~\eqref{eq:parameter-s-estimation}---is determined only up to the gauge orbit \(\mathbf{u}+tv\). Under Assumption~\ref{ass:flip-invariance}, reversing the classical-signal direction on a single link supplies one further probe equation, i.e.\ a row \(F\) with \(Fv\neq0\); appending it \emph{fixes the gauge}, giving a stacked system of full column rank \(6\) with a unique solution. Flipping any single link suffices. Note that \(F\) is a \emph{measured} row, not an arbitrary constraint, so it selects the physical point of the orbit rather than a conventional one.}
\end{lemma}
\begin{proof}
In the ordering \((r_1^{(1)},r_2^{(1)},r_3^{(1)},r_2^{(2)},r_1^{(2)},r_3^{(2)})\), the incidence matrix of the six ordered pairs in~\eqref{eq:combined-network-tomography-equation-original} has rank \(5\) with null space spanned by \(v=(+1,-1,-1,+1,-1,+1)\) (each row sums two entries of \(v\) with opposite signs, so \(Av=0\)). Stacking any single flipped-link batch adds a row not orthogonal to \(v\), removing the null direction and giving rank \(6\); the unique solution then follows by back-substitution, e.g.\ \(r_1^{(1)}=\sqrt{(R_{1,2}/T_{1,2})(R_{1,3}'/T_{1,3})/(R_{3,2}/T_{3,2})}\) with \(R_{1,3}'\) the flipped-link measurement.
\end{proof}

To \newtext{fix the gauge} we therefore flip the direction of the classical signal on one link---say link~\(3\)---so that the mapping on that link is swapped between the two probe directions (i.e., \(x(3,1)\) and \(x(3,2)\) are interchanged).
This yields a second batch of measurements, from which we obtain a complementary set of equations~\eqref{eq:combined-network-tomography-equation-alternative}.
Putting the original and alternative groups together gives
\begin{align}
    \label{eq:combined-network-tomography-equation-original}
    \begin{cases}
        r_1^{(1)}\, r_2^{(1)}
         & = \frac{R_{1,2}}{T_{1,2}}
        \\
        r_1^{(1)}\, r_3^{(1)}
         & = \frac{R_{1,3}}{T_{1,3}}
        \\
        r_2^{(2)}\, r_1^{(2)}
         & = \frac{R_{2,1}}{T_{2,1}}
        \\
        r_2^{(2)}\, r_3^{(1)}
         & = \frac{R_{2,3}}{T_{2,3}}
        \\
        r_3^{(2)}\, r_1^{(2)}
         & = \frac{R_{3,1}}{T_{3,1}}
        \\
        r_3^{(2)}\, r_2^{(1)}
         & = \frac{R_{3,2}}{T_{3,2}}
    \end{cases},
    \\
    \label{eq:combined-network-tomography-equation-alternative}
    \begin{cases}
        r_1^{(1)}\, r_2^{(1)}
         & = \frac{R_{1,2}}{T_{1,2}}
        \\
        r_1^{(1)}\, \textcolor{red}{r_3^{(2)}}
         & = \frac{R_{1,3}}{T_{1,3}}
        \\
        r_2^{(2)}\, r_1^{(2)}
         & = \frac{R_{2,1}}{T_{2,1}}
        \\
        r_2^{(2)}\, \textcolor{red}{r_3^{(2)}}
         & = \frac{R_{2,3}}{T_{2,3}}
        \\
        \textcolor{red}{r_3^{(1)}}\, r_1^{(2)}
         & = \frac{R_{3,1}}{T_{3,1}}
        \\
        \textcolor{red}{r_3^{(1)}}\, r_2^{(1)}
         & = \frac{R_{3,2}}{T_{3,2}}
    \end{cases}.
\end{align}

Combining~\eqref{eq:parameter-s-estimation},~\eqref{eq:combined-network-tomography-equation-original}, and~\eqref{eq:combined-network-tomography-equation-alternative} yields estimators for all per-link parameters of the star network.

We note that flipping the classical-signal direction on a single link is the asymmetry that breaks this degeneracy. Crucially, this flip is power-preserving under mild and commonly satisfied conditions: the Raman-induced depolarization that a probe accumulates scales with the \emph{time-averaged} classical power rather than with the data payload, since the injected noise photons integrate over a detection window vastly longer than a classical symbol period~\cite{Peters_2009,patel2012coexistence}. Whenever the classical launch power is held at a fixed target---for instance a strong carrier on which the modulation is a small perturbation, a constant-envelope (phase- or polarization-) modulation format, or an amplifier output clamped by automatic power control---reversing the propagation direction on a link leaves $d^{(1)}$ and $d^{(2)}$ unchanged in magnitude and merely interchanges which one a given probe sees. The single-link flip therefore reuses exactly the same per-link unknowns. Such a reversal is practical in several deployments---bidirectional links with matched transceivers, a dedicated fixed-power pilot tone, or software-defined/ROADM networks where the control plane reroutes an individual channel per link---and in a star it is a \emph{local} operation between the hub and one end node, performed during a scheduled maintenance window rather than a network-wide reconfiguration.

Assumption~\ref{ass:flip-invariance} can moreover be relaxed. Suppose the reversed batch perturbs the flipped link's depolarization to \(d_3^{(x)}\to d_3^{(x)}(1+\varepsilon_x)\) through a launch-power or transceiver mismatch. Because the success-ratio estimator~\eqref{eq:parameter-s-estimation} uses only \(2M-R = T\prod_k s_k\), it is independent of \(d^{(x)}\) and hence unbiased to all orders in \(\varepsilon\); the mismatch enters only the depolarization recovery, where the flipped entries of~\eqref{eq:combined-network-tomography-equation-alternative} carry a relative bias of order \(\varepsilon_x\,p_3^{(x)}\) that vanishes as \(\varepsilon_x\to0\). \newtext{Here \(p_3^{(x)}=d_3^{(x)}/r_3^{(x)}\) is link~3's depolarization probability of \S\ref{sec:model}, not a received ratio: the perturbation moves \(r_3^{(x)}\) to \(r_3^{(x)}+\varepsilon_x d_3^{(x)}\), a relative change of \(\varepsilon_x p_3^{(x)}\).}

%% file: sections/general-topology.tex
\subsection{Extension to General Topologies}\label{sec:general-topology}

While the preceding development in this section focused on the star network, the multiplicative nature of the success ratio $s_e$ and the received ratio $r_e^{(x_e)} = s_e + d_e^{(x_e)}$ (defined in Section~\ref{sec:model}) allows for an extension to general topologies.
In a general network, an end-to-end quantum probe traverses a path $P$ consisting of a sequence of coexisting fiber links. Based on the model in Section~\ref{sec:model}, the following proposition expresses the end-to-end observables of such a path in terms of the per-link parameters, generalizing the two-link star-network relations~\eqref{eq:group2}--\eqref{eq:group1}.

\begin{proposition}\label{prop:path-observables}
    Consider an end-to-end quantum probe that traverses a path $P$ of coexisting fiber links, where each link $e\in P$ has success ratio $s_e$ and depolarization ratio $d_e^{(x_e)}$, with $x_e \in \{1, 2\}$ denoting the coexistence scenario (co- or counter-propagation) encountered on link $e$. Let $T_P$, $R_P$, and $M_P$ denote the numbers of input, received, and matched-outcome photons of the probe, respectively. Then the expected end-to-end observables satisfy
    \begin{align}
        \frac{R_P}{T_P}      & = \prod_{e \in P} (s_e + d_e^{(x_e)}) \eqqcolon \prod_{e \in P} r_e^{(x_e)}, \label{eq:general-R} \\
        \frac{2M_P - R_P}{T_P} & = \prod_{e \in P} s_e. \label{eq:general-S}
    \end{align}
\end{proposition}

\begin{proof}
We generalize the two-link composition~\eqref{eq:group2}--\eqref{eq:group1} to an $n$-link path by tracking how the received-photon population splits according to polarization.
Order the path as $P=(e_1,\dots,e_n)$, with the probe entering $e_1$ first and exiting $e_n$ last, and abbreviate $s_j\coloneqq s_{e_j}$, $d_j\coloneqq d_{e_j}^{(x_{e_j})}$, and $r_j\coloneqq s_j+d_j$ for the success, depolarization, and received ratios of the $j$-th link.
\newtext{Split the population surviving the first $j$ links into two disjoint groups. Let $I_j$ be the expected number that are polarization-\emph{intact}, i.e.\ successfully transmitted by every one of those $j$ links, so the encoded state is preserved. Let $D_j$ be the expected number that are \emph{depolarized}, with uniformly random polarization---either signal photons depolarized along the way or Raman photons injected by the classical signal.}
Every input photon starts intact, so $I_0=T_P$ and $D_0=0$.
Because an intact photon always yields the matched outcome whereas a depolarized one does so with probability $\tfrac12$ (cf.\ the single-link rule~\eqref{eq:measure-ratio}), the received and matched counts after $j$ links are
\begin{equation}
    R^{(j)}=I_j+D_j,\quad M^{(j)}=I_j+\tfrac12 D_j
    \;\Longrightarrow\;
    2M^{(j)}-R^{(j)}=I_j .
    \label{eq:twoMminusR}
\end{equation}
That is, $2M-R$ counts exactly the photons that survive \emph{intact} through all $j$ links.

\newtext{Adding an additional link $e_j$ results in population changes analogous to those of the single-link channel of \S\ref{sec:model}.
An intact photon stays intact only when it is successfully transmitted (ratio $s_j$); otherwise it is depolarized (intrinsic fraction $d_j^{(0)}$) or joined by a freshly injected Raman photon (fraction $d_j-d_j^{(0)}$), and the two together contribute $d_j\,I_{j-1}$ depolarized photons.
An already-depolarized photon, by contrast, can never recover its polarization: it survives depolarized with probability $s_j+d_j^{(0)}$, and the Raman photons scattered from the classical signal make up the remainder.
The depolarized stream is therefore scaled by exactly $r_j=s_j+d_j$, a population factor rather than the survival probability of a single photon.}
This gives the linear recursion
\begin{equation}
    I_j=s_j\,I_{j-1},\qquad
    D_j=d_j\,I_{j-1}+r_j\,D_{j-1}.
    \label{eq:ID-recursion}
\end{equation}
For $n=2$ this reproduces the two-link relations exactly: $M_2=I_2+\tfrac12 D_2$ recovers~\eqref{eq:network-tomography-measurement}, while $R_2=I_2+D_2=(s_1+d_1)(s_2+d_2)\,T_P$ recovers~\eqref{eq:group2}.

The intact branch of~\eqref{eq:ID-recursion} is decoupled and telescopes immediately to
$I_n=\big(\prod_{j=1}^n s_j\big)I_0=T_P\prod_{e\in P}s_e$.
Substituting this into the identity~\eqref{eq:twoMminusR} at the end of the path ($j=n$) yields
\begin{equation*}
    \frac{2M_P-R_P}{T_P}=\frac{I_n}{T_P}=\prod_{e\in P}s_e,
\end{equation*}
which is~\eqref{eq:general-S}; the two-link estimator~\eqref{eq:group1} is the special case $n=2$.
Likewise, summing the two parts of~\eqref{eq:ID-recursion} gives $R^{(j)}=I_j+D_j=r_j\big(I_{j-1}+D_{j-1}\big)=r_j R^{(j-1)}$, so the received count is multiplicative as well,
\begin{equation*}
    R_P=R^{(n)}=\Big(\prod_{j=1}^n r_j\Big)T_P=T_P\prod_{e\in P}\big(s_e+d_e^{(x_e)}\big),
\end{equation*}
which is~\eqref{eq:general-R}.
\end{proof}

\newtext{\begin{remark}[Injection convention]\label{rem:injection-convention}
With $\nu_j\coloneqq d_j-d_j^{(0)}$ the Raman excess, the recursion~\eqref{eq:ID-recursion} adds $\nu_j R^{(j-1)}$ at link $j$: injected photons are proportional to the population \emph{arriving} there. This matches the per-link calibration, in which $\nu_j$ is an excess per received photon, and is what the emulation of \S\ref{subsec:star-experiment} generates. Spontaneous Raman is pumped by the classical signal alone, however, so a deployed path would attenuate link $j$'s injection only \emph{downstream}, giving $R_P/T_P=\prod_j t_j+\sum_j\nu_j\prod_{k>j}t_k$ with $t_j\coloneqq s_j+d_j^{(0)}$. The two agree when one link dominates the injection or the per-hop transmittances approach unity. We keep the multiplicative form, under which the estimators below are derived and validated.
\end{remark}}

\paragraph{Relation to prior work and novelty} The recursive peeling/promotion machinery we use below for trees, and the path-decomposition idea for meshes, are the progressive-etching procedure of our prior work~\cite{wang2025quantum}, developed there for bit-/phase-flip channels with SPAM errors; we do not claim that recursion as new. The contribution of this section is what makes that machinery \emph{applicable} to coexistence channels: (i)~Proposition~\ref{prop:path-observables}, which shows that the coexistence observables \(R_P/T_P\) and \((2M_P-R_P)/T_P\) compose \emph{multiplicatively} along a path \newtext{even though Raman injection makes the received ratio \(r_e^{(x_e)}\) a per-photon \emph{yield} rather than a survival probability (\S\ref{sec:cptp-model})}---a property absent from the additive flip-channel observables of~\cite{wang2025quantum}; and (ii)~the \emph{direction-dependent} unknowns \(d_e^{(1)},d_e^{(2)}\), resolved by the single-link classical-flip protocol of \S\ref{sec:network-tomography}, which have no analogue there.

\subsubsection{Peeling Algorithm for Tree Topologies}

\newtext{\paragraph{Notation} The network is an undirected graph in which each link $e=\{u,v\}$ is one coexisting fiber with parameters $s_e$ and $d_e^{(1)},d_e^{(2)}$. Here it is a \emph{tree} $\mathcal{T}=(V,E)$: connected and acyclic, so $|E|=|V|-1$ and any two nodes are joined by a unique path, whose link set we write $P(i,j)$. A \emph{leaf} has degree one, a \emph{leaf link} is incident to a leaf, and \emph{monitors} (set $\mathcal{M}$, initially the leaves) are the nodes where probes are prepared and measured. We use the calligraphic $\mathcal{T}$ for the tree to avoid a collision with the launched-photon count $T$ of \S\ref{sec:model}.}

\newtext{By a \emph{probe ratio} we mean one of the two end-to-end path observables of Proposition~\ref{prop:path-observables}, formed from the counts of probes exchanged between monitors $i$ and $j$ and, as in \S\ref{sec:network-tomography}, averaged over the two probe directions:}
\begin{align}
    \newtext{Q_{i,j}}       & \newtext{\;\coloneqq\; \tfrac12\bigl(Q^{\rightarrow}_{i,j}+Q^{\rightarrow}_{j,i}\bigr) \;\overset{\eqref{eq:general-S}}{=}\; \prod_{e\in P(i,j)} s_e ,} \label{eq:probe-ratio-Q} \\
    \newtext{G_{i,j}}       & \newtext{\;\coloneqq\; \sqrt{\Lambda_{i,j}\,\Lambda_{j,i}} \;\overset{\eqref{eq:general-R}}{=}\; \prod_{e\in P(i,j)} g_e ,} \label{eq:probe-ratio-Lambda}
\end{align}
\newtext{where $Q^{\rightarrow}_{i,j}\coloneqq(2M_{i,j}-R_{i,j})/T_{i,j}$ and $\Lambda_{i,j}\coloneqq R_{i,j}/T_{i,j}$ are the one-directional forms and $g_e\coloneqq\sqrt{r_e^{(1)}r_e^{(2)}}$. We call $Q_{i,j}$ the \emph{success probe ratio} and $G_{i,j}$ the \emph{received probe ratio}; these are the only two path quantities the estimators below consume. Averaging is optional for $Q$, whose factors $s_e$ do not depend on the propagation mode, but necessary for the received side: reversing the probe direction reverses the mode on every link, so $\Lambda_{i,j}\Lambda_{j,i}$ contributes $r_e^{(1)}r_e^{(2)}$ per link and $G_{i,j}$ alone is multiplicative in a direction-independent quantity. Note that $g_e$ is the gauge-invariant combination of Lemma~\ref{lem:star-identifiability}.}

Building on the path observables of Proposition~\ref{prop:path-observables}, the peeling estimator reuses the leaf-link-then-promote schedule of the progressive-etching procedure~\cite{wang2025quantum}, instantiated on the multiplicative success- and received-ratio observables of coexisting links.
For tree-structured networks, we can employ a ``peeling'' algorithm to identify all per-link parameters using probes between leaf nodes. A link $e$ incident to a leaf node $v$ can be isolated by considering paths to two other leaf nodes $j$ and $k$ located in different subtrees relative to $v$'s parent. Let $Q_{v,j}$, $Q_{v,k}$, and $Q_{j,k}$ be the end-to-end success ratios for the respective paths. Then, the success ratio of the leaf link $e$ is given by
\begin{equation}
    \hat{s}_e = \sqrt{\frac{Q_{v,j} Q_{v,k}}{Q_{j,k}}},
\end{equation}
which is identical in form to the star-network estimator in~\eqref{eq:parameter-s-estimation}. Once $\hat{s}_e$ is determined, its effect can be ``peeled off'' from other paths containing $e$, allowing for the recursive estimation of internal link parameters. \newtext{The identical construction on the symmetrized received probe ratios returns $\hat{g}_e=\sqrt{G_{v,j}G_{v,k}/G_{j,k}}$. (The unsymmetrized $\Lambda_{i,j}$ cannot be used here: the mode indices on the shared sub-paths of $v\!\to\!j$ and $j\!\to\!k$ are then opposite and fail to cancel.)} This recursive procedure---identify a current leaf link, peel it off, and promote its far endpoint to a monitor so that interior links in turn become leaf links---is summarized in Algorithm~\ref{alg:peeling} and illustrated on a five-link tree in Fig.~\ref{fig:peeling-steps}.
Two preconditions make each step well defined. First, the chosen leaf link must be \emph{branching}: its interior endpoint $u$ must have two further monitors in distinct components of \newtext{$\mathcal{T}\setminus\{u\}$}, so that the paths $v\!\to\! j$ and $v\!\to\! k$ traverse $e$ and diverge at $u$ while $j\!\to\! k$ avoids $e$; this is exactly the condition under which $\sqrt{Q_{v,j}Q_{v,k}/Q_{j,k}}=s_e$. Any tree with at least three leaves contains such a link, and peeling preserves the property---removing $e$ and promoting $u$ keeps the residual graph a tree with the promoted node available as a required monitor---so the loop terminates after $|E|$ peels. (A pure degree-two chain, in which no interior node branches, is the lone exception and falls back to the log-linear least-squares formulation \newtext{described} below.) Second, every $Q$ entering the square root must be strictly positive, which holds once the sample size is large enough that the empirical $Q=(2M-R)/T$ stays positive.

\newtext{\paragraph{Separating the two propagation modes} Symmetrization is deliberately lossy: it discards the direction asymmetry $\Lambda_{i,j}/\Lambda_{j,i}$, which is exactly what distinguishes $r_e^{(1)}$ from $r_e^{(2)}$. Writing $\delta_e\coloneqq\tfrac12\log\bigl(r_e^{(1)}/r_e^{(2)}\bigr)$, so that $\log r_e^{(1,2)}=\log g_e\pm\delta_e$, each probe contributes a \emph{signed} row in the $\delta_e$ whose signs are fixed by the modes it meets; one traffic configuration therefore leaves only the kernel of the resulting matrix $A_1$ undetermined (one dimensional for the three-link star, matching Lemma~\ref{lem:star-identifiability}). A second round reversing the classical signal on a known subset $S\subseteq E$ contributes further rows $A_S$, and the $d_e^{(x)}$ are identifiable if and only if the \emph{stacked} matrix $[A_1;A_S]$ has full column rank $|E|$---which the single-link flip of \S\ref{sec:network-tomography} achieves on the star. Being signed and supported only on $S$, this system admits no leaf-link cancellation, so we solve it by least squares rather than by peeling.}

\subsubsection{General Graphs and Path Selection}

For general mesh topologies, the estimation problem can be transformed into a system of linear equations by taking the logarithm of the path observables in Proposition~\ref{prop:path-observables} (Eqs.~\eqref{eq:general-R} and~\eqref{eq:general-S}). For instance, for the success ratios, we have
\begin{equation}
    \sum_{e \in P} \log s_e = \log \left( \frac{2M_P - R_P}{T_P} \right).
\end{equation}
This is precisely the log-additive linear system of classical loss tomography~\cite{caceres1999multicast}, here applied to the intact-photon success ratios $s_e$ rather than to classical packet pass-probabilities.
Given a sufficient set of independent paths $\{P\}$, this system can be solved using standard linear regression or maximum likelihood techniques. To resolve the direction-dependent depolarization parameters $d_e^{(1)}$ and $d_e^{(2)}$, one can perform multiple measurement rounds with different classical signal configurations, ensuring that each link is probed in both co- and counter-propagation modes.

The condition for this system to be solvable is exactly the classical route-identifiability condition~\cite{he2021network}.
\begin{proposition}[Identifiability]\label{prop:identifiability}
Fix a probe set $\mathcal{P}$ and let $A\in\{0,1\}^{|\mathcal{P}|\times|E|}$ be the path-by-link incidence matrix with $A_{P,e}=1$ iff $e\in P$. The per-link log-success-ratios $\{\log s_e\}_{e\in E}$ (resp.\ the log-received-ratios $\{\log r_e^{(x)}\}$) are uniquely identifiable from the expected end-to-end observables~\eqref{eq:general-S} (resp.~\eqref{eq:general-R}) if and only if $A$ has full column rank, $\operatorname{rank}(A)=|E|$.
\end{proposition}
\begin{proof}
Taking logarithms of~\eqref{eq:general-S} gives the linear system $A\mathbf{x}=\mathbf{b}$ with $x_e=\log s_e$ and $b_P=\log((2M_P-R_P)/T_P)$. A linear map is injective iff its matrix has full column rank, so $\mathbf{x}$ is determined uniquely by $\mathbf{b}$ iff $\operatorname{rank}(A)=|E|$; the same argument applies to $\{\log r_e^{(x)}\}$.
\end{proof}
For trees, the leaf-to-leaf paths yield full column rank and Algorithm~\ref{alg:peeling} is the constructive solver realizing it; for meshes, any probe set with $\operatorname{rank}(A)=|E|$ gives a unique least-squares estimate. Separating $d_e^{(1)}$ from $d_e^{(2)}$ additionally requires each link to appear in $\mathcal{P}$ under both propagation modes, generalizing the single-link flip of \S\ref{sec:network-tomography}.

\newtext{\paragraph{Relation to the star gauge fixing} This multi-round protocol is the general-topology form of the single-link classical flip of \S\ref{sec:network-tomography}, and the rank condition of Proposition~\ref{prop:identifiability} plays the role that Lemma~\ref{lem:star-identifiability} plays for the star. Each probe touches, on every link it traverses, only the one propagation mode it happens to encounter; so when monitors sit at the network edge and no probe isolates a single link, a path can meet different links in opposite modes, which may leave $A$ column-rank deficient---its kernel being the residual gauge. The two-link star is the canonical case. Reconfiguring the classical traffic adds no unknowns (Assumption~\ref{ass:flip-invariance}) but adds rows, so the incidence-rank condition is really a statement about how many distinct traffic configurations must be exercised, not merely how many paths are probed.}

\subsubsection{Sample Size and Accuracy}

As suggested in related network tomography studies, the accuracy of these estimators depends on the number of probes $T_P$. As the network depth increases, the end-to-end success ratio $\prod s_e$ decreases exponentially, requiring larger sample sizes to maintain low relative error. We confirm these properties with Monte-Carlo simulations in the Validation section (\S\ref{subsec:general-topology-experiment}), where the peeling and least-squares estimators recover all per-link parameters on tree and cyclic-mesh networks and the required sample size is shown to grow exponentially with path depth; there we also observe that symmetrized probes (averaging $Q_{k,\ell}$ and $Q_{\ell,k}$) significantly reduce the variance in noisy environments.

\newtext{Since these estimators are $\sqrt{T}$-consistent (\S\ref{subsec:link-statistics}), a variance reduction of factor $\kappa$ is a sample-size reduction of the same factor. Averaging the two probe directions combines two unbiased estimates of equal variance and correlation $c$, giving $\kappa_{\mathrm{sym}}=2/(1+c)\in[1,2]$ for $c\in[0,1]$: at most a factor of two, attained only when the directions are independent, and less in practice since both traverse the same fiber. For calibration, classical loss tomography ties estimator variance directly to probe count~\cite{caceres1999multicast}, and Lawrence et al.~\cite{lawrence2007statistical} report, in a delay-tomography study, that least-squares moment estimators need roughly $1.3\times$--$9\times$ the probes of the MLE for equal variance ($1.1\times$--$1.3\times$ under generalized least squares). The larger classical savings thus come from estimator choice and probe \emph{design}~\cite{he2015fisher,kveton2022optimal,wang2025optimal} rather than from directional symmetrization, for which we know of no classical analogue.}

\begin{algorithm}[!t]
    \caption{Peeling estimator for tree topologies}
    \label{alg:peeling}
    \begin{algorithmic}[1]
        \Input \newtext{Tree $\mathcal{T}=(V,E)$; probe ratios $Q_{i,j}$ and $G_{i,j}$ of~\eqref{eq:probe-ratio-Q}--\eqref{eq:probe-ratio-Lambda} between monitor (leaf) nodes}
        \Output \newtext{Per-link estimates $\{\hat{s}_e,\hat{g}_e\}_{e\in E}$; splitting $\hat{g}_e$ into $\hat{r}_e^{(1)},\hat{r}_e^{(2)}$ needs a second traffic configuration}
        \Initial \newtext{$\mathcal{M}\gets\{\text{leaf nodes of }\mathcal{T}\}$};\quad $\mathcal{D}\gets\emptyset$
        \While{$\mathcal{D}\neq E$}
            \State pick a link $e=\{v,u\}$ with $v\in\mathcal{M}$, $u\notin\mathcal{M}$, and $e\notin\mathcal{D}$ \Comment{a current leaf link}
            \State choose monitors $j,k\in\mathcal{M}\setminus\{v\}$ whose paths to $v$ first meet at $u$
            \State \newtext{$\hat{s}_e \gets \bigl(Q_{v,j}\,Q_{v,k}\,/\,Q_{j,k}\bigr)^{1/2}$} \Comment{leaf-link estimator, Eq.~\eqref{eq:parameter-s-estimation}}
            \State \newtext{$\hat{g}_e \gets \bigl(G_{v,j}\,G_{v,k}\,/\,G_{j,k}\bigr)^{1/2}$} \Comment{received-ratio analogue}
            \State \textbf{peel:} \newtext{\textbf{for all} $m\in\mathcal{M}\setminus\{v\}$ \textbf{do} $\;Q_{u,m}\gets Q_{v,m}/\hat{s}_e$, $\;G_{u,m}\gets G_{v,m}/\hat{g}_e$} \Comment{$u$ inherits the residual probe ratios}
            \State $\mathcal{D}\gets\mathcal{D}\cup\{e\}$;\quad $\mathcal{M}\gets\mathcal{M}\cup\{u\}$ \Comment{promote $u$ to a monitor}
        \EndWhile
        \State \Return \newtext{$\{\hat{s}_e,\hat{g}_e\}_{e\in E}$}
    \end{algorithmic}%
    
\end{algorithm}

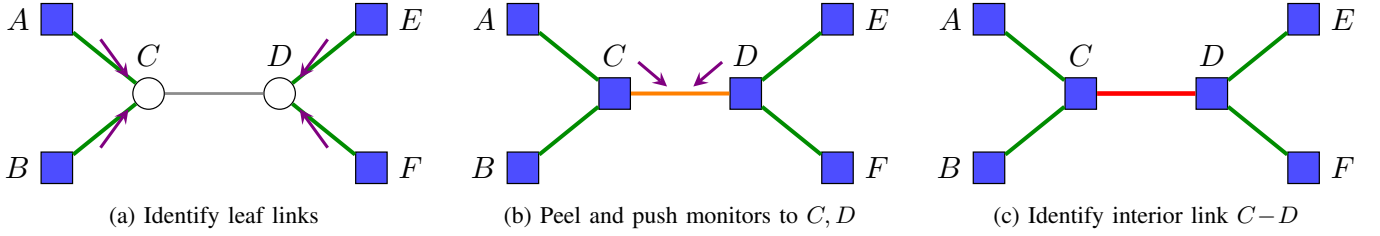
\begin{figure*}[!t]
    \centering
    \tikzset{
        mon/.style   ={draw,fill=blue!70,minimum size=3.6mm,inner sep=0pt},
        inode/.style ={draw,circle,fill=white,minimum size=3.6mm,inner sep=0pt},
        done/.style  ={green!55!black,line width=1.3pt},
        next/.style  ={orange,line width=1.3pt},
        last/.style  ={red,line width=1.5pt},
        todo/.style  ={black!45,line width=0.9pt},
        etch/.style  ={->,>=stealth,violet,line width=1pt},
    }
    \begin{subfigure}{0.32\textwidth}
        \centering
        \resizebox{\textwidth}{!}{%
        \begin{tikzpicture}
            \node[mon,label=left:$A$]    (A) at (0,0.85){};
            \node[mon,label=left:$B$]    (B) at (0,-0.85){};
            \node[inode,label=above:$C$] (C) at (1.05,0){};
            \node[inode,label=above:$D$] (D) at (2.55,0){};
            \node[mon,label=right:$E$]   (E) at (3.6,0.85){};
            \node[mon,label=right:$F$]   (F) at (3.6,-0.85){};
            \draw[done] (A)--(C); \draw[done] (B)--(C);
            \draw[todo] (C)--(D);
            \draw[done] (E)--(D); \draw[done] (F)--(D);
            \draw[etch] (0.50,0.62) -- (0.82,0.18);
            \draw[etch] (0.50,-0.62) -- (0.82,-0.18);
            \draw[etch] (3.10,0.62) -- (2.78,0.18);
            \draw[etch] (3.10,-0.62) -- (2.78,-0.18);
        \end{tikzpicture}}
        \caption{Identify leaf links}
    \end{subfigure}\hfill
    \begin{subfigure}{0.32\textwidth}
        \centering
        \resizebox{\textwidth}{!}{%
        \begin{tikzpicture}
            \node[mon,label=left:$A$]    (A) at (0,0.85){};
            \node[mon,label=left:$B$]    (B) at (0,-0.85){};
            \node[mon,label=above:$C$]   (C) at (1.05,0){};
            \node[mon,label=above:$D$]   (D) at (2.55,0){};
            \node[mon,label=right:$E$]   (E) at (3.6,0.85){};
            \node[mon,label=right:$F$]   (F) at (3.6,-0.85){};
            \draw[done] (A)--(C); \draw[done] (B)--(C);
            \draw[next] (C)--(D);
            \draw[done] (E)--(D); \draw[done] (F)--(D);
            \draw[etch] (1.32,0.36) -- (1.65,0.08);
            \draw[etch] (2.28,0.36) -- (1.95,0.08);
        \end{tikzpicture}}
        \caption{Peel and push monitors to $C,D$}
    \end{subfigure}\hfill
    \begin{subfigure}{0.32\textwidth}
        \centering
        \resizebox{\textwidth}{!}{%
        \begin{tikzpicture}
            \node[mon,label=left:$A$]    (A) at (0,0.85){};
            \node[mon,label=left:$B$]    (B) at (0,-0.85){};
            \node[mon,label=above:$C$]   (C) at (1.05,0){};
            \node[mon,label=above:$D$]   (D) at (2.55,0){};
            \node[mon,label=right:$E$]   (E) at (3.6,0.85){};
            \node[mon,label=right:$F$]   (F) at (3.6,-0.85){};
            \draw[done] (A)--(C); \draw[done] (B)--(C);
            \draw[last] (C)--(D);
            \draw[done] (E)--(D); \draw[done] (F)--(D);
        \end{tikzpicture}}
        \caption{Identify interior link $C\!-\!D$}
    \end{subfigure}
    \caption{Progressive peeling on a five-link tree (the same network
    used for the Monte-Carlo simulations in \S\ref{subsec:general-topology-experiment}),
    mirroring the etching illustration in
    \cite[Fig.~6]{wang2025quantum}. Blue squares are monitor (leaf) nodes and white
    circles are internal nodes; violet arrows indicate the peeling order from the
    periphery inward. (a) The four leaf links $A\!-\!C$, $B\!-\!C$, $D\!-\!E$, $D\!-\!F$
    (green) are identified from leaf-to-leaf probes via
    $\hat{s}_e=\sqrt{Q_{v,j}Q_{v,k}/Q_{j,k}}$, with monitors at the leaves $A,B,E,F$.
    (b) Each identified link is ``peeled'' and its far endpoint promoted to a monitor,
    pushing the monitors inward to $C$ and $D$ and exposing the interior link $C\!-\!D$
    (orange) as a link between two monitors. (c) The interior link $C\!-\!D$ is then
    identified by \newtext{the same leaf-link estimator $\hat{s}_e=\sqrt{Q_{v,j}Q_{v,k}/Q_{j,k}}$
    of~\eqref{eq:parameter-s-estimation}, now applied with $C$ and $D$ as monitors} (red).}
    \label{fig:peeling-steps}
\end{figure*}


%% file: sections/experiment.tex
\section{Validation}\label{sec:experiment}

\begin{figure*}[!t]
    \centering
    \begin{subfigure}{1.2\columnwidth}
        \centering
        \includegraphics[width=\columnwidth]{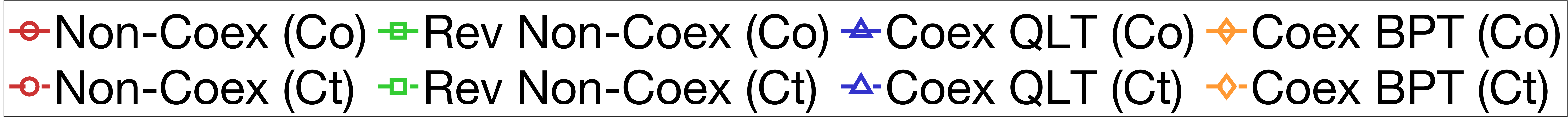}
    \end{subfigure}
    \\
    \begin{subfigure}{0.65\columnwidth}
        \centering
        \includegraphics[width=\columnwidth]{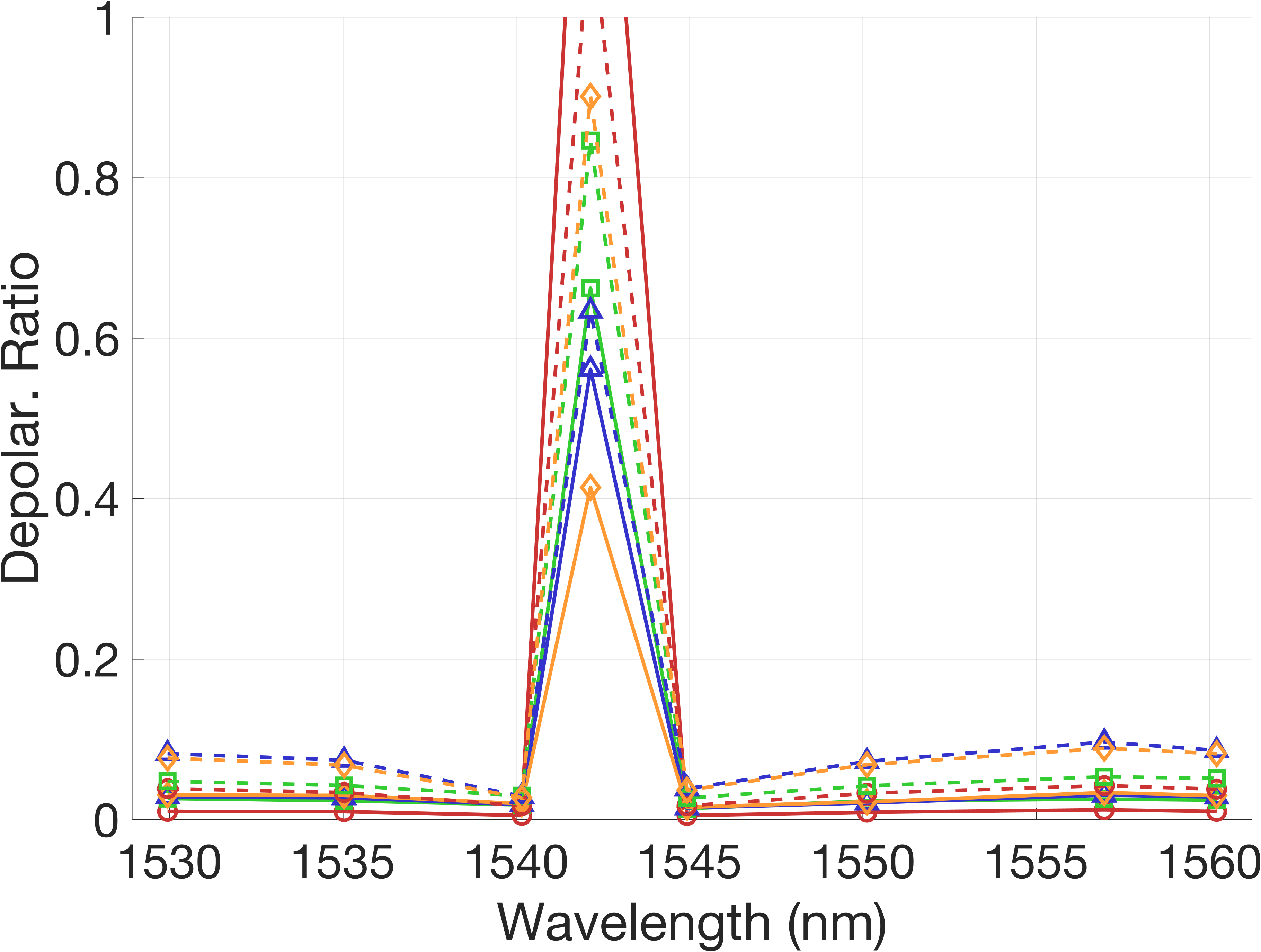}
        \caption{Depolar. probability at 0.5~km}
    \end{subfigure}
    \begin{subfigure}{0.65\columnwidth}
        \centering
        \includegraphics[width=\columnwidth]{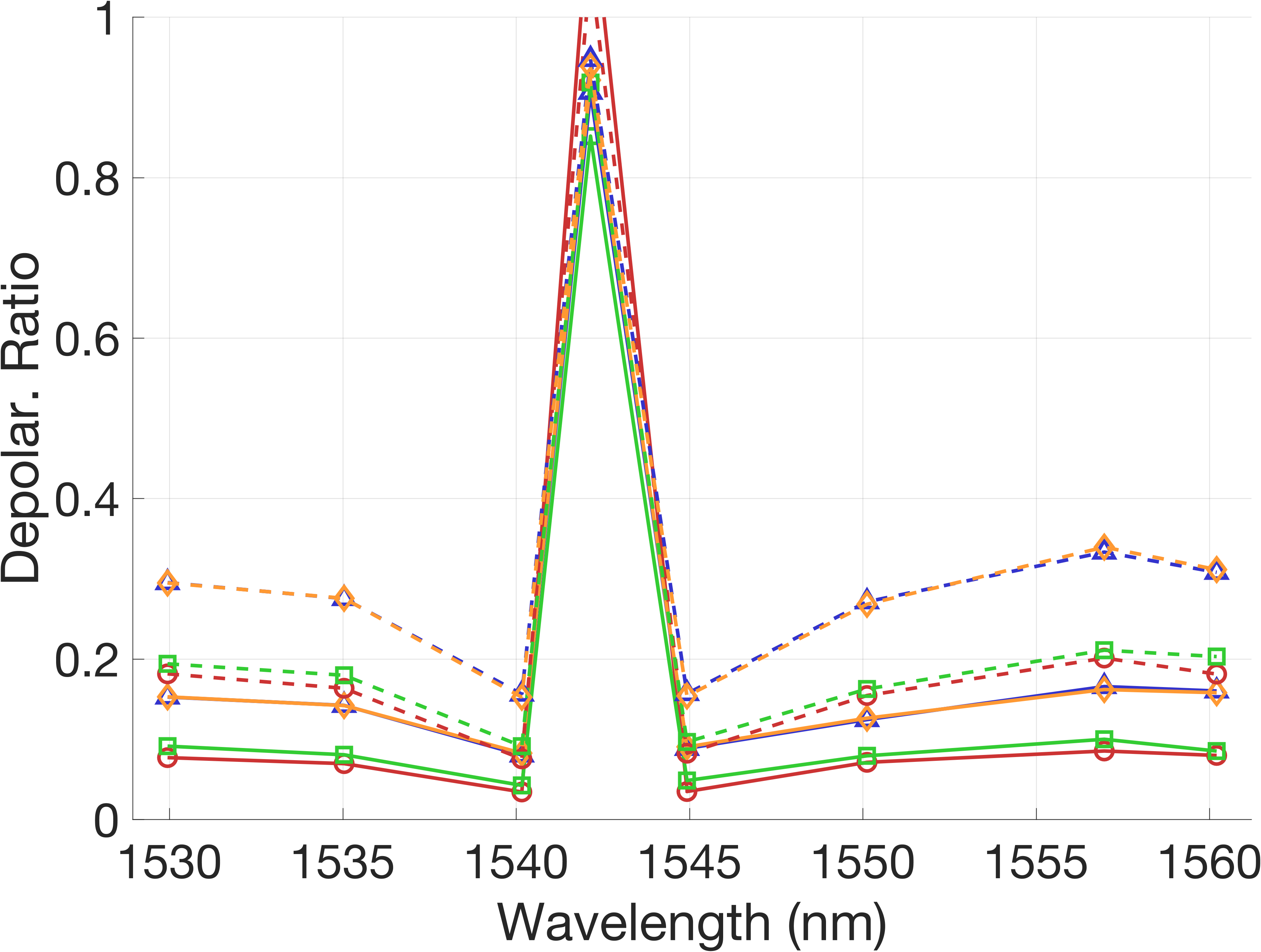}
        \caption{{Depolar. probability at 5.0~km}}
    \end{subfigure}
    \begin{subfigure}{0.65\columnwidth}
        \centering
        \includegraphics[width=\columnwidth]{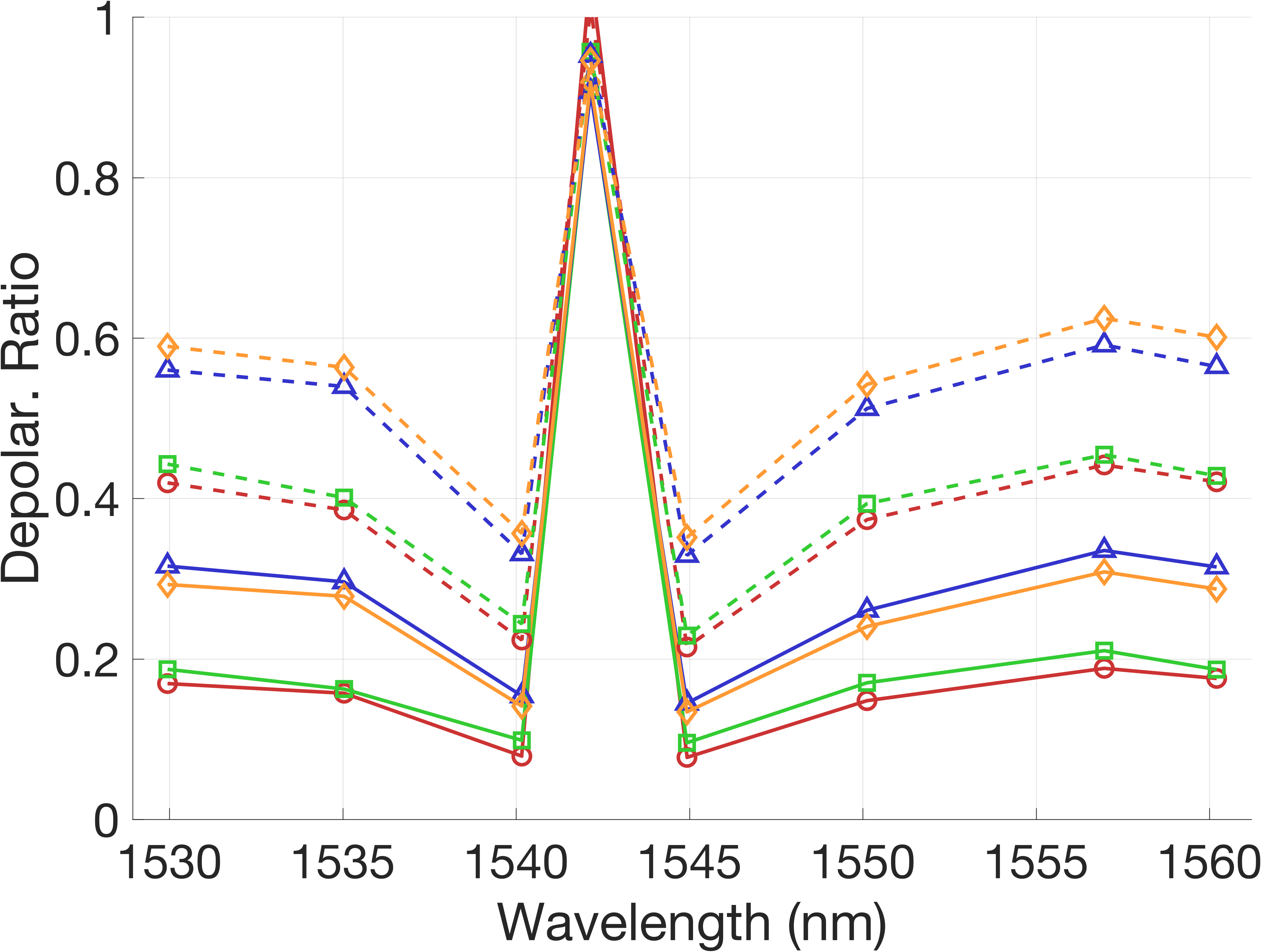}
        \caption{Depolar. probability at 15~km}
    \end{subfigure}
    \caption{{Depolarization probability vs.\ wavelength for different fiber lengths and propagation directions, where (Co) denotes co-propagation and (Ct) denotes counter-propagation. The feature near 1542~nm is a measurement artifact of the notch filter used to suppress the 1542.5~nm classical networking laser (imperfect filter isolation near the laser line)~\cite{chapman2023coexistent}, not a property of the depolarization channel.}}
    \label{fig:link-tomography-depolarization-probability-comparison}
\end{figure*}

\begin{figure*}
    \centering
    \begin{subfigure}{0.65\columnwidth}
        \centering
        \includegraphics[width=\columnwidth]{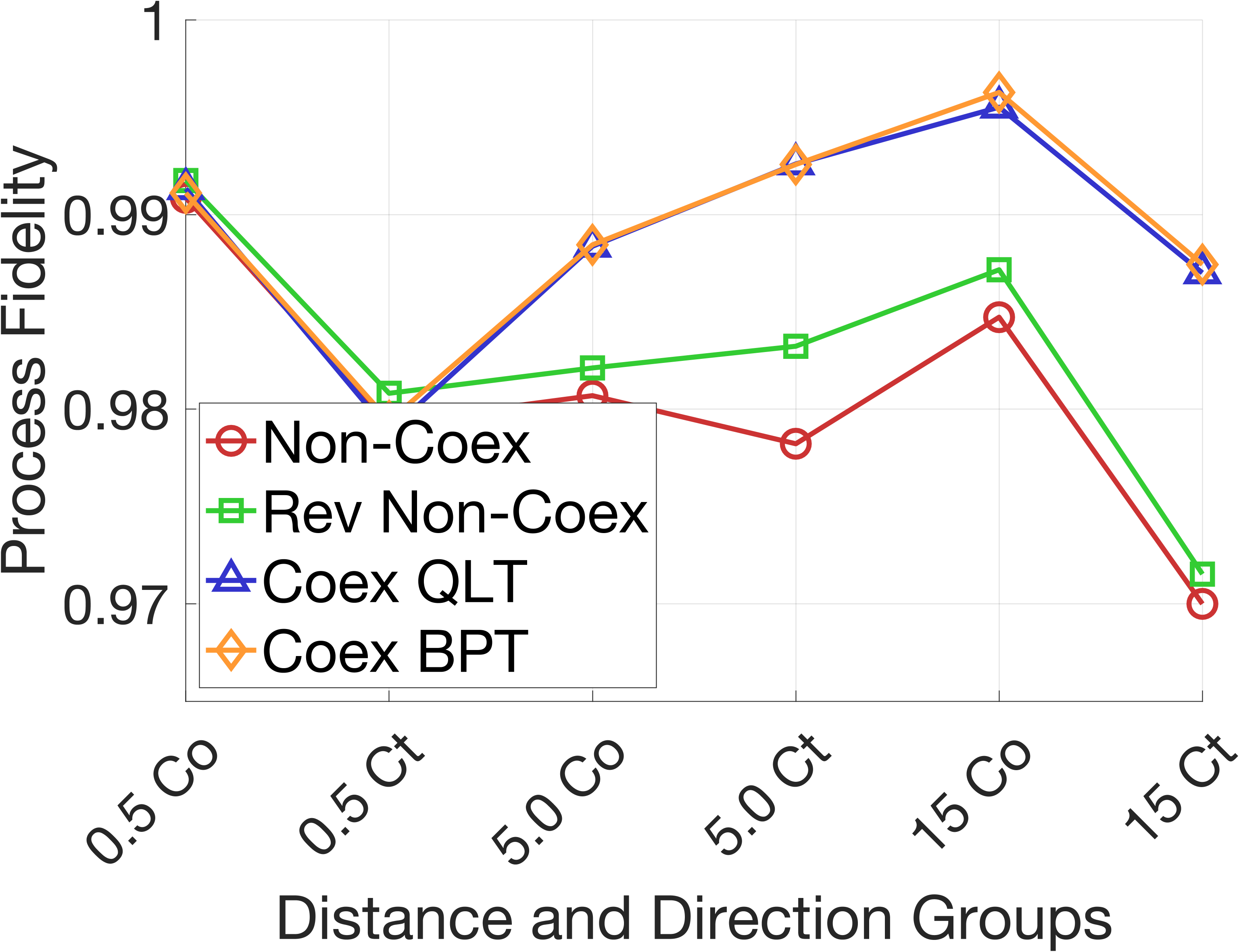}
        \caption{1560~nm wavelength}
    \end{subfigure}
    \begin{subfigure}{0.65\columnwidth}
        \centering
        \includegraphics[width=\columnwidth]{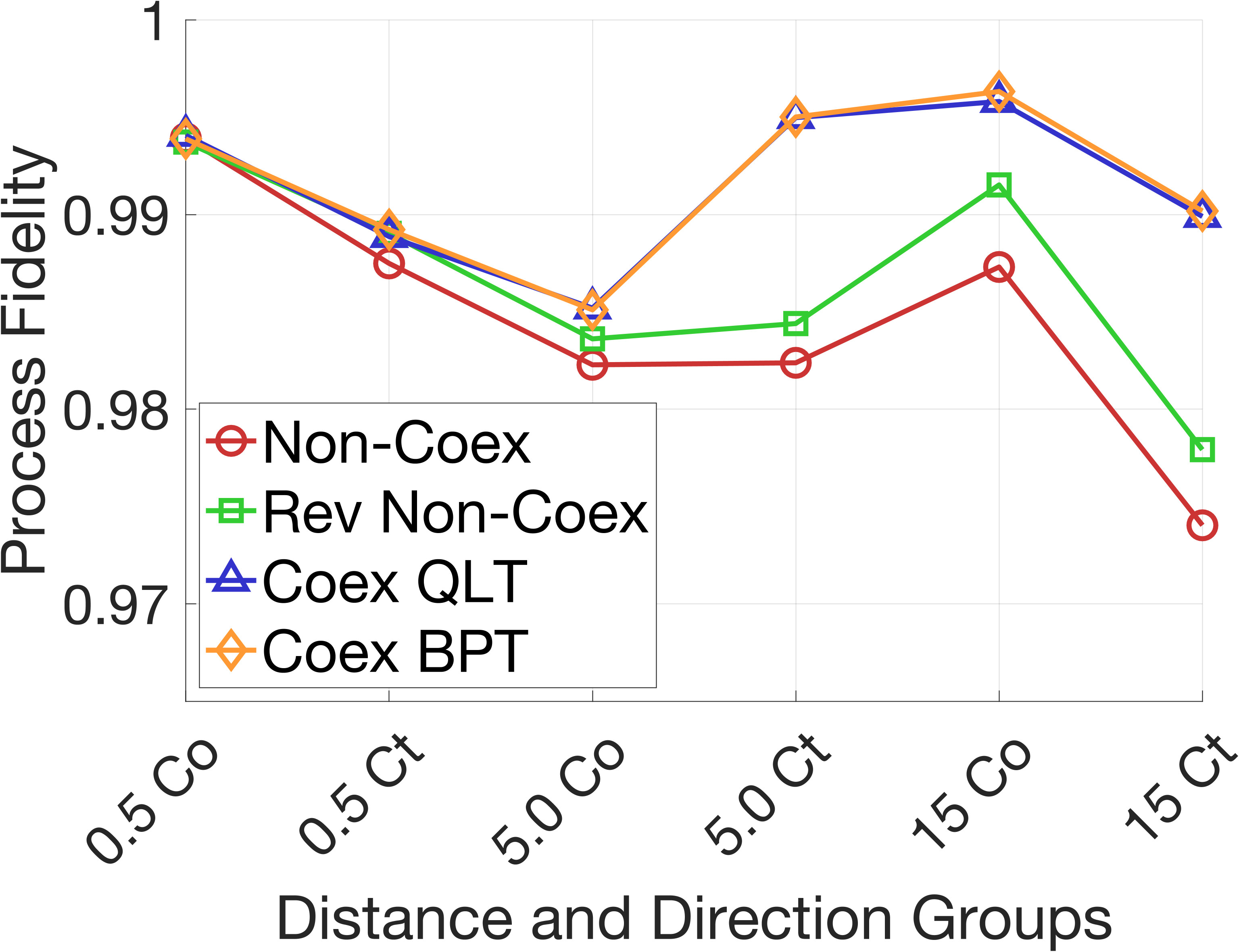}
        \caption{1550~nm wavelength}
    \end{subfigure}
    \begin{subfigure}{0.65\columnwidth}
        \centering
        \includegraphics[width=\columnwidth]{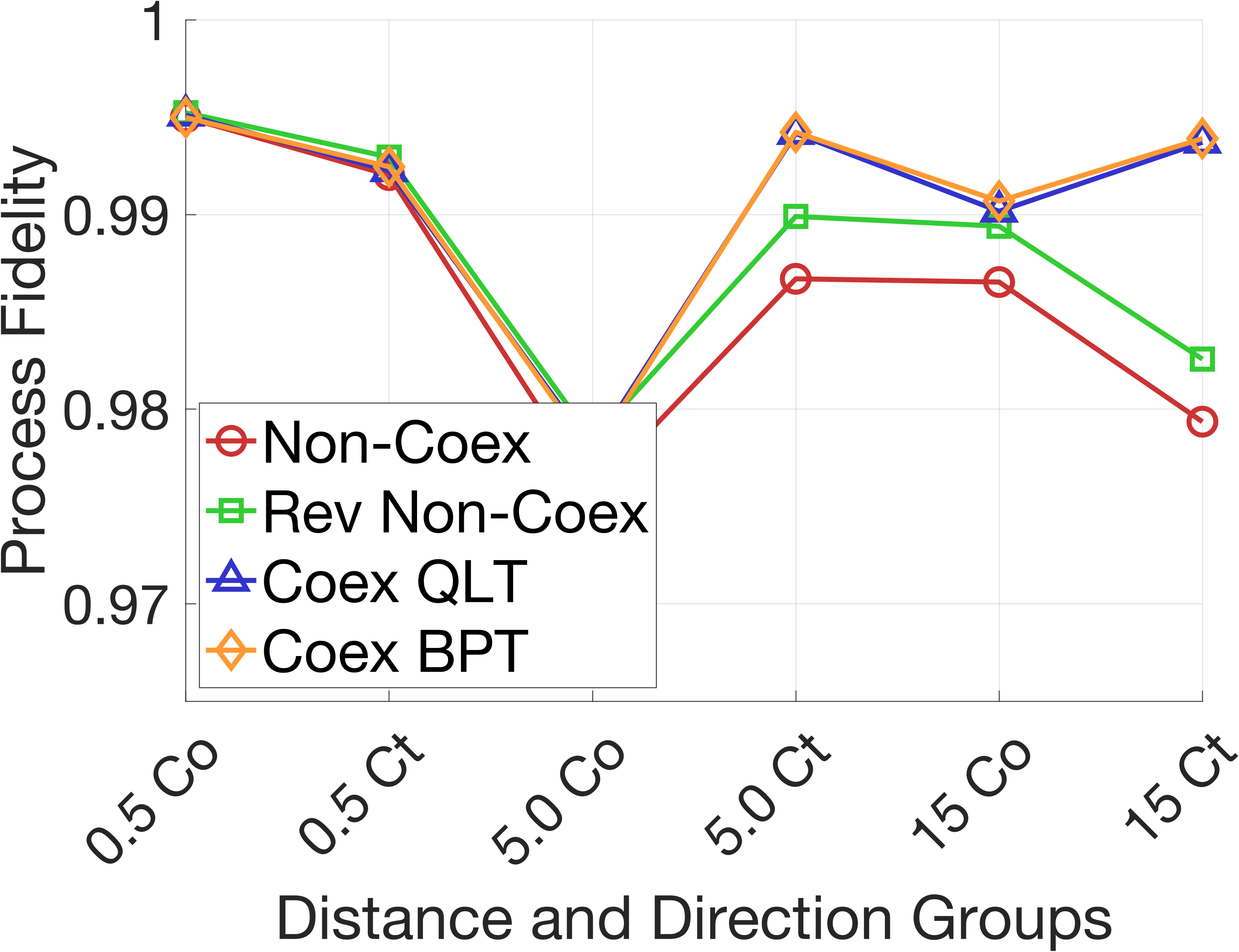}
        \caption{1545~nm wavelength}
    \end{subfigure}
    \\
    \begin{subfigure}{0.65\columnwidth}
        \centering
        \includegraphics[width=\columnwidth]{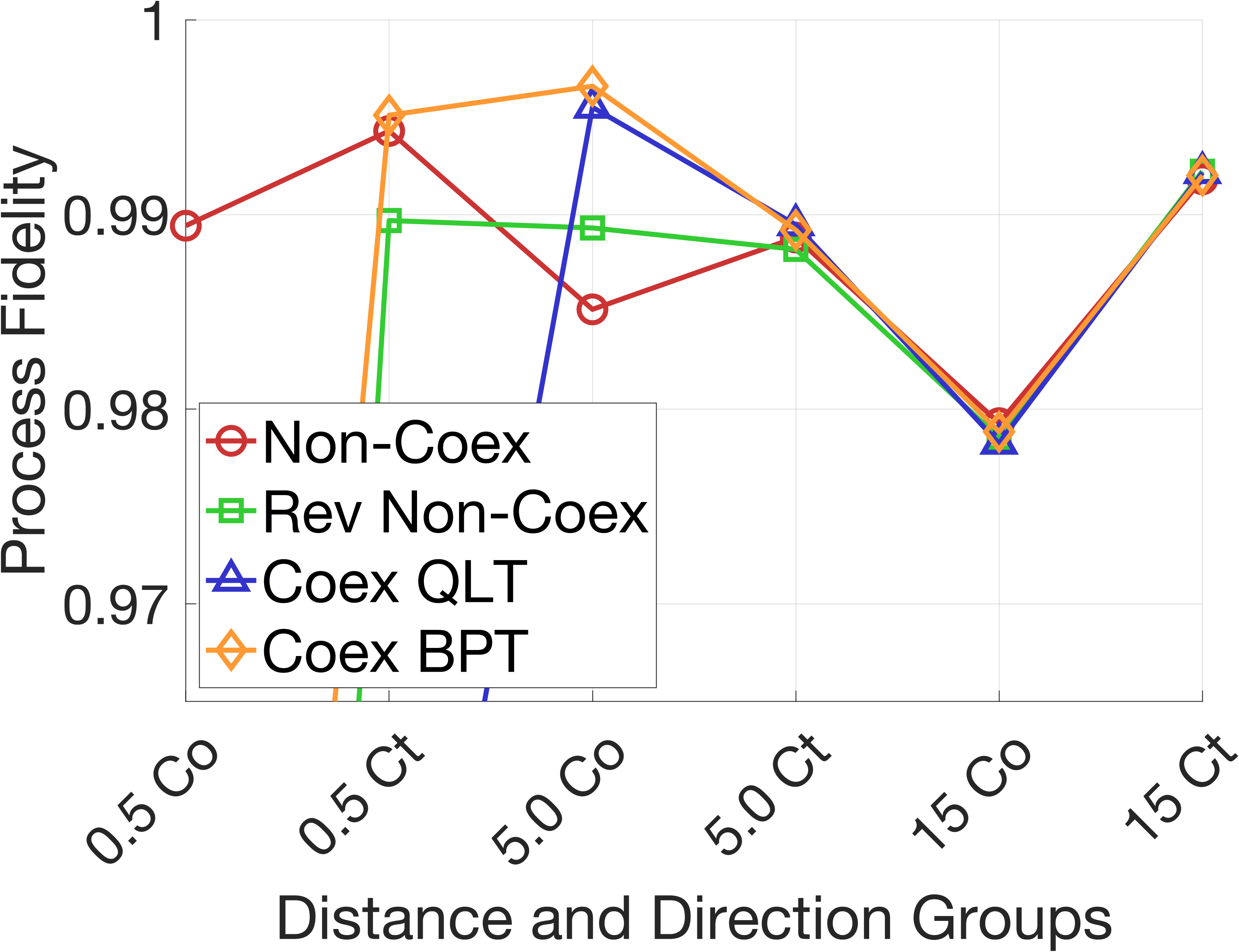}
        \caption{1542~nm wavelength (artifact)}
    \end{subfigure}
    \begin{subfigure}{0.65\columnwidth}
        \centering
        \includegraphics[width=\columnwidth]{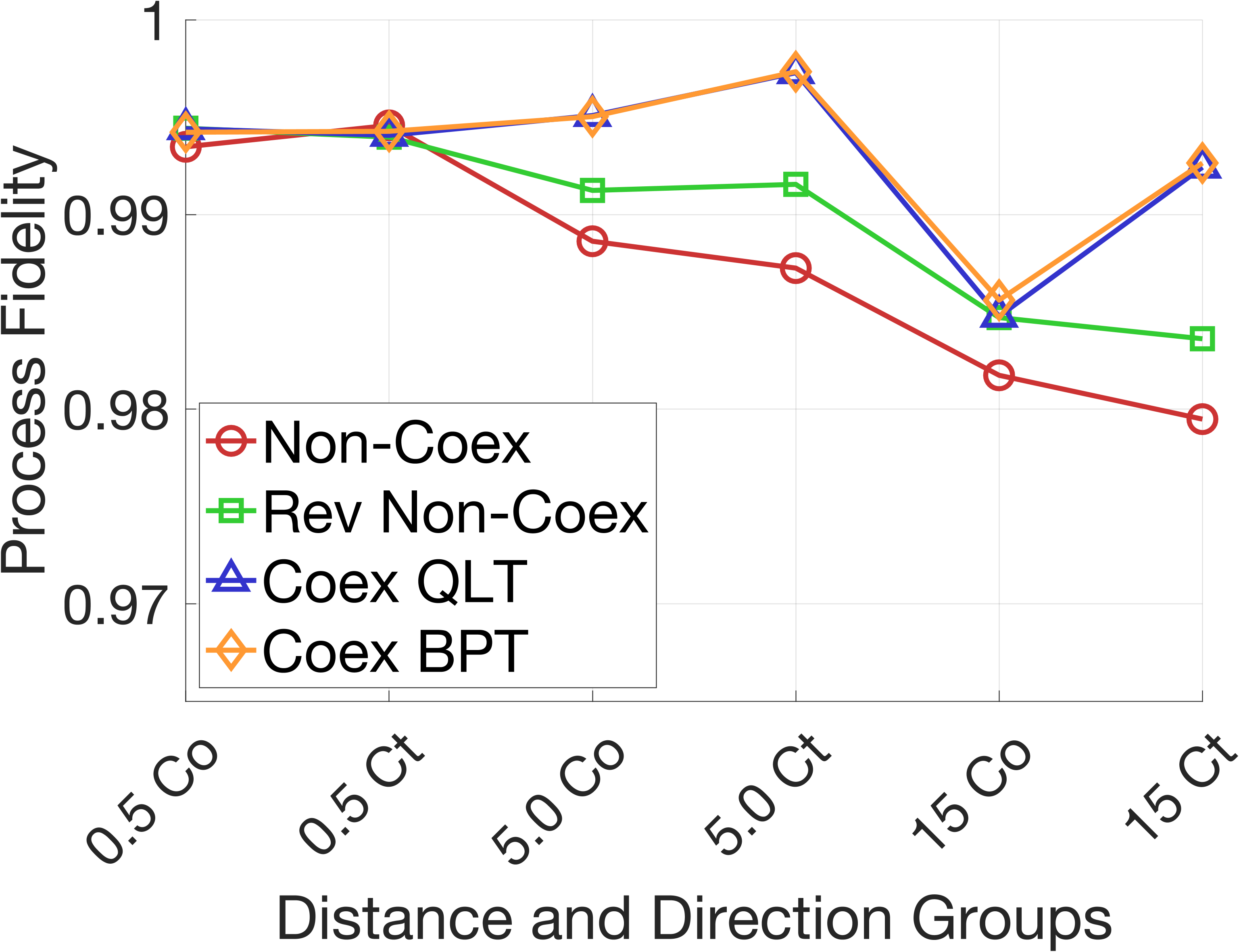}
        \caption{1540~nm wavelength}
    \end{subfigure}
    \begin{subfigure}{0.65\columnwidth}
        \centering
        \includegraphics[width=\columnwidth]{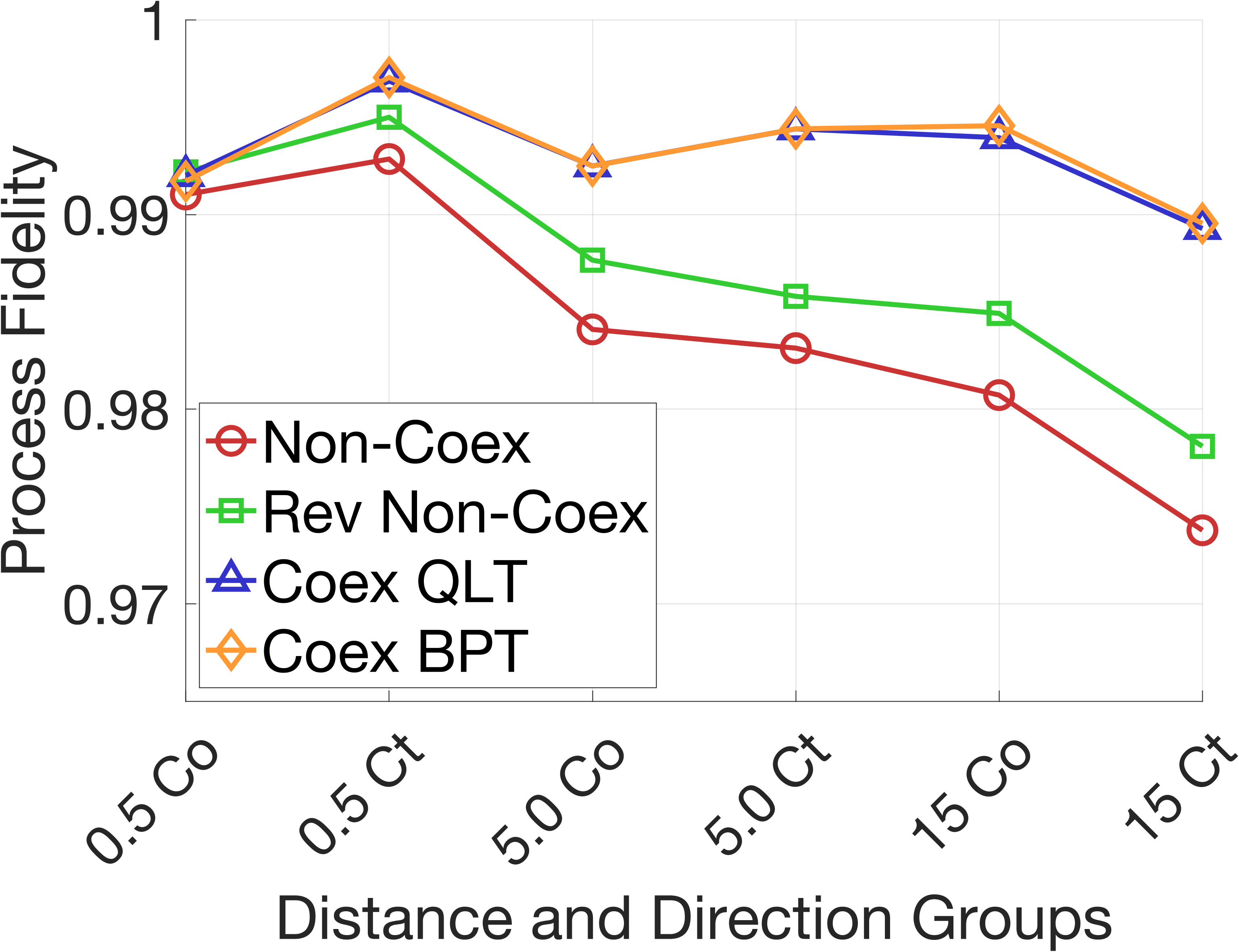}
        \caption{1530~nm wavelength}
    \end{subfigure}
    \caption{{Process fidelity vs.\ distance (km) and propagation direction (co- vs.\ counter-propagating) for six representative wavelengths (the 1535 and 1555~nm panels are omitted; their behavior is qualitatively identical) and different methods}: (1) Non-Coex, the method of Chapman et al.~\cite[Figure~10]{chapman2023coexistent};
        (2) Rev Non-Coex, accounting for the depolarization observed under pure quantum transmission;
        (3) Coex QLT, the proposed coexistence quantum link tomography estimator using both co- and counter-propagation data (\S\ref{sec:link-tomography});
        (4) Coex BPT, Bayesian process tomography from Chapman et al.~\cite[\S II]{chapman2023coexistent}. The residual fidelity gap between QLT and BPT reflects the off-depolarizing (coherent/non-unital) weight that the depolarizing model cannot represent.
    }
    \label{fig:link-tomography-process-fidelity-comparison}
\end{figure*}

This section reports the performance of our proposed estimators for link, star-network, and general-topology tomography.
Throughout, we distinguish three validation modes: \emph{experiments} denote lab-based testbed measurements; \emph{emulations} denote studies that recombine the measured single-link data into multi-link configurations; and \emph{simulations} denote Monte-Carlo studies in which all data, including every per-link parameter, are generated on a computer with no testbed provenance.
For link tomography (\S\ref{subsec:link-experiment}), we utilize experimental data collected by Chapman et al.~\cite{chapman2023coexistent} and focus on verifying the correctness of our channel model and the accuracy of our proposed estimators.
For star network tomography (\S\ref{subsec:star-experiment}), we emulate a multi-link star network from link-level data drawn from the same experimental dataset, under two distinct link-to-network compilations.
Finally, we validate the general-topology estimators of \S\ref{sec:general-topology} on simulated tree and mesh networks (\S\ref{subsec:general-topology-experiment}).

\subsection{Dataset and Testbed Description}\label{subsec:dataset}

We use subsets of the coexisting fiber dataset collected by Chapman et al.~\cite{chapman2023coexistent}, which consists of polarization-encoded single-photon measurements over fibers of three different lengths: 0.5~km, 5.0~km, and 15~km.
For each fiber length, the data span the three coexistence/depolarization scenarios \(x\in\{0,1,2\}\)---(0)~no classical signal (pure quantum signal), (1)~co-propagation, (2)~counter-propagation---together with an auxiliary classical-only run (no quantum signal, pure classical signal) used only as a baseline in this section. We refer to scenario (0) and the classical-only run as the \emph{non-coexistence} cases.
The quantum-channel wavelength varies across 91 C-band wavelengths from 1528.38~nm to 1564.68~nm (i.e., 196.150~THz to 191.600~THz), enabling the study of wavelength-dependent Raman noise effects. Eight wavelengths are selected for analysis.

\textbf{Input states and measurements.}
In each experimental run, single photons are prepared in one of six polarization states spanning three mutually unbiased bases \newtext{(MUBs)}:
horizontal/vertical (\(\ket{H}, \ket{V}\)),
diagonal/anti-diagonal (\(\ket{D}, \ket{A}\)),
and right/left circular (\(\ket{R}, \ket{L}\)).
For each input state, the received photons are measured in each of the three bases, giving three sub-experiments per input state.
In each sub-experiment, photons are prepared in the fixed input state, transmitted over the link, and measured in the fixed basis, with a detected count on the order of \(10^5\) (the \(\sim\!1.25\times10^5\) baseline reported by Chapman et al.~\cite{chapman2023coexistent}).
In total, the experimental configuration yields 36 input-measurement combinations (6 input states \(\times\) 6 measurement outcomes).


\textbf{Deriving the observables \(T\), \(R\), and \(M\).}
\newtext{From each experimental configuration we form the \(6\times6\) count matrix of the six input states against the six measurement outcomes (the three MUBs). From this matrix we extract the three quantities required by our channel model (\S\ref{sec:model}). The first is the total number of input photons \(T\), i.e.\ the effective launched-photon count. We back-calibrate \(T\) from the detected counts using the measured fiber-plus-system insertion-loss budget, rather than measuring the source flux directly.}

\newtext{The second is the total number of received (detected) photons \(R\), summed across all input states and all three measurement bases. The third is the number of matched-outcome photons \(M\), which we obtain from the same aggregated counts together with the estimated depolarization probability via \(M = R\,(1-p/2)\). This relation is exactly the model identity \(\mathbb{E}[M/T]=s+d^{(x)}/2 = r^{(x)}(1-p^{(x)}/2)\).}

\newtext{Aggregating over bases leaves the link relations intact. Because \(R\) and \(M\) share the same aggregation and the same \(T\), the number of bases summed cancels in the ratios \(R/T\) and \(M/T\). The single-basis relations of~\eqref{eq:measure-ratio} therefore carry over unchanged.}

\newtext{We subtract detector dark counts and background before forming the observables. Such background is independent of the classical-traffic configuration, so it cancels in the success-ratio estimator \(\hat{s}=(2M-R)/T\): a background of \(C\) counts adds \(C\) to \(R\) and \(C/2\) to \(M\), and \(2(C/2)-C=0\). It biases each \(\hat{d}^{(x)}\) only by a common \(+C/T\), which leaves the Raman excess \(d^{(x)}-d^{(0)}\) unaffected.}

\newtext{We compute these quantities for the three depolarization scenarios \(x \in \{0,1,2\}\), giving the triples \((T^{(x)}, R^{(x)}, M^{(x)})\). An auxiliary classical-only run (no quantum signal) yields one further triple, which only the non-coexistence baselines below consult.}

\newtext{Finally, we note that the two headline link quantities reported below---the depolarization probability \(p^{(x)}=d^{(x)}/r^{(x)}=2(R-M)/R\) and the process fidelity, which is a function of \(p^{(x)}\) alone---are exactly invariant to the absolute scale of \(T\). The QLT-versus-BPT comparisons are thus robust to any miscalibration of the launched-photon count. Only the absolute success ratio \(s\), together with the Table~\ref{tab:star-network-parameters} parameters that depend on it, scales with \(T\).}

\textbf{Polarization drift correction.}
Due to fiber birefringence, the polarization reference frame drifts slowly, misaligning the measurement bases with the prepared states by an unknown \(SU(2)\) rotation \(U_t\) that is constant within a batch but varies across batches.
Parameterizing \(U_t\) by three angles, we estimate it per batch by maximum likelihood \newtext{estimation} over the full 36-entry input--measurement table and apply \(U_t^{-1}\) to realign the bases.
This is not circular: \(U_t\) is a purity-preserving frame rotation (\(U_t(I/2)U_t^\dagger=I/2\)), so it cannot absorb the basis-independent depolarization \(p\,I/2\); the unitary and depolarization parameters are separately identifiable in the joint fit, and genuine coherent errors are corrected rather than mislabeled as success.
After realignment, the residual channel is well approximated by the isotropic depolarizing channel of \S\ref{sec:model}.

\textbf{Baseline: Bayesian process tomography (BPT).}
As a baseline, we compare our link tomography estimators against the Bayesian process tomography (BPT) method of Chapman et al.~\cite{chapman2023coexistent}, which reconstructs the full single-qubit quantum process matrix from the 36-entry input-measurement combinations via Bayesian inference.
While BPT provides a complete channel characterization, it is computationally far heavier: it is a member of the broad family of quantum process tomography (QPT) methods~\cite{chuang1997prescription,mohseni2008quantum} that reconstruct an unconstrained single-qubit process matrix (a \(4\times4\) complex matrix with \(d^4-d^2=12\) free real parameters under complete positivity and trace preservation) by drawing samples from a posterior via iterative Bayesian inference; structured variants such as compressed-sensing QPT~\cite{shabani2011efficient} and gate-set tomography~\cite{blumekohout2017demonstration} reduce this cost under sparsity or calibration-robustness assumptions but still target the full process matrix. By contrast, our QLT estimator~\eqref{eq:link-tomography-estimator} exploits the depolarization structure to recover the four physical parameters \((s,d^{(0)},d^{(1)},d^{(2)})\) in closed form, requiring only \(O(1)\) arithmetic on the aggregate counts---no iterative optimization or sampling---thus providing a complete \emph{structured} channel characterization at far lower computational cost, at the price of the modeling assumptions BPT does not require.

\subsection{Link Tomography Analysis}\label{subsec:link-experiment}

We evaluate the link-level tomography estimators from \S\ref{sec:link-tomography} on the coexisting fiber dataset described above, using quantum-signal wavelengths \(\{1530, 1535, 1540, 1542, 1545, 1550, 1555, 1560\}~\text{nm}\).

\paragraph{\newtext{Method comparison}}
We compare four methods, organized into two groups: two non-coexistence baselines that estimate parameters from data collected in isolation (either purely quantum or purely classical traffic), and two coexistence methods that jointly exploit data from all three coexistence scenarios.

The non-coexistence baselines exploit a simple separation: when only the quantum signal is present, the observed depolarization is attributable solely to the fiber and system itself (i.e., \(d^{(0)}\)); when only the classical signal is present, the observed depolarized photons are attributable solely to co- or counter-propagating classical traffic\footnote{To make the parameter interpretation consistent with the coexistence scenarios, the co- and counter-propagation cases here correspond to receiving the (would-be) quantum signal at the same end as the classical signal or at the opposite end, respectively.} (i.e., \(d^{(x)} - d^{(0)}\) for \(x = 1, 2\)).
This separation lets these baselines read off depolarization parameters directly from each isolated measurement, without solving the coupled system in~\eqref{eq:link-tomography-equations}.
Concretely:
(1)~\textbf{Non-Coex}~\cite[Figure~10]{chapman2023coexistent} estimates the classical-traffic-induced component \(d^{(x)} - d^{(0)}\) from the classical-only data and ignores the baseline depolarization \(d^{(0)}\);
(2)~\textbf{Rev Non-Coex} is our revision that adds back the baseline term, yielding \((d^{(x)} - d^{(0)}) + d^{(0)} = d^{(x)}\).

The coexistence methods, in contrast, use measurements from all three coexistence scenarios (no classical signal, co-propagation, and counter-propagation) jointly:
(3)~\textbf{Coex QLT}, our coexistence quantum link tomography estimator (\S\ref{sec:link-tomography}), which solves~\eqref{eq:link-tomography-equations} in closed form using both co- and counter-propagation data; and
(4)~\textbf{Coex BPT}, the Bayesian process tomography baseline of Chapman et al.~\cite{chapman2023coexistent}.

\paragraph{Depolarization probability estimation}
Figure~\ref{fig:link-tomography-depolarization-probability-comparison} shows the estimated depolarization probabilities \(p^{(x)} = d^{(x)} / r^{(x)}\) as a function of quantum-channel wavelength for the three fiber lengths.
Several observations are worth noting.
First, counter-propagation consistently produces greater depolarization than co-propagation, confirming the direction-dependent noise structure predicted by our model \newtext{(\S\ref{subsec:physical-interpretation}, Eqs.~\eqref{eq:raman-co}--\eqref{eq:raman-counter})}.
Second, depolarization increases with fiber length, as expected from the accumulation of Raman scattering noise over longer propagation distances.
Third, our QLT estimates closely match the BPT baseline across all wavelengths and fiber lengths, validating the isotropic-depolarization approximation that remains after frame-drift correction.
Finally, the anomalous feature near 1542~nm is a notch-filter artifact, not a property of the depolarization channel: Chapman et al.~\cite{chapman2023coexistent} attribute it to the notch filter used to remove the 1542.5~nm classical networking laser, whose imperfect isolation lets the laser tails leak through near 1542~nm. It is therefore unrelated to any misspecification of our channel model.

\paragraph{Process fidelity estimation}
We report the \emph{process fidelity} as the Choi--Jamio\l{}kowski state fidelity~\newtext{\cite{nielsen2010quantum}} between the depolarizing channel implied by each estimator and the experimentally reconstructed channel. Writing \(\rho_{\mathcal{P}}\) for the trace-normalized Choi state of a channel \(\mathcal{P}\), an estimate \(\hat{p}^{(x)}\) yields the depolarizing-channel Choi \(\rho_{\mathrm{dep}}(\hat{p}^{(x)})\), and we report the Uhlmann fidelity \(F=\big(\operatorname{Tr}\sqrt{\sqrt{\rho_{\mathrm{dep}}}\,\rho_{\mathrm{exp}}\sqrt{\rho_{\mathrm{dep}}}}\,\big)^2\), maximized over a local input rotation to remove the residual frame drift, with \(\rho_{\mathrm{exp}}\) the Choi state reconstructed by BPT from the same dataset. For an ideal depolarizing channel this reduces to the entanglement/process fidelity \(F=1-\tfrac34\hat{p}^{(x)}\) (not the average gate fidelity \(1-\tfrac12\hat{p}^{(x)}\)); all four curves of Fig.~\ref{fig:link-tomography-process-fidelity-comparison} use this identical definition and the same \(\rho_{\mathrm{exp}}\), so the comparison is on equal footing.
Figure~\ref{fig:link-tomography-process-fidelity-comparison} shows the estimated process fidelity as a function of fiber distance for several wavelengths.
The four methods disentangle three effects. Comparing Non-Coex with Rev Non-Coex isolates the accounting value of the intrinsic baseline \(d^{(0)}\): Non-Coex omits it, producing an optimistically biased fidelity, and adding it back removes a systematic gap that is an information/accounting effect, not a modeling one.
\newtext{Concretely, Non-Coex reads the depolarization off the classical-only run and so estimates \(\hat{d}^{(x)}_{\mathrm{NC}}=d^{(x)}-d^{(0)}\), whereas Rev Non-Coex restores the intrinsic term to give \(d^{(x)}\). Since \(F=1-\tfrac34 p^{(x)}\) with \(p^{(x)}=d^{(x)}/(s+d^{(x)})\), omitting \(d^{(0)}\) raises the reported fidelity by exactly}
\begin{equation}\label{eq:accounting-gap}
    \newtext{\Delta F = \frac{3}{4}\cdot\frac{s\,d^{(0)}}{\bigl(s+d^{(x)}\bigr)\bigl(s+d^{(x)}-d^{(0)}\bigr)} > 0,}
\end{equation}
\newtext{largest at short reach, where \(d^{(x)}\) is small and \(d^{(0)}\) is a larger share of the total. Both methods assume the same channel and the same estimator algebra and differ only in which runs they consult, which is why we call the gap an accounting rather than a modeling effect.}
Comparing Rev Non-Coex with Coex QLT isolates the value of measuring under genuine coexistence rather than additively combining isolated quantum-only and classical-only measurements. Finally, comparing Coex QLT with Coex BPT---both consuming the full coexistence data---isolates the adequacy of the closed-form depolarization model: Coex QLT closely tracks the BPT reference, with only a small residual gap attributable to the non-depolarizing (coherent/asymmetric) component the two-parameter model omits. Together this establishes that the depolarization channel model is an accurate abstraction for process-fidelity prediction in QCNs, while the first two comparisons quantify the separate benefit of jointly exploiting all coexistence scenarios.


\subsection{Star Network Tomography Emulation}\label{subsec:star-experiment}

Lacking a multi-link quantum testbed, we evaluate the star network tomography methods of \S\ref{sec:network-tomography} by composing experimentally derived link-level descriptions---specifically, the per-link channel parameter estimates and the BPT-estimated channel maps---into emulated multi-link paths.
We emulate a three-link star network using link-level estimates from the 192.550~THz ($\approx$1555~nm) channel, with fiber lengths of 0.5~km, 5.0~km, and 15~km.
The hub node connects all three links, and end-to-end statistics for each pair of end nodes are computed by composing two link-level channel descriptions.
We consider two compilation strategies that differ in how the individual link channels are modeled (Figure~\ref{fig:two-consecutive-links}).

\begin{figure}[!t]
    \centering
    \begin{subfigure}[c]{0.4\textwidth}
        \resizebox{\textwidth}{!}{%
            \begin{tikzpicture}[yscale=0.7, every node/.style={font=\fontsize{15}{18}\selectfont}]
                \definecolor{inputcol}{RGB}{40,40,40}
                \definecolor{chonecol}{RGB}{29,106,60}
                \definecolor{chtwocol}{RGB}{45,95,160}
                \definecolor{ramancol}{RGB}{244,162,97}
                \fill[inputcol] (0,0) rectangle (0.75,8);
                \fill[chonecol]  (4.8,0)   rectangle (5.55,3.2);
                \fill[ramancol]  (4.8,3.2) rectangle (5.55,4.6);
                \fill[chtwocol]  (9.2,0)   rectangle (9.95,1.9);
                \fill[ramancol]  (9.2,1.9) rectangle (9.95,2.4);
                \draw[dashed] (0,3.2) -- (4.8,3.2);
                \draw[dashed] (0.75,8) -- (4.8,3.2);
                \draw[dashed] (5.55,4.6) -- (9.2,1.9);
                \draw[dashed] (5.55,1.9) -- (9.2,1.9);
                \node at (2.1,4.8)                              {Loss 1};
                \node[anchor=south]      at (5.175,4.6)  {Raman 1};
                \node       at (2.1,1.6)    {Pass 1};
                \node at (6.77,2.8)                             {Loss 2};
                \node[anchor=south] at (9.5,2.4)   {Raman 2};
                \node       at (6.77,0.95)  {Pass 2};
                \draw[dashed] (0.75,0) -- (4.8,0);
                \draw[dashed] (5.55,0) -- (9.2,0);
                \node[anchor=north] at (0.375,-0.15) {Input};
                \node[anchor=north] at (5.175,-0.15) {Middle};
                \node[anchor=north] at (9.575,-0.15) {Output};
            \end{tikzpicture}%
        }
        \caption{{Illustration of two consecutive links}
        }
    \end{subfigure}
    \begin{subfigure}[c]{0.4\textwidth}
        \centering
        \begin{tabular}{lcc}
            \hline
                                        & Type~I    & Type~II            \\
            \hline
            Loss \(q\)                  & Link Est. & Link Est.          \\
            Raman \(d^{(x)} - d^{(0)}\) & Link Est. & Link Est.          \\
            Pass \(s+d^{(0)}\)          & Link Est. & \emph{Bayesian PT} \\
            \hline
        \end{tabular}
        \caption{Two types of compilation}
    \end{subfigure}
    \caption{Star network tomography compilation strategies. Each end-to-end path traverses two links through the hub. Type~I uses idealized depolarization channels; Type~II uses BPT-estimated channel maps.}
    \vspace{-0.2cm}
    \label{fig:two-consecutive-links}
\end{figure}

\paragraph{Compilation Type~I (Idealized Depolarization)}
In the first compilation, each link is represented by the idealized depolarization channel from \S\ref{sec:model}.
The channel parameters \((s, d^{(0)}, d^{(1)}, d^{(2)})\) for each link are taken from the link tomography estimates of \S\ref{subsec:link-experiment}.
Given these parameters, end-to-end measurement statistics are generated by composing two depolarization channels according to~\eqref{eq:network-tomography-measurement}: a photon traverses link~\(k\) from the sender to the hub and then link~\(\ell\) from the hub to the receiver, experiencing the appropriate co- or counter-propagation depolarization on each segment depending on the classical traffic direction.
This compilation provides a controlled setting where the channel model is exactly matched, isolating the estimation accuracy of the network tomography algorithms.

\paragraph{Compilation Type~II (BPT-Based Channel Maps)}
The second compilation uses the full quantum process matrices obtained from Bayesian process tomography (BPT) on the experimental data, rather than the idealized depolarization model.
Specifically, for each link and each classical-traffic configuration, we use the BPT-estimated channel map \(\mathcal{P}\) to generate end-to-end statistics. Given an input population \(T\), the expected number of received photons is
\begin{align}
    R & = T \, \bigl(s_k + d_k^{(x(k,1))}\bigr)\bigl(s_\ell + d_\ell^{(x(\ell,2))}\bigr),
\end{align}
with \(s_k\) and \(d_k^{(x)}\) taken from the link tomography estimates and each factor evaluated at the propagation scenario actually seen on that hop (cf.~\eqref{eq:group2}),
and the number of matched outcomes \(M\) is obtained by drawing each of the \(R\) received photons' outcome from the composed BPT map: we sample \(\ket{0}\) with probability \(\bra{0}\mathcal{P}_2 \circ \mathcal{P}_1 (\ketbra{0}{0})\ket{0}\), where \(\mathcal{P}_1, \mathcal{P}_2\) are the BPT-estimated qubit process maps for the two links.
Note that the received \emph{count} \(R\) follows the depolarization count model (loss and Raman injection are not represented by the \(2\times2\) process maps), while the per-photon matched fraction \(M/R\) follows the BPT maps; thus Type~II stresses model mismatch in the polarization/process-fidelity channel---and hence in \(s_k\), and through \(d_k = r_k - s_k\) in the depolarization estimates---while the received ratios \(r_k = R/T\) are governed by the depolarization count model.

\newtext{\paragraph{Simulation statistics} The counts are drawn as follows. \(T\) is deterministic---the batch size, \(10^5\) or \(10^{10}\). Given \(T\), each hop applies the binomial cascade implied by \S\ref{sec:model}: the first hop draws the intact and depolarized populations as two \emph{independent} binomials, \(\mathrm{Binom}(T,s_k)\) and \(\mathrm{Binom}(T,d_k^{(x(k,1))})\)---which is what lets \(R\) exceed \(T\) under Raman injection---and the second hop splits each of them again by \(s_\ell\) and \(d_\ell^{(x(\ell,2))}\). Then \(R\) is the sum of the four resulting populations and \(M=R^{\mathrm{ss}}+\mathrm{Binom}(R-R^{\mathrm{ss}},\tfrac12)\), with \(R^{\mathrm{ss}}\) the doubly-intact population, giving \(\mathbb{E}[R]=T r_k r_\ell\) and \(\mathbb{E}[2M-R]=T s_k s_\ell\) as required by~\eqref{eq:group2} and~\eqref{eq:group1}. For \(T>10^{9}\) each binomial is replaced by \(\mathcal{N}(np,np(1-p))\), rounded and floored at zero---the ``normal-distribution approximation'' of Fig.~\ref{fig:star-network-tomography-parameter-estimation-large-sample-size}. Type~II compilation changes only \(M\): the count \(R\) is drawn exactly as above, since the \(2\times2\) BPT maps carry no loss or injection information, while each received photon is propagated through the composed maps in Kraus form, \(\rho_{\mathrm{out}}=\mathcal{P}_2(\mathcal{P}_1(\ketbra{0}{0}))\), and matched with probability \(\bra{0}\rho_{\mathrm{out}}\ket{0}\). Since the implementation cycles deterministically through the BPT posterior samples, \(M\mid R\) is Poisson-binomial rather than binomial, and the error bars reflect shot noise at the posterior-averaged map rather than posterior uncertainty.}

\paragraph{Network tomography results}

\begin{figure*}
    \centering
    \begin{subfigure}{0.65\columnwidth}
        \centering
        \includegraphics[width=\columnwidth]{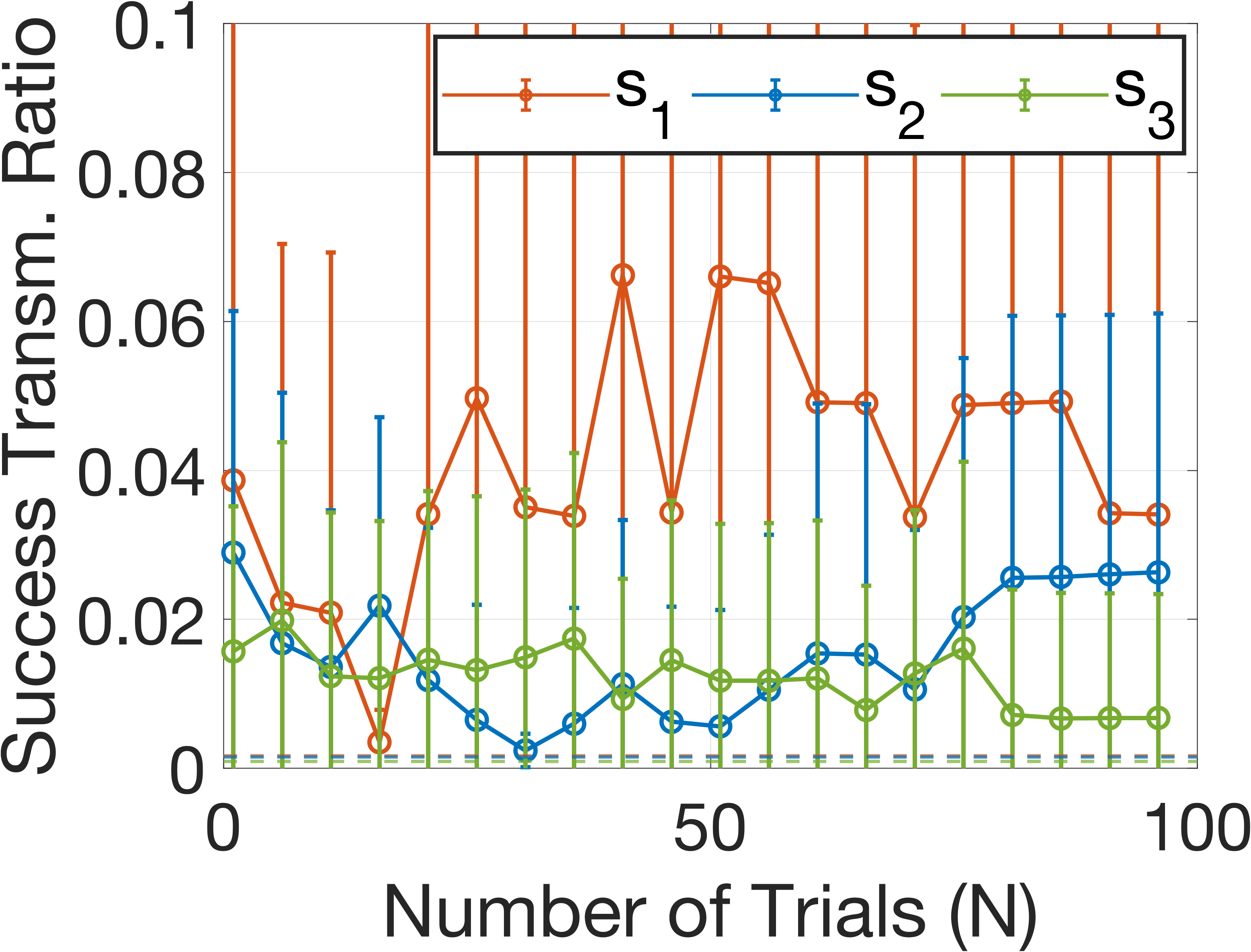}
        \caption{Per-link success ratios \(s_k\)}
    \end{subfigure}
    \begin{subfigure}{0.65\columnwidth}
        \centering
        \includegraphics[width=\columnwidth]{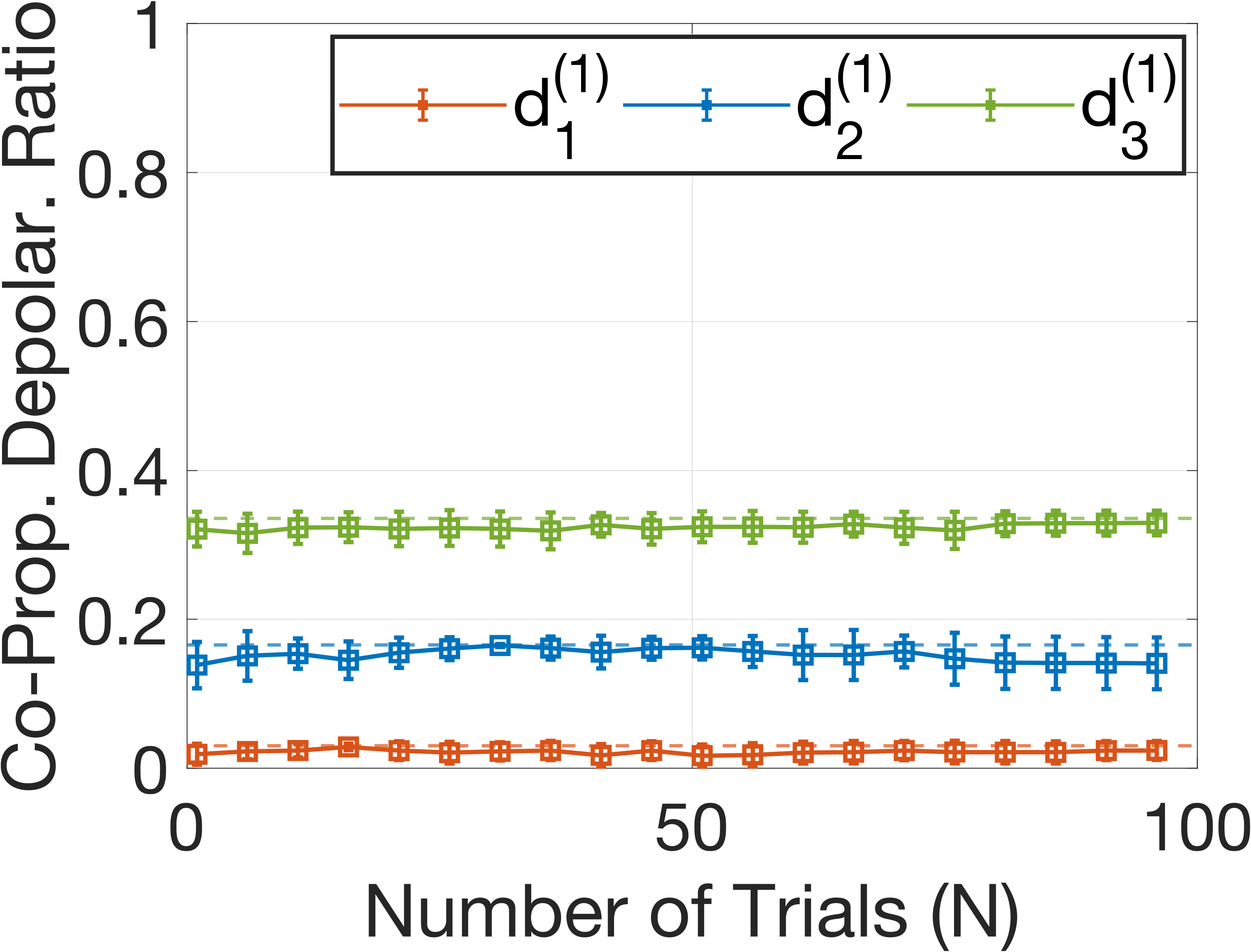}
        \caption{Per-link co-prop.\ depolar. ratios \(d_k^{(1)}\)}
    \end{subfigure}
    \begin{subfigure}{0.65\columnwidth}
        \centering
        \includegraphics[width=\columnwidth]{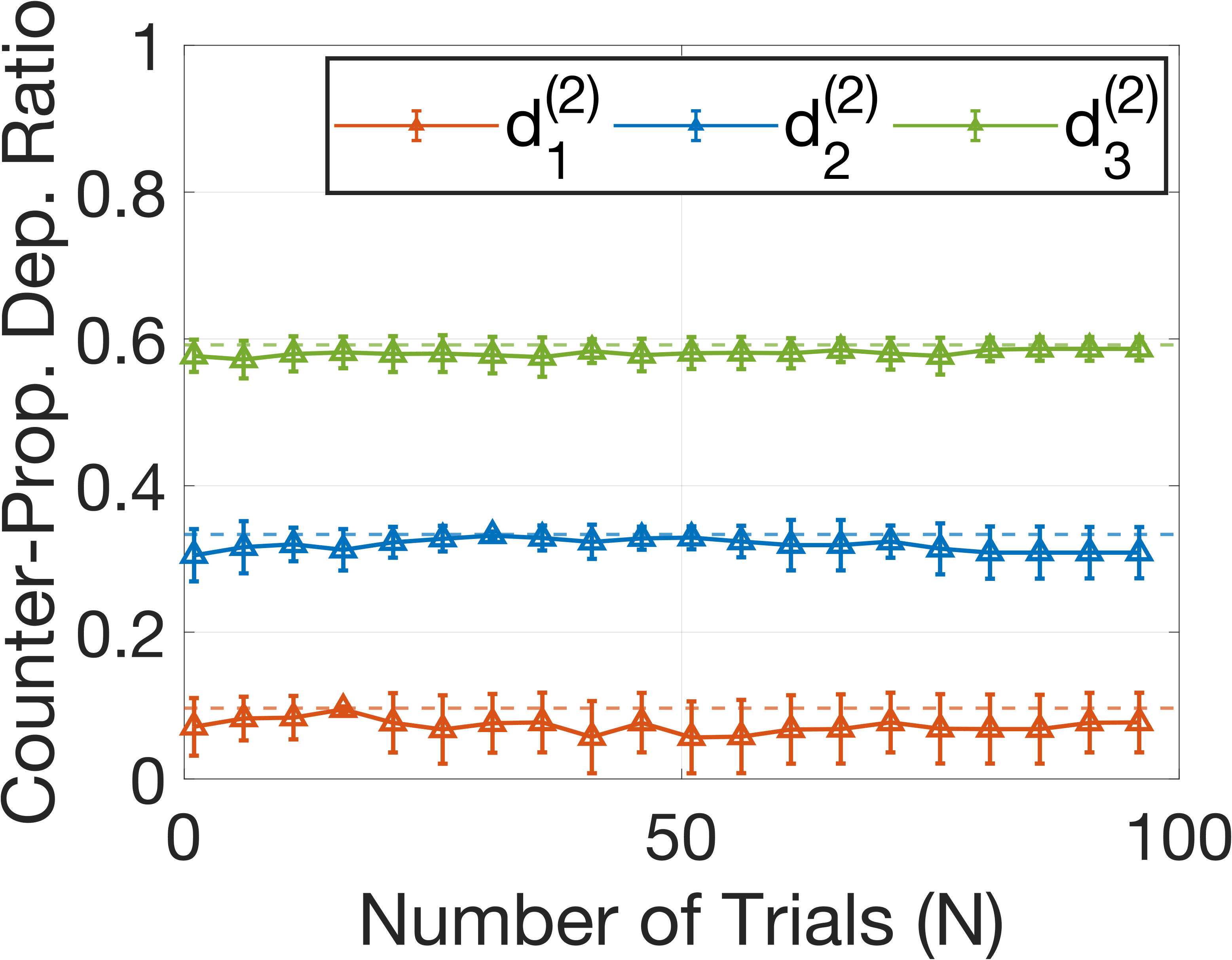}
        \caption{Per-link counter-prop.\ depolar. ratios \(d_k^{(2)}\)}
    \end{subfigure}
    \caption{{Estimated channel parameters under Type~I compilation (sample size \(10^5\)).} \newtext{Solid markers are the star-network estimates as a function of the cumulative number of trials; horizontal dashed lines mark the ground-truth link parameters of Table~\ref{tab:star-network-parameters}; error bars show \(\pm1\) sample standard deviation across 10 independent macro-replicates, each aggregating 100 batches of \(T\) photons.}}
    \label{fig:star-network-tomography-parameter-estimation-small-sample-size}
\end{figure*}

\begin{figure*}
    \centering
    \begin{subfigure}{0.65\columnwidth}
        \centering
        \includegraphics[width=\columnwidth]{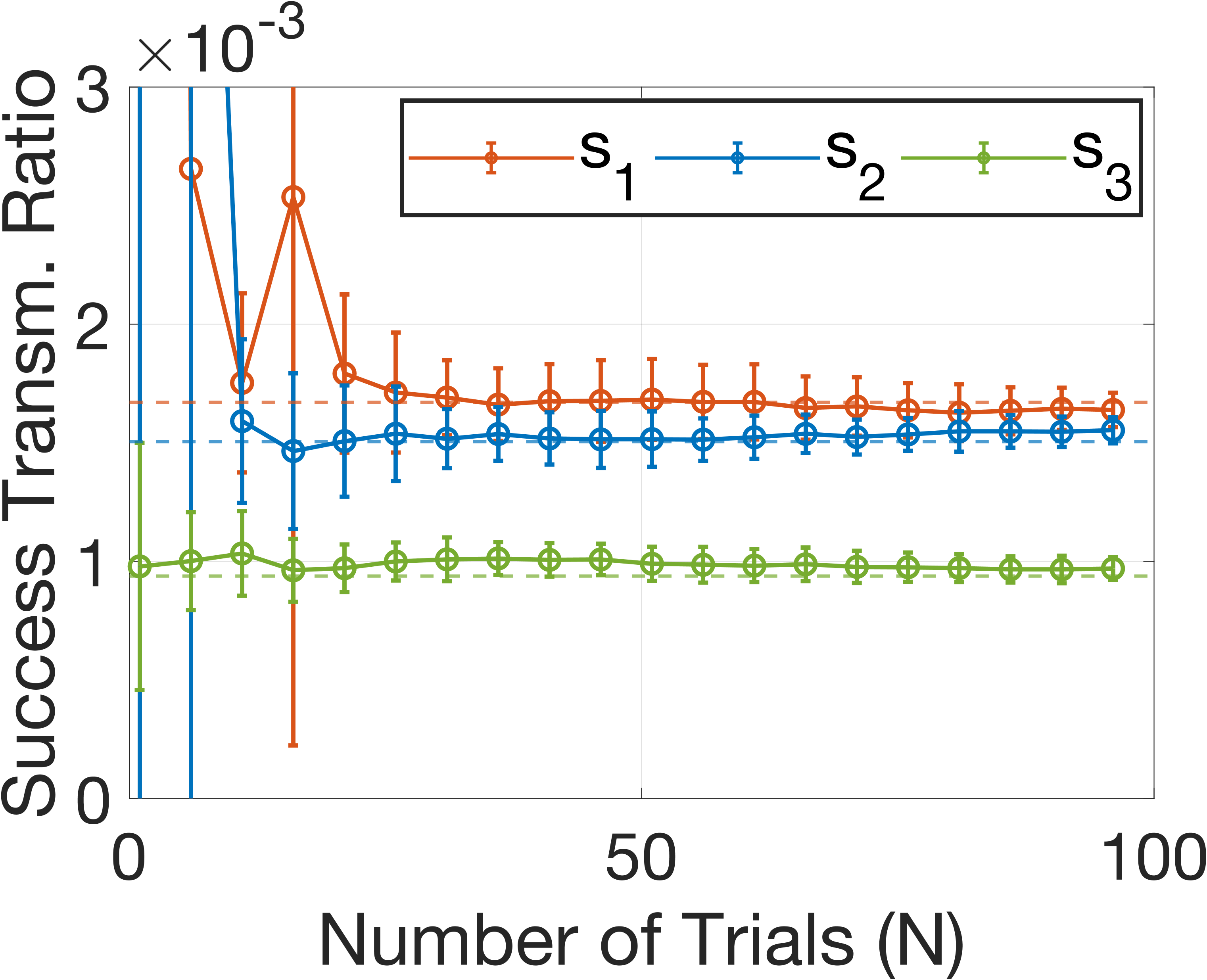}
        \caption{Per-link success ratios \(s_k\)}
    \end{subfigure}
    \begin{subfigure}{0.65\columnwidth}
        \centering
        \includegraphics[width=\columnwidth]{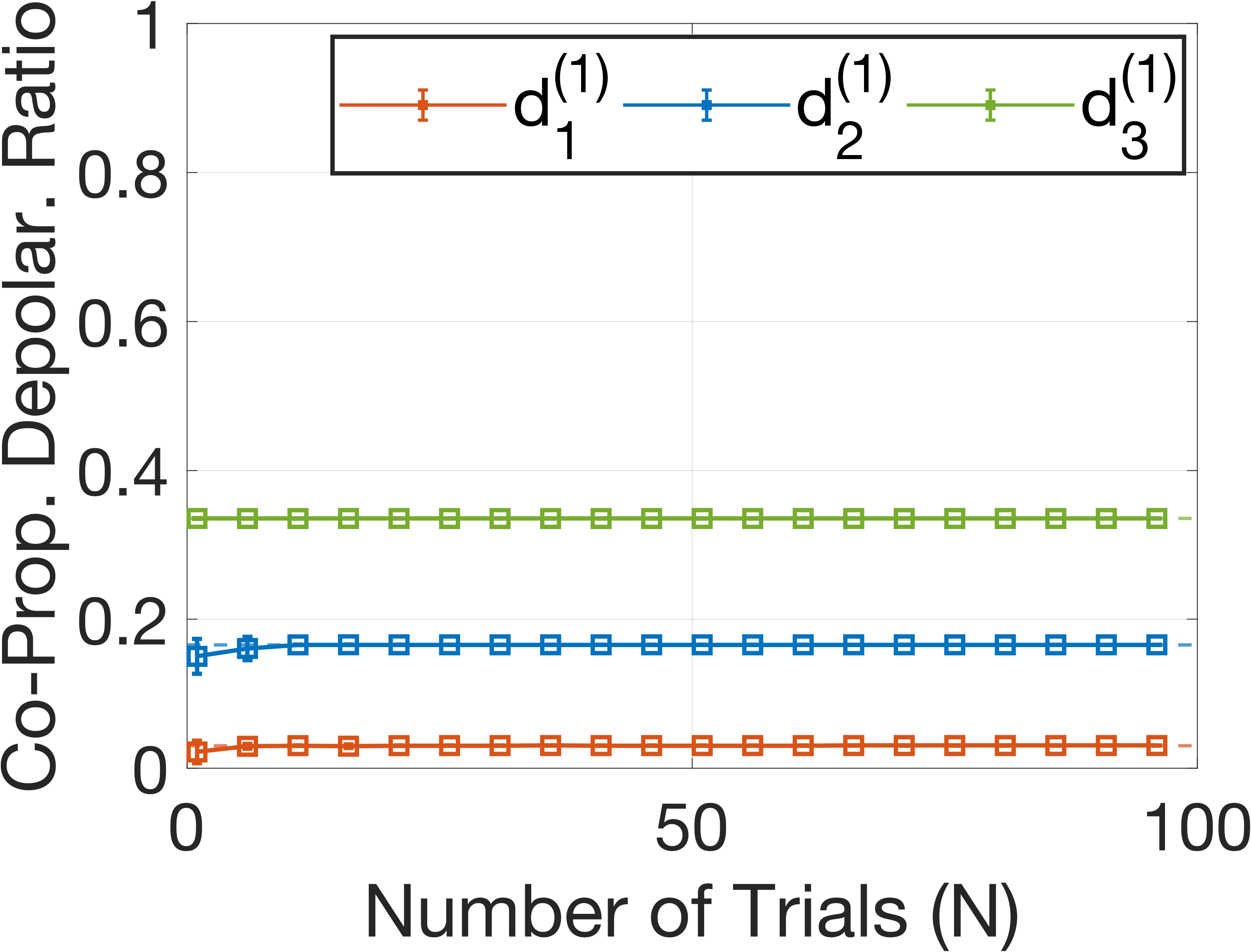}
        \caption{Per-link co-prop.\ depolar. ratios \(d_k^{(1)}\)}
    \end{subfigure}
    \begin{subfigure}{0.65\columnwidth}
        \centering
        \includegraphics[width=\columnwidth]{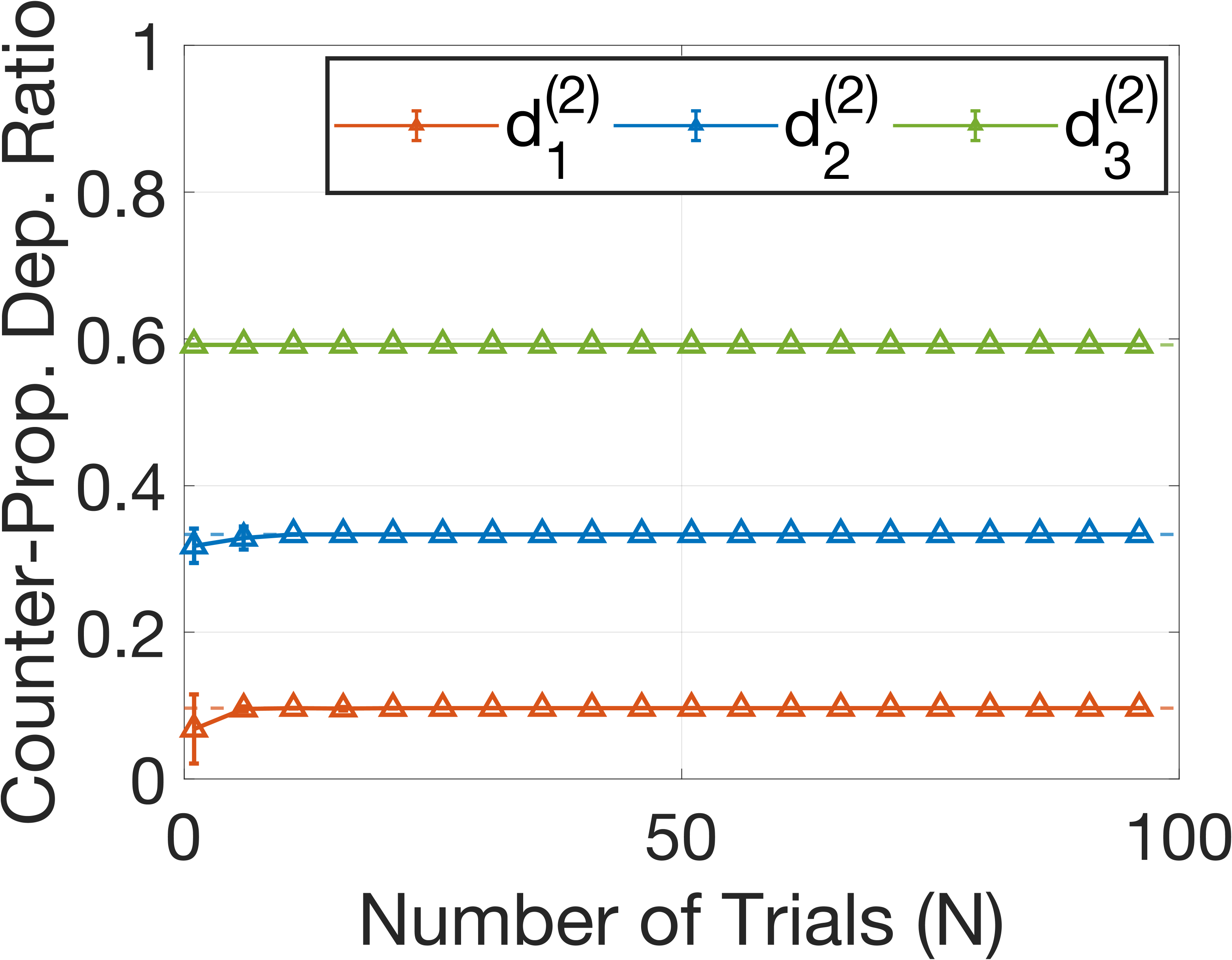}
        \caption{Per-link counter-prop.\ depolar. ratios \(d_k^{(2)}\)}
    \end{subfigure}
    \caption{{Estimated channel parameters under Type~I compilation (sample size \(10^{10}\), with normal-distribution approximation).} \newtext{Dashed lines and error bars are as in Fig.~\ref{fig:star-network-tomography-parameter-estimation-small-sample-size}: dashed lines mark the Table~\ref{tab:star-network-parameters} ground truth, and error bars show \(\pm1\) sample standard deviation across 10 macro-replicates. Note the narrower success-ratio axis relative to Fig.~\ref{fig:star-network-tomography-parameter-estimation-small-sample-size}.}}
    \label{fig:star-network-tomography-parameter-estimation-large-sample-size}
\end{figure*}

\begin{figure*}
    \centering
    \begin{subfigure}{0.65\columnwidth}
        \centering
        \includegraphics[width=\columnwidth]{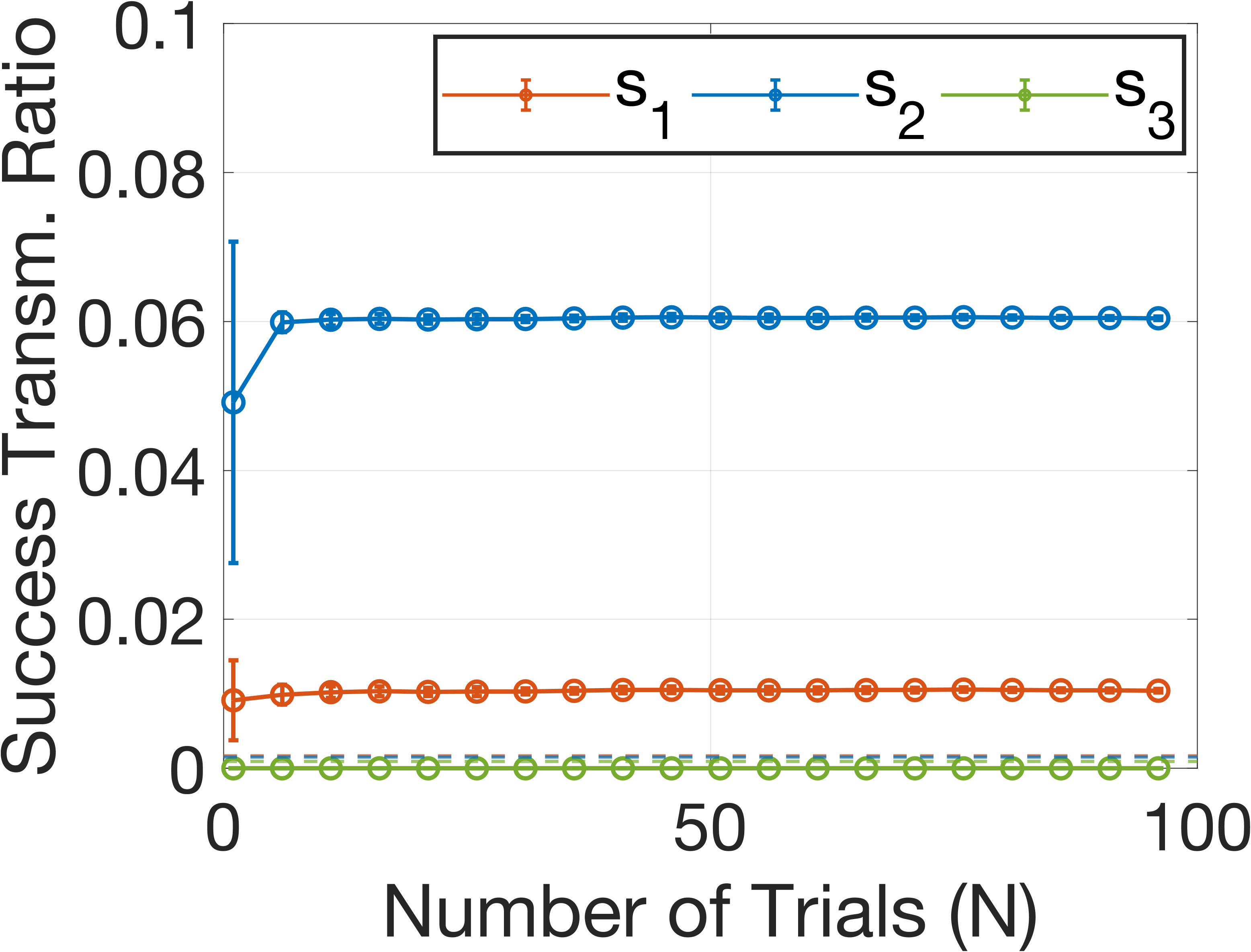}
        \caption{Per-link success ratios \(s_k\)}
    \end{subfigure}
    \begin{subfigure}{0.65\columnwidth}
        \centering
        \includegraphics[width=\columnwidth]{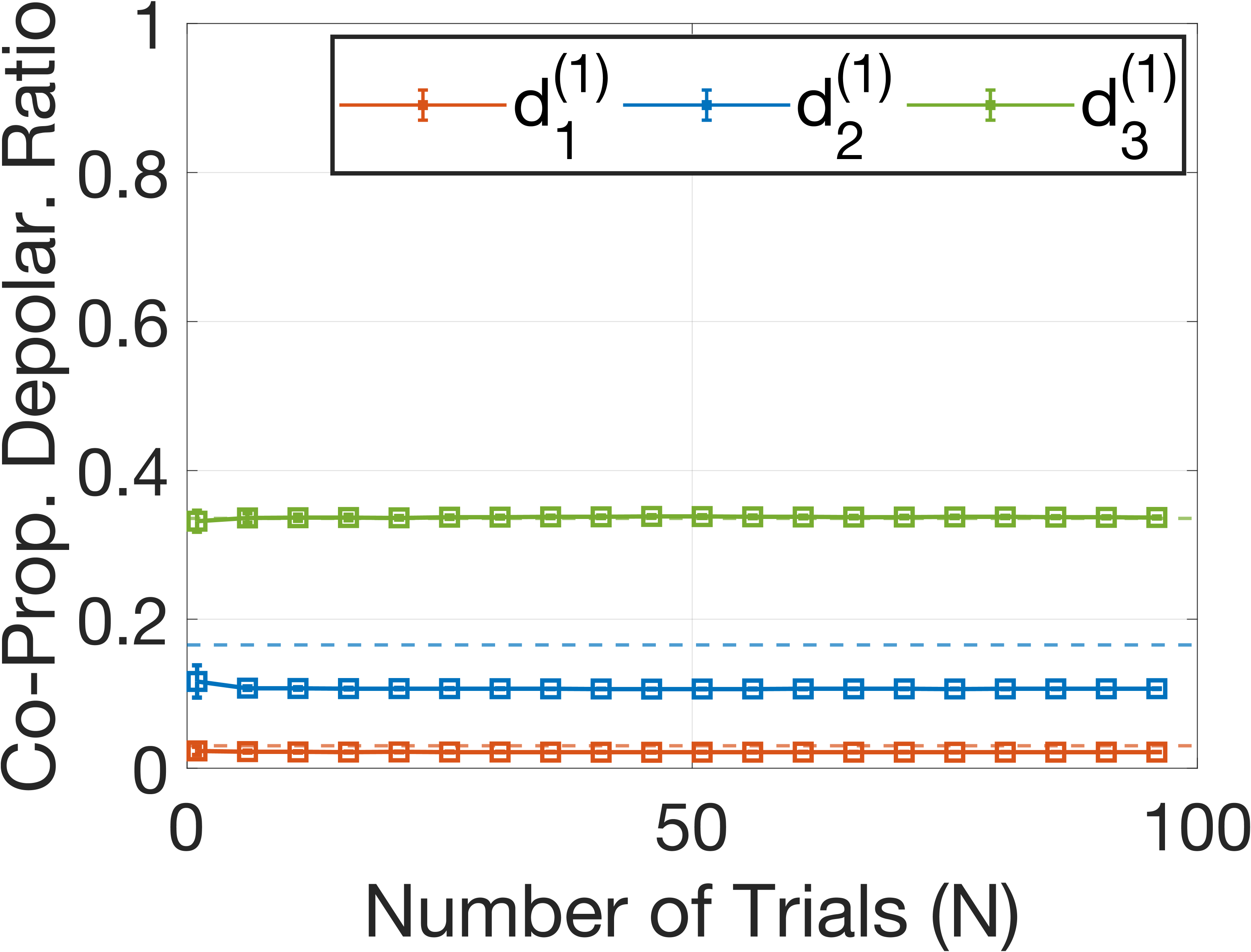}
        \caption{Per-link co-prop.\ depolar. ratios \(d_k^{(1)}\)}
    \end{subfigure}
    \begin{subfigure}{0.65\columnwidth}
        \centering
        \includegraphics[width=\columnwidth]{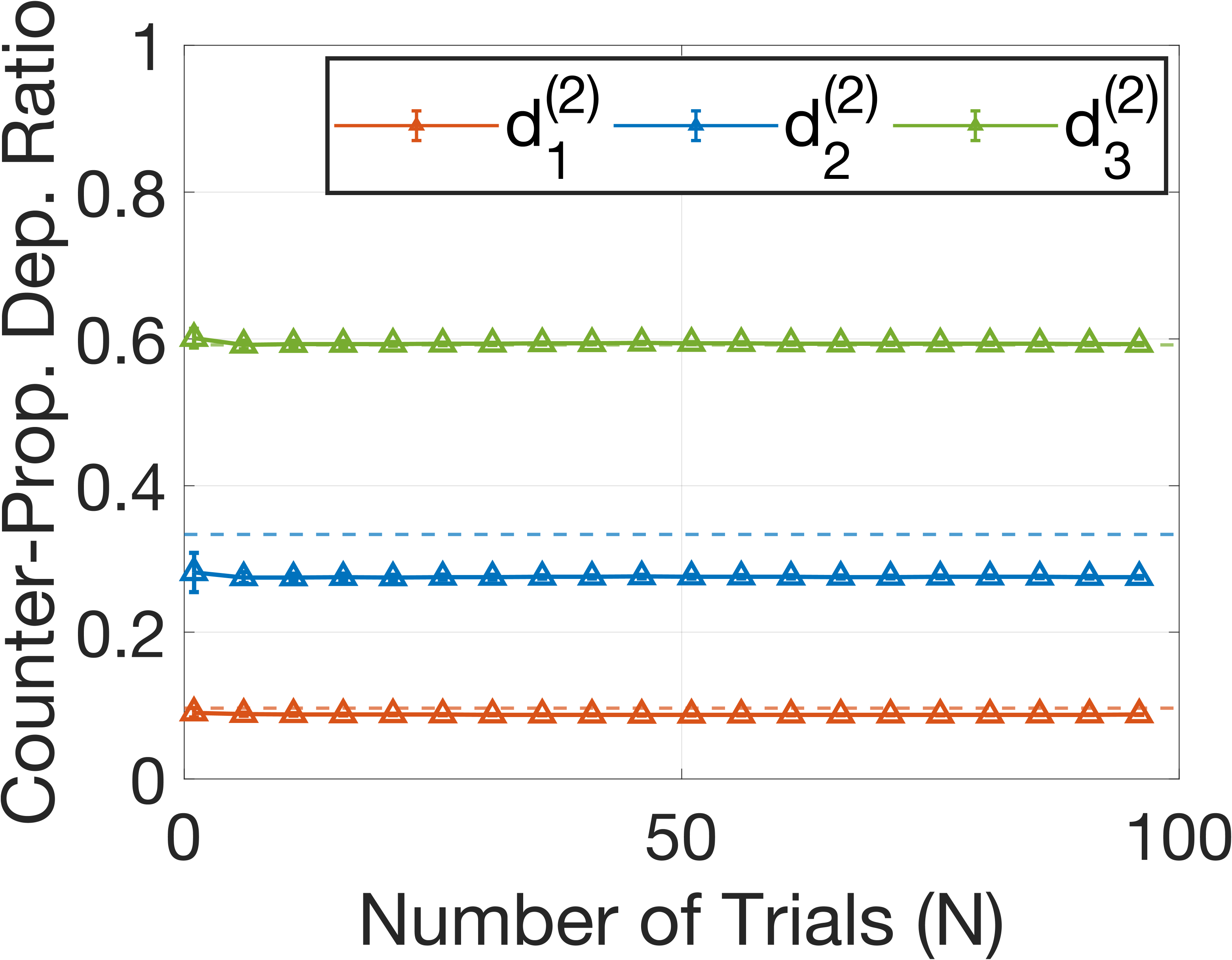}
        \caption{Per-link counter-prop.\ depolar. ratios \(d_k^{(2)}\)}
    \end{subfigure}
    \caption{{Estimated channel parameters under Type~II compilation (sample size \(10^{5}\)).} \newtext{Dashed lines and error bars are as in Fig.~\ref{fig:star-network-tomography-parameter-estimation-small-sample-size}. Because the ground truth (dashed) is set by the idealized depolarization model while the data are generated from the BPT channel maps, the residual offset visible for the 5-km link is the model mismatch discussed in the text, not an estimator failure.}}
    \label{fig:star-network-tomography-parameter-estimation-approximate-channel}
\end{figure*}

\begin{table}
    \centering
    \caption{{Channel parameters at 192.550~THz ($\approx$1555~nm)} for the three fiber lengths (dashed lines in Figs.~\ref{fig:star-network-tomography-parameter-estimation-small-sample-size}--\ref{fig:star-network-tomography-parameter-estimation-approximate-channel}).}
    \label{tab:star-network-parameters}
    \begin{tabular}{@{}lccc@{}}
        \toprule
        Channel & Success ratio & \multicolumn{2}{c}{Depolarization}                              \\
        \cmidrule(lr){3-4}
        length  & \(s\)         & Co-prop.\ \(d^{(1)}\)              & Counter-prop.\ \(d^{(2)}\) \\
        \midrule
        0.5~km  & 0.0017        & 0.0306                             & 0.0967                     \\
        5~km    & 0.0015        & 0.1659                             & \newtext{0.3334}           \\
        15~km   & 0.0009        & \newtext{0.3355}                   & 0.5916                     \\
        \bottomrule
    \end{tabular}
    \vspace{-0.2cm}
\end{table}

Table~\ref{tab:star-network-parameters} lists the link-level parameters used to configure the three-link star network.
To recover all six per-link depolarization parameters, the emulation realizes the two-configuration protocol of \S\ref{sec:network-tomography}: an original batch and a complementary batch in which the classical-signal direction on link~3 is flipped (interchanging \(d_3^{(1)}\) and \(d_3^{(2)}\) between the two probe directions, cf.~\eqref{eq:combined-network-tomography-equation-original}--\eqref{eq:combined-network-tomography-equation-alternative}), and the combined system is solved in closed form. Without the flip only five of the six depolarization equations are independent (Lemma~\ref{lem:star-identifiability}), so the six per-link depolarization parameters shown below could not be identified.
We apply the star network estimators of \S\ref{sec:network-tomography} to the emulated end-to-end measurements and examine convergence as a function of sample size.
For reference, the parameter values from Table~\ref{tab:star-network-parameters} are plotted as dashed lines in the figures below.
In these figures, each \emph{trial} corresponds to one emulation run with a fixed batch sample size (e.g., \(10^5\) or \(10^{10}\) photon transmissions); the \(x\)-axis shows the cumulative number of consecutive trials used to refine the estimates.
Error bars show \(\pm1\) sample standard deviation across 10 independent macro-replicates, each aggregating 100 batches of \(T\) photons.
As the number of trials grows, more measurement data is aggregated, and the link-level estimators converge more closely to the true parameters.

Figure~\ref{fig:star-network-tomography-parameter-estimation-small-sample-size} shows results under Type~I compilation with a batch sample size of \(10^5\): the success ratio \(s\) does not reliably converge (due to its low magnitude, \(\sim 10^{-3}\)), but the depolarization parameters \(d^{(1)}\) and \(d^{(2)}\) do.
Increasing the batch sample size to \(10^{10}\) (Figure~\ref{fig:star-network-tomography-parameter-estimation-large-sample-size}) substantially improves stability: the range of the success-ratio estimates narrows from \([0,0.1]\) to \([0,0.003]\).
Figure~\ref{fig:star-network-tomography-parameter-estimation-approximate-channel} shows Type~II compilation results using BPT-based channel maps with a batch sample size of \(10^5\), illustrating the estimator's behavior under more realistic channel conditions.
Notably, under Type~II compilation, a batch sample size of \(10^5\) already suffices for stable convergence.
However, because the BPT-predicted matched fraction for the 5-km link differs from the depolarization model's \(s+d^{(x)}/2\) (consistent with the imperfect process fidelity in Figure~\ref{fig:link-tomography-process-fidelity-comparison}), the estimates for that link, while stably converging, exhibit a persistent bias relative to the ground truth set by the idealized depolarization model; this mismatch enters through \(2M-R\) (hence \(s_k\)) and propagates to the depolarization estimates via \(d_k = r_k - s_k\).

\newtext{This also explains \emph{which} parameters carry the visible bias. Under Type~II only \(M\) is perturbed, so \(\hat{r}_k^{(x)}\)---which consumes \(R/T\) alone---stays unbiased to leading order and the mismatch is carried by \(\hat{s}_k\); since \(\hat{d}_k^{(x)}=\hat{r}_k^{(x)}-\hat{s}_k\), the displacement \(\operatorname{Bias}(\hat{d}_k^{(x)})=-\operatorname{Bias}(\hat{s}_k)\) is \emph{independent of \(x\)}. The bias is thus a per-\emph{link} effect: on the 5-km link \(\hat{s}_2\approx0.06\) against a truth of \(0.0015\), and both of its depolarization estimates are displaced downward by that same \(\approx0.06\); the 15-km link carries the opposite sign, since there \(\hat{s}_3\) sits at the estimator's positivity floor. Because the displacement is shared between the two modes, the \emph{relative} bias is larger for the smaller parameter (\(\approx36\%\) for \(d_2^{(1)}\) vs.\ \(\approx18\%\) for \(d_2^{(2)}\)), whereas counter-propagation carries the larger error \emph{bars} since \(\operatorname{Var}(\hat{d}^{(x)})\) grows with \(d^{(x)}\)---bias and variance rank the two modes oppositely.}


\paragraph{Why sample size matters}
The estimator~\eqref{eq:parameter-s-estimation} relies on the difference \(2M-R\), whose expectation is the product of two per-link success ratios, \(\mathbb{E}[2M_{k,\ell}-R_{k,\ell}] = T\,s_k s_\ell\), of order \(10^{-6}\,T\) in the coexistence regime. Because each received photon contributes a value in \(\{-1,+1\}\) to \(2M-R\) (\(+1\) if intact, \(\pm1\) with equal probability if depolarized), its variance is at most the received count, \(\operatorname{Var}(2M-R)\le \mathbb{E}[R]=T\,r_k r_\ell\) with \(r=s+d\). A Gaussian approximation then gives \(\Pr[2M-R<0]\approx\Phi\!\big(-\sqrt{T}\,s_k s_\ell/\sqrt{r_k r_\ell}\big)\), so the square-root estimator is reliable only when this tail is small, i.e.\ when \(T\gtrsim r_k r_\ell/(s_k s_\ell)^2\). For our star parameters this evaluates to \(\approx0.49\) at \(T=10^5\) (the success ratio behaves like a coin flip, explaining its failure to converge in Fig.~\ref{fig:star-network-tomography-parameter-estimation-small-sample-size}) and to \(\approx0\) by \(T=10^{10}\) (explaining the narrowing of the estimate range), with the longest, highest-loss link setting the binding threshold. This converts the qualitative observation into a falsifiable, figure-matching prediction and explains why coexistence network tomography needs sample sizes orders of magnitude larger than link-level estimation. A heuristic absolute value \(|2M-R|\) helps Type~I compilation but worsens the more realistic Type~II; the principled small-sample remedy is a non-negativity-constrained estimator on the binomial likelihood. In our implementation the moment ratio is clamped to the physical range and \(2M-R\) is floored at machine epsilon to guard against these sub-zero excursions, which the quoted success-ratio axes reflect.

\newtext{To quantify this, we re-ran the Type~I emulation with the two guards applied to \emph{identical} generated data (5 seeds \(\times\) 2000 macro-replicates). At \(T=10^5\) about half the symmetrized \(2M-R\) values are non-positive, and after 100 batches the floored estimator is \(\approx6.5\times\) worse in \(\mathrm{RMSE}(\hat{s})\) and \(\approx3.6\times\) worse in \(\mathrm{RMSE}(\hat{d})\) than the \(|2M-R|\) variant. At \(T=10^{10}\) over 100 batches nothing is floored and the two variants agree to three significant figures. The harm is not the clipped values collapsing toward zero, which costs at most the true \(s\sim10^{-3}\), but the reverse: a clipped value in the \emph{denominator} of~\eqref{eq:parameter-s-estimation} makes the ratio saturate the physical-range clamp and drives \(\hat{s}\) to \(\approx0.16\). (These figures are for Type~I; the Type~II variant additionally requires the BPT posterior samples.)}

To put the scale in perspective, resolving \(s\) to within its own magnitude requires \(\sim\!10^{10}\) input photons per end-pair; at an illustrative clocked source rate of \(10^7\)--\(10^8\) transmissions per second this is minutes of integration per end-pair, but because the requirement grows super-exponentially with depth (Fig.~\ref{fig:exp-topology-scaling}) the depth-5 case (\(2.5\times10^{11}\) probes) already pushes acquisition to hours. Recovering the small coexistence success-ratio signal---rather than the comparatively easy depolarization parameters---is therefore the binding practical bottleneck.

\subsection{General-Topology Tomography Simulation}\label{subsec:general-topology-experiment}

\begin{figure}[!t]
    \centering
    \includegraphics[width=\columnwidth]{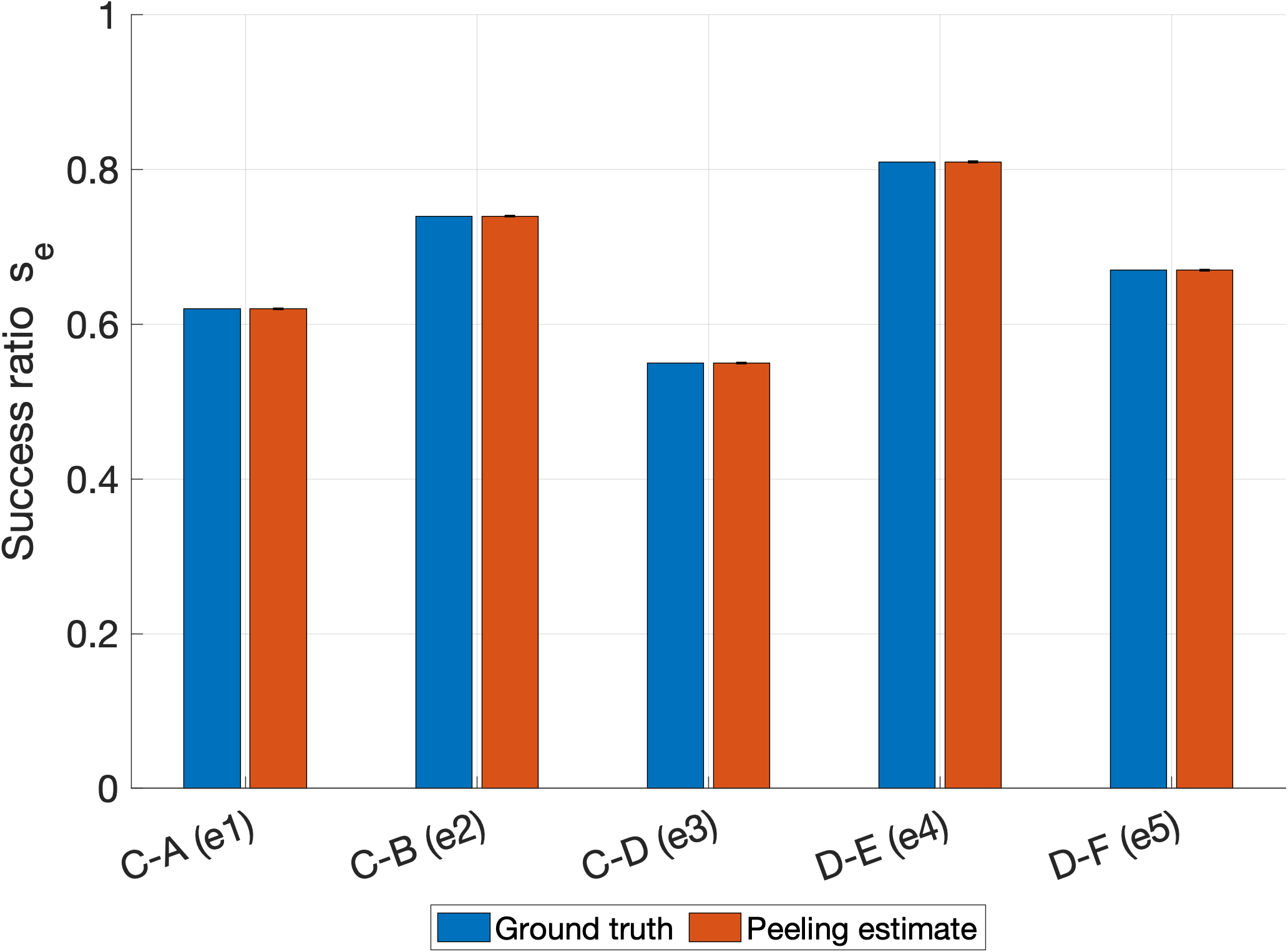}
    \caption{Per-link success-ratio recovery on a five-link tree
            ($A\!-\!C\!-\!B$, $C\!-\!D$, $D\!-\!E$, $D\!-\!F$; leaves $A,B,E,F$) using only
            leaf-to-leaf probes. Leaf links are estimated by
            $\hat{s}_e=\sqrt{Q_{v,j}Q_{v,k}/Q_{j,k}}$ and the internal link $C\!-\!D$ is
            recovered by peeling. Estimates (sample size $T=5\times10^{6}$, $40$ repeats)
            match ground truth for every link, including the internal one; error bars are
            smaller than the markers.}
    \label{fig:exp-topology-peeling}
\end{figure}

\begin{figure*}[!t]
    \centering
    \includegraphics[width=0.92\textwidth]{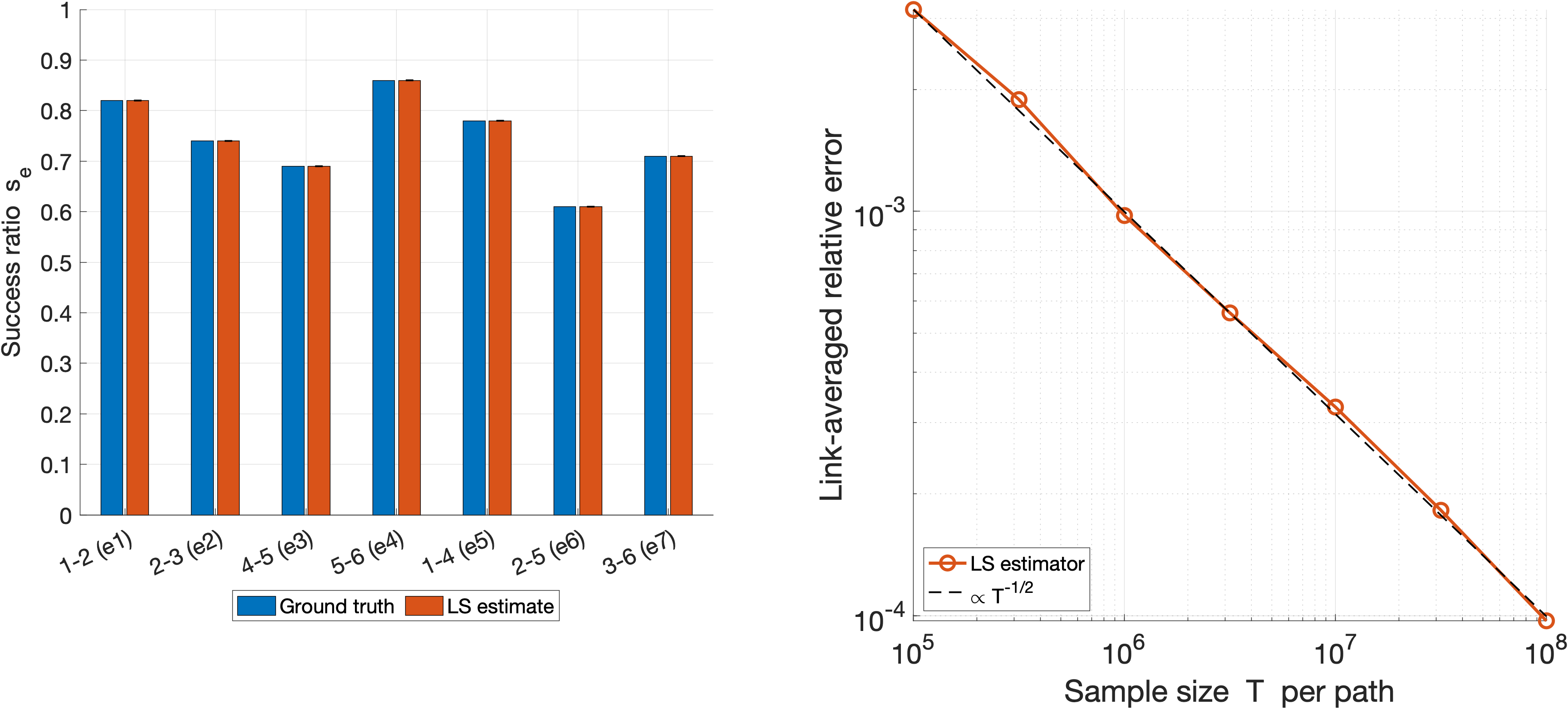}
    \caption{General-graph least-squares recovery on a cyclic mesh
            ($2\times3$ ladder/grid: top path $1\!-\!2\!-\!3$, bottom path $4\!-\!5\!-\!6$, rungs
            $1\!-\!4$, $2\!-\!5$, $3\!-\!6$; seven links, two $4$-cycles, no leaves). All per-link success
            ratios are inferred from corner-to-corner probes only (terminals $1,3,4,6$;
            nodes $2,5$ are relays) by solving the log-linearized system
            $A\mathbf{x}=\mathbf{b}$, $x_e=\log s_e$, in the least-squares sense over twelve
            paths (incidence matrix rank $7$, full column rank). \emph{Left:} recovered
            $s_e$ versus ground truth at $T=10^{7}$ per path. \emph{Right:} link-averaged
            relative error decreasing as $T^{-1/2}$.}
    \label{fig:exp-mesh-regression}
\end{figure*}

\begin{figure*}[!t]
    \centering
    \includegraphics[width=0.92\textwidth]{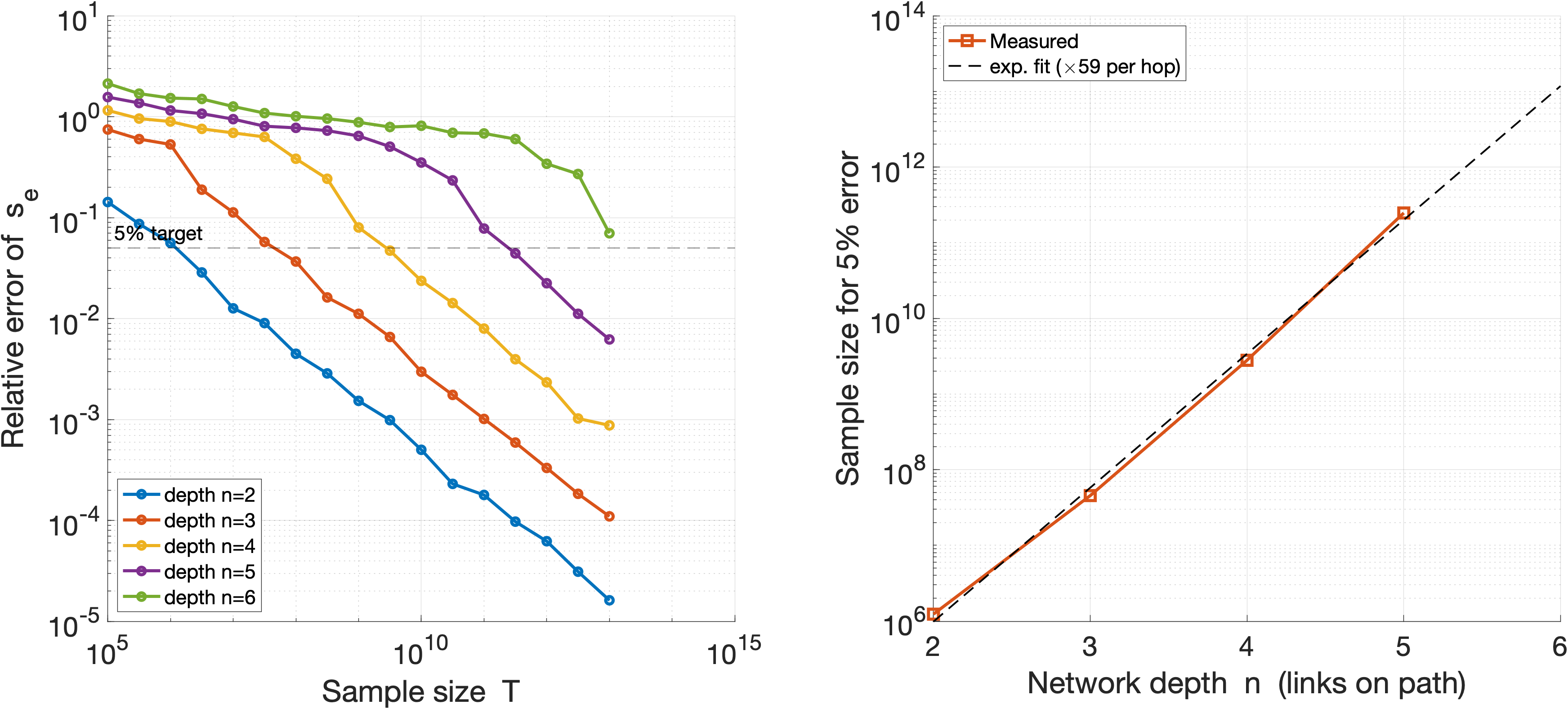}
    \caption{Sample-size requirements grow super-exponentially with network
            depth. \emph{Left:} relative error of the per-link estimator $\hat{s}$ versus
            sample size $T$ for identical-link chains of depth $n=2,\dots,6$ ($s_e=0.05$).
            \emph{Right:} the sample size needed to reach $5\%$ relative error versus depth
            ($1.2\times10^{6}$, $4.5\times10^{7}$, $2.8\times10^{9}$, $2.5\times10^{11}$
            probes for $n=2,\dots,5$). The consecutive per-hop multipliers \emph{increase}
            ($38\times,62\times,89\times$), consistent with the closed form
            $T_{\mathrm{req}}(n)=[s^n(1-s^n)+((s{+}d)^n-s^n)]/(n^2 s^{2n}\epsilon^2)$, whose
            per-hop ratio tends to $(s{+}d)/s^2$ (here $=100$); the $\approx\!59\times$
            quoted elsewhere is the geometric-mean slope of the semi-log fit at the
            illustrative $s_e=0.05$.}
    \label{fig:exp-topology-scaling}
\end{figure*}

    We finally validate the general-topology constructions of
    \S\ref{sec:general-topology}---the tree-peeling estimator, the general-graph
    least-squares solver, and the sample-size scaling---with Monte-Carlo simulations.
    These simulations reuse the multi-link path-composition sampler of the
    star-network emulation (\S\ref{subsec:star-experiment}), but with all per-link
    parameters assigned synthetically rather than taken from testbed measurements,
    so all data here are generated entirely numerically.
    These Monte-Carlo simulations are deliberately self-consistent: the path sampler composes the same depolarization channels the estimators invert, so they verify estimator correctness, numerical identifiability, and sample-complexity scaling under a correctly-specified model. They are not a test of model fidelity, which is established only at the link level (Figs.~\ref{fig:link-tomography-depolarization-probability-comparison},~\ref{fig:link-tomography-process-fidelity-comparison}) and stressed under Type~II compilation.

    \paragraph{Peeling on a tree}
    We construct a five-link tree with two internal nodes and four leaves and assign
    distinct per-link success ratios. Using only end-to-end probes between leaves, the
    symmetrized success ratios $Q=(2M-R)/T$ identify each leaf link through
    $\hat{s}_e=\sqrt{Q_{v,j}Q_{v,k}/Q_{j,k}}$~\eqref{eq:parameter-s-estimation}, and
    the interior link $C\!-\!D$ is then recovered by peeling its leaf contributions
    from a leaf-to-leaf path. Figure~\ref{fig:exp-topology-peeling} shows that all
    five per-link parameters are recovered accurately at $T=5\times10^{6}$ probes per
    path, confirming the correctness of the peeling construction for tree topologies.

    \paragraph{General graphs via least squares}
    For topologies with cycles the peeling recursion no longer applies, since no link
    is incident to a leaf. We instead solve the log-linearized system of
    \S\ref{sec:general-topology}: taking
    logarithms of \eqref{eq:general-S} turns each probe into the linear constraint
    $\sum_{e\in P}\log s_e=\log\big((2M_P-R_P)/T_P\big)$, and stacking a sufficient set
    of paths yields $A\mathbf{x}=\mathbf{b}$ with $x_e=\log s_e$ and $A$ the
    path-by-link incidence matrix. We validate this on a seven-link $2\times3$ ladder (grid)
    mesh in which only the four corner nodes (each of degree two) are probe terminals
    and the two degree-three nodes act as relays. Twelve corner-to-corner paths make the incidence
    matrix full column rank (rank $7$, $\mathrm{cond}(A^{\!\top}A)\approx14$), so all
    seven links---including the one between the two relays---are identifiable. Solving
    in the least-squares sense recovers every per-link success ratio, with the
    link-averaged relative error decreasing as $T^{-1/2}$
    (Figure~\ref{fig:exp-mesh-regression}). The same construction applied to the
    received ratios $r_e^{(x)}$ recovers the depolarization parameters once each link
    is probed in both propagation modes.

    \paragraph{Sample size grows with depth}
    Because the end-to-end success signal $\prod_{e\in P}s_e$ shrinks geometrically
    with path length, longer paths demand disproportionately more probes. We quantify
    this on chains of $n$ identical links ($s_e=0.05$, $d_e=0.20$) and record the sample size
    required for the per-link estimator $\hat{s}=\big((2M-R)/T\big)^{1/n}$ to reach
    $5\%$ relative error. A delta-method analysis predicts
    $T_{\mathrm{req}}(n)=[s^n(1-s^n)+((s{+}d)^n-s^n)]/(n^2 s^{2n}\epsilon^2)$, whose per-hop
    multiplier is not constant but grows ($38\times,62\times,89\times$ for the simulated points),
    approaching $(s{+}d)/s^2$ for long paths ($=100$ at $s_e=0.05$). Figure~\ref{fig:exp-topology-scaling}
    confirms this law, reproducing the simulated probe counts to within tens of percent; the
    $\approx\!59\times$ figure is the geometric-mean slope at this illustrative $s_e$, and in the
    experimental coexistence regime $s\sim10^{-3}$ (Table~\ref{tab:star-network-parameters}) the
    per-hop factor is far larger. This substantiates
    the observation (Section~\ref{sec:conclusion}) that network-level estimation in
    coexisting quantum networks demands sample sizes several orders of magnitude
    larger than link-level estimation, and motivates the variance-reduction techniques
    (symmetrized probes, repeated batches) discussed in \S\ref{sec:general-topology}.

%% file: sections/conclusion.tex
\section{Conclusion}\label{sec:conclusion}

This work presented a tomography framework for characterizing quantum-classical coexistence networks (QCNs) from end-to-end measurements.
We introduced a channel model for the coexisting fiber that decomposes the quantum signal evolution into photon loss, successful transmission, and direction-dependent depolarization, capturing the distinct noise signatures of co- and counter-propagating classical traffic due to Raman scattering.
From this model, we derived closed-form link-level estimators for all channel parameters, including the success ratio, photon loss ratio, and three depolarization ratios.
We then extended the approach to star-topology networks by formulating a system of multiplicative equations across end-node pairs, and proposed a classical-signal-direction-switching strategy to resolve all per-link depolarization parameters.

Analysis of real coexisting fiber testbed data validated the channel model at the link level---the link tomography estimates closely match Bayesian process tomography baselines across multiple fiber lengths and wavelengths---while the star-network estimators were verified for correctness and identifiability on multi-link paths emulated by composing the same testbed data (with robustness to model mismatch assessed using BPT-reconstructed channel maps), and the general-topology estimators and the sample-complexity scaling were verified in fully synthetic Monte-Carlo simulations.

Beyond the core framework, we developed two extensions: \newtext{a quantum-channel model that factorizes the coexisting fiber into a CPTP \emph{tensor product} of two maps on separate optical modes---a depolarizing-with-loss signal channel on the signal mode and a Raman-noise-injection channel on the broadband mode carrying the scattered photons---whose two-mode photon-number observable accounts for the extra photons that coexistence delivers, and whose link observables reduce exactly to the basic model,} and a generalization to arbitrary topologies via tree peeling and a general-graph least-squares estimator that we verified in Monte-Carlo simulation under a correctly-specified model.

Our results also revealed that network-level estimation in QCNs demands sample sizes several orders of magnitude larger than link-level estimation, because the signal of interest---the product of per-link success ratios---is extremely small in the coexistence regime. We derive a closed-form law for this cost, $T_{\mathrm{req}}(n)\propto[s^n(1-s^n)+((s+d)^n-s^n)]/(n^2 s^{2n})$, whose per-hop multiplier grows with depth toward $(s+d)/s^2$ (Figure~\ref{fig:exp-topology-scaling}); the often-quoted $\approx\!59\times$ per added link is the geometric-mean slope at the illustrative value $s_e=0.05$, and is far larger in the experimentally measured regime $s\sim10^{-3}$.

\noindent\textbf{Deployment guidelines.}
Three of our findings translate directly into operator decisions.
(1)~\emph{Prefer co-propagation routing where topology allows}: counter-propagating classical traffic inflates the depolarization probability by a factor of $1.8\times$--$3.2\times$ over co-propagation in our data (Table~\ref{tab:star-network-parameters}; largest at short reach), so co-propagating the quantum and classical channels on a shared link yields measurably higher fidelity.
(2)~\emph{Keep monitored path depth shallow}: because the end-to-end success signal is the product $\prod_{e\in P}s_e$, the probe budget to reach a fixed relative error grows super-exponentially with the number of hops (Figure~\ref{fig:exp-topology-scaling}), so where direct access is available one should characterize links individually and confine purely end-to-end tomography to paths of depth $\lesssim 2$--$3$.
(3)~\emph{Avoid quantum wavelengths inside the classical crosstalk band}: the anomalous depolarization feature near the $1542$~nm classical-laser line (Figure~\ref{fig:link-tomography-depolarization-probability-comparison}) shows that guard-band placement of the quantum channel away from the classical carrier is essential.

\noindent\textbf{Future Directions.}
Several promising directions remain open.
We establish exact identifiability of the per-link parameters---for the star via the gauge argument of Lemma~\ref{lem:star-identifiability} and for general \newtext{network} topologies via the incidence-rank condition of Proposition~\ref{prop:identifiability}---and give a sample-complexity scaling law for the success-ratio estimator; sharp finite-sample (e.g.\ minimax) bounds and optimal probe-set design for general topologies remain open.
Our current model also focuses on the depolarization noise structure induced by Raman scattering; other noise sources---such as direct crosstalk between the classical and quantum signals, signal conversion at the transmitter and receiver, or nonlinear effects such as four-wave mixing---introduce additional noise components in certain scenarios and could be incorporated into a more comprehensive channel model.
Moreover, the sample-size requirements for network-level estimation could potentially be reduced through more efficient probe allocation strategies or regularized estimators that exploit prior knowledge of the parameter regime.
Finally, the single highest-value next experiment is a direct physical two-link test---splicing two calibrated spools and measuring the end-to-end $(R/T, M)$ statistics against the composition law of~\eqref{eq:network-tomography-measurement}---which would close the gap between the composed-link emulation used here and a fully physical multi-link validation.

%% file: bibliography.bib
@article{thomas2024quantum,
  title     = {Quantum teleportation coexisting with classical communications in optical fiber},
  author    = {Thomas, Jordan M and Yeh, Fei I and Chen, Jim Hao and Mambretti, Joe J and Kohlert, Scott J and Kanter, Gregory S and Kumar, Prem},
  journal   = {Optica},
  volume    = {11},
  number    = {12},
  pages     = {1700--1707},
  year      = {2024},
  publisher = {Optica Publishing Group}
}

@article{chapman2023coexistent,
  title     = {Coexistent quantum channel characterization using spectrally resolved Bayesian quantum process tomography},
  author    = {Chapman, Joseph C and Lukens, Joseph M and Alshowkan, Muneer and Rao, Nageswara and Kirby, Brian T and Peters, Nicholas A},
  journal   = {Physical Review Applied},
  volume    = {19},
  number    = {4},
  pages     = {044026},
  year      = {2023},
  publisher = {APS}
}

@article{townsend1997simultaneous,
  title     = {Simultaneous quantum cryptographic key distribution and conventional data transmission over installed fibre using wavelength-division multiplexing},
  author    = {Townsend, Paul D},
  journal   = {Electronics Letters},
  volume    = {33},
  number    = {3},
  pages     = {188--190},
  year      = {1997},
  publisher = {IET}
}

@article{rubeling2024quantum,
  title     = {Quantum and coherent signal transmission on a single-frequency channel via the electro-optic serrodyne technique},
  author    = {R{\"u}beling, Philip and Heine, Jan and Johanning, Robert and Kues, Michael},
  journal   = {Science Advances},
  volume    = {10},
  number    = {30},
  pages     = {eadn8907},
  year      = {2024},
  publisher = {American Association for the Advancement of Science}
}

@article{wang2024time,
  title     = {Time-interleaved C-band Co-propagation of quantum and classical channels},
  author    = {Wang, Jing and Rollick, Brian J and Jia, Zhensheng and Huberman, Bernardo A},
  journal   = {Journal of Lightwave Technology},
  year      = {2024},
  publisher = {IEEE}
}

@article{wang2017long,
  title     = {Long-distance copropagation of quantum key distribution and terabit classical optical data channels},
  author    = {Wang, Liu-Jun and Zou, Kai-Heng and Sun, Wei and Mao, Yingqiu and Zhu, Yi-Xiao and Yin, Hua-Lei and Chen, Qing and Zhao, Yong and Zhang, Fan and Chen, Teng-Yun and others},
  journal   = {Physical Review A},
  volume    = {95},
  number    = {1},
  pages     = {012301},
  year      = {2017},
  publisher = {APS}
}

@article{chapman2023two,
  title     = {Two-mode squeezing over deployed fiber coexisting with conventional communications},
  author    = {Chapman, Joseph C and Miloshevsky, Alexander and Lu, Hsuan-Hao and Rao, Nageswara and Alshowkan, Muneer and Peters, Nicholas A},
  journal   = {Optics Express},
  volume    = {31},
  number    = {16},
  pages     = {26254--26275},
  year      = {2023},
  publisher = {Optica Publishing Group}
}

@article{patel2012coexistence,
  title     = {Coexistence of high-bit-rate quantum key distribution and data on optical fiber},
  author    = {Patel, KA and Dynes, JF and Choi, I and Sharpe, AW and Dixon, AR and Yuan, ZL and Penty, RV and Shields, AJ},
  journal   = {Physical Review X},
  volume    = {2},
  number    = {4},
  pages     = {041010},
  year      = {2012},
  publisher = {APS}
}

@inproceedings{wang2023field,
  title        = {Field trial of a dynamically switched quantum network supporting co-existence of entanglement, prepare-and-measure QKD and classical channels},
  author       = {Wang, Rui and Yang, R and Clark, MJ and Oliveira, Romerson D and Bahrani, Sima and Perani{\'c}, M and Lon{\v{c}}ari{\'c}, Martin and Stip{\v{c}}evi{\'c}, Mario and Rarity, J and Joshi, Siddarth Koduru and others},
  booktitle    = {IET Conference Proceedings CP839},
  volume       = {2023},
  number       = {34},
  pages        = {1682--1685},
  year         = {2023},
  organization = {IET}
}

@article{antesberger2024distribution,
  title     = {Distribution of telecom entangled photons through a 7.7 km antiresonant hollow-core fiber},
  author    = {Antesberger, Michael and Richter, Carla MD and Poletti, Francesco and Slav{\'\i}k, Radan and Petropoulos, Periklis and H{\"u}bel, Hannes and Trenti, Alessandro and Walther, Philip and Rozema, Lee A},
  journal   = {Optica Quantum},
  volume    = {2},
  number    = {3},
  pages     = {173--180},
  year      = {2024},
  publisher = {Optica Publishing Group}
}

@book{he2021network,
  title     = {Network tomography: identifiability, measurement design, and network state inference},
  author    = {He, Ting and Ma, Liang and Swami, Ananthram and Towsley, Don},
  year      = {2021},
  publisher = {Cambridge University Press}
}

@article{coates2002internet,
  title   = {Internet tomography},
  author  = {Coates, Mark and Hero III, Alfred O and Nowak, Robert and Yu, Bin},
  journal = {IEEE Signal Processing Magazine},
  volume  = {19},
  number  = {3},
  pages   = {47--65},
  year    = {2002}
}

@article{caceres1999multicast,
  title   = {Multicast-based inference of network-internal loss characteristics},
  author  = {C{\'a}ceres, Ram{\'o}n and Duffield, Nick G and Horowitz, Joseph and Towsley, Don},
  journal = {IEEE Transactions on Information Theory},
  volume  = {45},
  number  = {7},
  pages   = {2462--2480},
  year    = {1999}
}

@article{lopresti2002multicast,
  title   = {Multicast-based inference of network-internal delay distributions},
  author  = {Lo Presti, Francesco and Duffield, Nick G and Horowitz, Joseph and Towsley, Don},
  journal = {IEEE/ACM Transactions on Networking},
  volume  = {10},
  number  = {6},
  pages   = {761--775},
  year    = {2002}
}

@article{tsang2003network,
  title   = {Network delay tomography},
  author  = {Tsang, Yolanda and Coates, Mark and Nowak, Robert D},
  journal = {IEEE Transactions on Signal Processing},
  volume  = {51},
  number  = {8},
  pages   = {2125--2136},
  year    = {2003}
}

@inproceedings{deandrade2022quantum,
  title        = {Quantum network tomography with multi-party state distribution},
  author       = {De Andrade, Matheus Guedes and Diaz, Jerry and Navas, Jordan and Guha, Saikat and Monta{\~n}o, Ines and Smith, Brian and Raymer, Michael and Towsley, Don},
  booktitle    = {2022 IEEE International Conference on Quantum Computing and Engineering (QCE)},
  pages        = {400--409},
  year         = {2022},
  organization = {IEEE}
}

@inproceedings{deandrade2023characterization,
  title        = {On the characterization of quantum flip stars with quantum network tomography},
  author       = {De Andrade, Matheus Guedes and Navas, Jordan and Monta{\~n}o, Ines and Towsley, Don},
  booktitle    = {2023 IEEE International Conference on Quantum Computing and Engineering (QCE)},
  volume       = {1},
  pages        = {1260--1270},
  year         = {2023},
  organization = {IEEE}
}

@article{deandrade2024quantum,
  title   = {Quantum network tomography},
  author  = {De Andrade, Matheus Guedes and Navas, Jordan and Guha, Saikat and Monta{\~n}o, Ines and Raymer, Michael and Smith, Brian and Towsley, Don},
  journal = {IEEE Network},
  year    = {2024}
}

@inproceedings{wang2025optimal,
  title     = {Online Optimal Probe Allocation for Quantum Network Tomography},
  author    = {Wang, Xuchuang and Chen, Yu-Zhen Janice and Guedes de Andrade, Matheus and Hajiesmaili, Mohammad and Lui, John C.S. and He, Ting and Towsley, Don},
  booktitle = {International Conference on Quantum Communications, Networking, and Computing},
  year      = {2026}
}

@article{azuma2023quantum,
  title     = {Quantum repeaters: From quantum networks to the quantum internet},
  author    = {Azuma, Koji and Economou, Sophia E and Elkouss, David and Hilaire, Paul and Jiang, Liang and Lo, Hoi-Kwong and Tzitrin, Ilan},
  journal   = {Reviews of Modern Physics},
  volume    = {95},
  number    = {4},
  pages     = {045006},
  year      = {2023},
  publisher = {APS}
}

@article{bennett2014quantum,
  title     = {Quantum cryptography: Public key distribution and coin tossing},
  author    = {Bennett, Charles H and Brassard, Gilles},
  journal   = {Theoretical computer science},
  volume    = {560},
  pages     = {7--11},
  year      = {2014},
  publisher = {Elsevier}
}

@article{cacciapuoti2019quantum,
  title     = {Quantum internet: Networking challenges in distributed quantum computing},
  author    = {Cacciapuoti, Angela Sara and Caleffi, Marcello and Tafuri, Francesco and Cataliotti, Francesco Saverio and Gherardini, Stefano and Bianchi, Giuseppe},
  journal   = {IEEE Network},
  volume    = {34},
  number    = {1},
  pages     = {137--143},
  year      = {2019},
  publisher = {IEEE}
}

@article{guo2020distributed,
  title     = {Distributed quantum sensing in a continuous-variable entangled network},
  author    = {Guo, Xueshi and Breum, Casper R and Borregaard, Johannes and Izumi, Shuro and Larsen, Mikkel V and Gehring, Tobias and Christandl, Matthias and Neergaard-Nielsen, Jonas S and Andersen, Ulrik L},
  journal   = {Nature Physics},
  volume    = {16},
  number    = {3},
  pages     = {281--284},
  year      = {2020},
  publisher = {Nature Publishing Group UK London}
}

@inproceedings{wang2025learn,
  title     = {Learning Best Paths in Quantum Networks},
  author    = {Wang, Xuchuang and Liu, Maoli and Liu, Xutong and Li, Zhuohua and Hajiesmaili, Mohammad and Lui, John C.S. and Towsley, Don},
  booktitle = {Proceedings of the IEEE Conference on Computer Communications},
  year      = {2025}
}

@article{azuma2025networking,
  title     = {Networking quantum networks with minimum cost aggregation},
  author    = {Azuma, Koji},
  journal   = {npj Quantum Information},
  volume    = {11},
  number    = {1},
  pages     = {51},
  year      = {2025},
  publisher = {Nature Publishing Group UK London}
}

@article{jiang2007distributed,
  title     = {Distributed quantum computation based on small quantum registers},
  author    = {Jiang, Liang and Taylor, Jacob M and S{\o}rensen, Anders S and Lukin, Mikhail D},
  journal   = {Physical Review A—Atomic, Molecular, and Optical Physics},
  volume    = {76},
  number    = {6},
  pages     = {062323},
  year      = {2007},
  publisher = {APS}
}

@inproceedings{vardoyan2023quantum,
  title        = {Quantum network utility maximization},
  author       = {Vardoyan, Gayane and Wehner, Stephanie},
  booktitle    = {2023 IEEE International Conference on Quantum Computing and Engineering (QCE)},
  volume       = {1},
  pages        = {1238--1248},
  year         = {2023},
  organization = {IEEE}
}

@article{pant2019routing,
  title     = {Routing entanglement in the quantum internet},
  author    = {Pant, Mihir and Krovi, Hari and Towsley, Don and Tassiulas, Leandros and Jiang, Liang and Basu, Prithwish and Englund, Dirk and Guha, Saikat},
  journal   = {npj Quantum Information},
  volume    = {5},
  number    = {1},
  pages     = {25},
  year      = {2019},
  publisher = {Nature Publishing Group UK London}
}

@article{1253508,
  author   = {Mandelbaum, I. and Bolshtyansky, M.},
  journal  = {IEEE Photonics Technology Letters},
  title    = {Raman amplifier model in single-mode optical fiber},
  year     = {2003},
  volume   = {15},
  number   = {12},
  pages    = {1704-1706},
  doi      = {10.1109/LPT.2003.819760}
}

@article{Peters_2009,
  doi       = {10.1088/1367-2630/11/4/045012},
  url       = {https://doi.org/10.1088/1367-2630/11/4/045012},
  year      = {2009},
  month     = {apr},
  publisher = {IOP Publishing},
  volume    = {11},
  number    = {4},
  pages     = {045012},
  author    = {Peters, N A and Toliver, P and Chapuran, T E and Runser, R J and McNown, S R and Peterson, C G and Rosenberg, D and Dallmann, N and Hughes, R J and McCabe, K P and Nordholt, J E and Tyagi, K T},
  title     = {Dense wavelength multiplexing of 1550 nm QKD with strong classical channels in reconfigurable networking environments},
  journal   = {New Journal of Physics}
}

@article{wang2025quantum,
  title={Quantum network tomography for general topology with spam errors},
  author={Wang, Xuchuang and De Andrade, Matheus Guedes and Avis, Guus and Chen, Yu-Zhen Janice and Hajiesmaili, Mohammad and Towsley, Don},
  journal={arXiv preprint arXiv:2511.01074},
  year={2025}
}

@article{chuang1997prescription,
  title     = {Prescription for experimental determination of the dynamics of a quantum black box},
  author    = {Chuang, Isaac L. and Nielsen, M. A.},
  journal   = {Journal of Modern Optics},
  volume    = {44},
  number    = {11--12},
  pages     = {2455--2467},
  year      = {1997},
  publisher = {Taylor \& Francis}
}

@article{mohseni2008quantum,
  title     = {Quantum-process tomography: Resource analysis of different strategies},
  author    = {Mohseni, M. and Rezakhani, A. T. and Lidar, D. A.},
  journal   = {Physical Review A},
  volume    = {77},
  number    = {3},
  pages     = {032322},
  year      = {2008},
  publisher = {APS}
}

@article{shabani2011efficient,
  title     = {Efficient measurement of quantum dynamics via compressive sensing},
  author    = {Shabani, A. and Kosut, R. L. and Mohseni, M. and Rabitz, H. and Broome, M. A. and Almeida, M. P. and Fedrizzi, A. and White, A. G.},
  journal   = {Physical Review Letters},
  volume    = {106},
  number    = {10},
  pages     = {100401},
  year      = {2011},
  publisher = {APS}
}

@article{blumekohout2017demonstration,
  title     = {Demonstration of qubit operations below a rigorous fault tolerance threshold with gate set tomography},
  author    = {Blume-Kohout, Robin and Gamble, John King and Nielsen, Erik and Rudinger, Kenneth and Mizrahi, Jonathan and Fortier, Kevin and Maunz, Peter},
  journal   = {Nature Communications},
  volume    = {8},
  pages     = {14485},
  year      = {2017},
  publisher = {Nature Publishing Group}
}

@article{drummond2001quantum,
  title     = {Quantum noise in optical fibers. I. Stochastic equations},
  author    = {Drummond, P. D. and Corney, J. F.},
  journal   = {Journal of the Optical Society of America B},
  volume    = {18},
  number    = {2},
  pages     = {139--152},
  year      = {2001},
  publisher = {Optica Publishing Group}
}

@book{nielsen2010quantum,
  title     = {Quantum Computation and Quantum Information},
  author    = {Nielsen, Michael A. and Chuang, Isaac L.},
  edition   = {10th Anniversary},
  year      = {2010},
  publisher = {Cambridge University Press},
  address   = {Cambridge, UK}
}

@incollection{lawrence2007statistical,
  title     = {Statistical inverse problems in active network tomography},
  author    = {Lawrence, Earl and Michailidis, George and Nair, Vijayan N},
  booktitle = {Complex Datasets and Inverse Problems: Tomography, Networks and Beyond},
  series    = {IMS Lecture Notes--Monograph Series},
  volume    = {54},
  pages     = {24--44},
  year      = {2007},
  publisher = {Institute of Mathematical Statistics},
  doi       = {10.1214/074921707000000049}
}

@inproceedings{he2015fisher,
  title     = {Fisher information-based experiment design for network tomography},
  author    = {He, Ting and Liu, Chang and Swami, Ananthram and Towsley, Don and Salonidis, Theodoros and Bejan, Andrei Iu and Yu, Paul},
  booktitle = {Proceedings of the 2015 ACM SIGMETRICS International Conference on Measurement and Modeling of Computer Systems},
  pages     = {389--402},
  year      = {2015},
  doi       = {10.1145/2745844.2745862}
}

@article{kveton2022optimal,
  title   = {Optimal probing with statistical guarantees for network monitoring at scale},
  author  = {Kveton, Branislav and Amjad, Muhammad Jehangir and Diot, Christophe and Konomis, Dimitris and Soule, Augustin and Yang, Xiaolong},
  journal = {Computer Communications},
  volume  = {192},
  pages   = {119--131},
  year    = {2022},
  doi     = {10.1016/j.comcom.2022.05.023}
}
